\documentclass[a4paper, 11pt]{article}

\usepackage{graphicx}
\usepackage{bm}
\usepackage{titlesec}
\usepackage[ruled, vlined]{algorithm2e}

\usepackage[shortlabels]{enumitem}
\usepackage{makeidx}
\makeindex

\usepackage[]{appendix}
\usepackage{enumitem}
\usepackage[]{algorithm2e}

\usepackage[nottoc]{tocbibind}

\usepackage[margin=1in]{geometry}

\usepackage{multirow}
\usepackage{amsmath, amssymb, amsthm, float}
\usepackage{mathtools}
\usepackage{caption}
\usepackage{color}
\usepackage{datetime} % British format dates
\usepackage{hyperref}
\usepackage[]{algorithm2e}

\usepackage{setspace}
\usepackage{mathtools}

\theoremstyle{definition}
\newtheorem{theorem}{Theorem}[section]
\newtheorem{lemma}[theorem]{Lemma}
\newtheorem{proposition}[theorem]{Proposition}
\newtheorem{corollary}[theorem]{Corollary}
\newtheorem{definition}[theorem]{Definition}

\theoremstyle{remark}
\newtheorem{remark}{Remark}

\numberwithin{equation}{section}
\numberwithin{figure}{section}

\newcommand{\bernchange}[1]{\textcolor{black}{#1}}
\newcommand{\tcr}[1]{\textcolor{black}{#1}}

\def\spacingset#1{\renewcommand{\baselinestretch}%
{#1}\small\normalsize} \spacingset{1}

\usepackage{amsfonts}

\usepackage{dsfont}
\usepackage{natbib}

\begin{document}

 \title{\bf Modelling Time-Varying First and Second-Order Structure of Time Series via Wavelets and Differencing}
   
  \author{Euan T. McGonigle \thanks{E.T. McGonigle gratefully acknowledges financial support from EPSRC and Numerical Algorithms Group Ltd. via The Smith Institute i-CASE award No. EP/R511997/1. \newline 
  Corresponding author: {\tt euan.mcgonigle@bristol.ac.uk} }\hspace{.2cm}\\
    STOR-i Centre for Doctoral Training, Lancaster University,\\
    School of Mathematics, University of Bristol, \\
    Rebecca Killick \\
    Department of Mathematics and Statistics, Lancaster University,\\
    and\\
    Matthew A. Nunes\\
    Department of Mathematics and Statistics, University of Bath}
  \maketitle

\bigskip
\begin{abstract}
Most time series observed in practice exhibit time-varying trend (first-order) and autocovariance (second-order) behaviour. Differencing is a commonly-used technique to remove the trend in such series, in order to estimate the time-varying second-order structure (of the differenced series). However, often we require inference on the second-order behaviour of the original series, for example, when performing trend estimation. In this article, we propose a method, using differencing, to jointly estimate the time-varying trend and second-order structure of a nonstationary time series, within the locally stationary wavelet modelling framework. We develop a wavelet-based estimator of the second-order structure of the original time series based on the differenced estimate, and show how this can be incorporated into the estimation of the trend of the time series. We perform a simulation study to investigate the performance of the methodology, and demonstrate the utility of the method by analysing data examples from environmental and biomedical science.
\end{abstract}

\noindent%
{\it Keywords:} Differencing, locally stationary time series, trend estimation, wavelet thresholding, wavelet spectrum.

\spacingset{1.5}

\section{Introduction}

Time series data can often possess complex and dynamic characteristics. Most commonly-encountered time series in practice are nonstationary - the mean and autocovariance of the series vary over time. Modelling how these properties change over time is crucial for making inference on the data. Nonstationary time series arise in a variety of applications, for example in environmental sciences \citep{hu2019two}, climatology \citep{das2020predictive}, and financial time series \citep{roueff2019time}. In this article, we consider a time series model of the form
\begin{equation}\label{ts-model}
X_{t,T} = \mu \left( \frac{t}{T} \right) + \epsilon_{t,T}, \quad 0 \leq t < T,
\end{equation}
where $\mu: [0,1] \rightarrow \mathbb{R}$ is a non-parametric deterministic trend function, and $\epsilon_{t,T}$ is a locally stationary wavelet (LSW) process with $\mathbb{E} (\epsilon_{t,T}) = 0$.  This model accounts for nonstationarity in \emph{both} the first and second-order structure of the time series: the time-varying mean function is encapsulated in the trend term, while the time-varying second-order behaviour is described by the LSW process term. A thorough explanation of Model (\ref{ts-model}), including necessary background on LSW processes, is given in Section \ref{model-def}.

First and second-order estimation of a time series are most commonly \tcr{considered separately}, rather than \tcr{jointly} within the one common methodological framework. However, there is an emerging literature where the mean function is estimated as well as parameters responsible for second-order nonstationarity. For example \cite{Tunyavetchaki} consider time varying AR(p)-processes where the \tcr{mean function and AR coefficients are estimated jointly within the same procedure}. Other methods, such as \cite{dahlhaus2001locally} which focuses on the semi-parametric setting, consider results where the mean function is time-varying and/or estimated. \cite{khismatullina2019multiscale} test for increases and decreases in trend in the presence of stationary time series errors, while \cite{dette2020prediction} consider the problem of prediction in locally stationary time series. 
%\cite{mcgonigle2020trend} extend the LSW model of \cite{nason2000wavelet} to incorporate smooth continuous trend functions by considering the trend and second-order components separately. {\bf **** But you say that you model these in one framework in the other paper!!!!****}  I've removed this sentence and we can put it in during the revision when the other paper may be further along in the cycle - given we aren't putting on ArXiV.

The problem of performing inference on the mean function in the presence of nonstationary second-order structure is a highly challenging one. In \cite{von2000non}, for example, the authors describe a method using wavelet thresholding, however the threshold used is data-driven and depends on the nonstationary second-order structure, which is ultimately unknown. \cite{vogt2012nonparametric} employs kernel-based methods to estimate a smooth non-parametric trend function in the presence of locally stationary errors. \cite{dette2019detecting} test for relevant changes in the mean of nonstationary processes. Similarly, there is less attention in the literature on nonstationary second-order estimation in the presence of a \tcr{non-zero} mean function. 

In order to estimate a nonstationary second-order structure a zero-mean process is often assumed. One of the most well-known methods for removing the trend in a time series is differencing: see for example \cite{chan1977note} and \cite{shumway2010time}. The time series $\{X_t\}$ can be, for example, first-differenced to obtain a new time series, $\{\nabla X_t = X_{t} - X_{t-1} \}$. Differencing aims to remove a trend without the need to estimate any parameters (whose estimation often includes an assumption of stationary errors). In the time series literature, it is often the case that inference is made on the properties of the differenced time series.  In contrast, this article proposes a method to jointly estimate the time-varying trend and second-order structure of the original time series, by employing the commonly-used strategy of differencing a time series in order to remove the trend. 

Our approach to this problem can be summarised as follows. By differencing the time series to remove the trend, we can estimate the appropriate second-order quantities of interest of the locally stationary wavelet part of Model (\ref{ts-model}). This is achieved by considering the effect of differencing on the second-order properties of the series in order to develop an appropriate estimation procedure. Expanding upon the rigorous theory developed in \cite{nason2000wavelet}, we obtain results on the consistent estimation of the second-order structure using our modified estimation strategy for the original time series. Using this estimate, we discuss a wavelet thresholding technique to estimate the trend function $\mu (t/T)$ in a principled manner, by taking into account the time-varying second-order behaviour. Our methodology thus enables the joint estimation of the first and second-order structure of nonstationary time series.  The applications demonstrate the added utility that estimating the second order structure of the original time series brings.

The rest of the article is organised as follows. In Section 2 we introduce the time series model which we focus on in this article, describe key assumptions, and discuss necessary background to LSW processes. In Section 3, we analyse the effect that differencing has on the spectral structure of a time series, and explain the intuition behind our methodology. Furthermore, we describe the methodology for consistent estimation of the second-order structure in the presence of trend, and in Section 4 we use this estimate in order to estimate the trend of the series. Simulation studies assessing the method's performance are given in Section 5. In Section 6, we apply our framework to a data example, demonstrating the utility of the method, while concluding remarks are given in Section 7. All proofs and additional simulation and data application results are contained within the supplementary material.

\section{Wavelets and Model Formulation}\label{model-def}

In this section we introduce the modelling paradigm that we will use, as well as explaining the necessary background concepts. 

\bernchange{In the LSW framework, wavelets act as building blocks, analogous to the Fourier exponentials in the classical Cram\'{e}r representation for stationary processes. Briefly, wavelets are oscillatory basis functions which provide efficient (sparse) multiscale representations of signals. Wavelets are useful in estimating time-varying quantities, especially nonstationary characteristics, due to these attractive properties. For an overview of wavelet techniques, see for example \cite{nason2010wavelet} or \cite{vidakovic2009statistical}.}

\bernchange{For example, for a function $f \in L^{2}(\mathbb{R})$, we have the expression $f(x) = \sum_{k \in \mathbb{Z}} c_{J} \phi_{J,k} (x) + \sum_{j \leq J} \sum_{k \in \mathbb{Z}} d_{j,k} \psi_{j,k} (x)$, where the wavelet $\psi_{j,k} (x) = 2^{-j/2} \psi(2^{-j} x - k)$ is a dilated and translated version of a (mother) wavelet $\psi (x)$ and similarly for the father wavelet $\phi(x)$. The coefficients $d_{j,k}$ at location $k$ and scale $j$ represent the oscillatory behaviour of the signal $f$ at a particular frequency, whereas the coefficients $c_{j,k}$ give information about the mean behaviour of the signal at different scales $j$. }

\tcr{Next, we define the discrete wavelets that make up the building blocks of the LSW model. Let $\psi$ be a compactly supported wavelet, for example any within the Daubechies family \cite{daubechies1992ten}. Denote by $\{h_k, g_k\}$ the low- /high pass filter pair associated with $\psi$. Let $N_h$ be the number of non-zero values of $\{h_k \}$, and define $L_j = (2^{-j} - 1)(N_h-1)+1$. The discrete wavelets at a given scale $j \in \mathbb{Z}^{-}$, as discussed in \cite{nason2000wavelet}, are defined as the vectors $\psi_j = (\psi_{j,0} , \ldots , \psi_{j, L_j-1})$ , where $\psi_{-1,n}= \sum_k g_{n-2k} \delta_{0k} = g_n$, and $\psi_{j-1,n} = \sum_k h_{n-2k} \psi_{j,k}$ for $n=0, \ldots ,L_j-1$, where $\delta_{0k}$ is the Kronecker delta. The discrete father wavelet is defined similarly using the associated low-pass filter $\{ h_k \}$.}

\bernchange{The simplest example of a wavelet basis is the Haar wavelet, which is given by
\begin{equation*}
\psi^{H}_{j,k} = 2^{-j/2} \mathbb{I} \left(0 \leq k \leq 2^{j-1} - 1 \right) - 2^{j/2} \mathbb{I} \left(2^{j-1}  \leq k \leq 2^{j} - 1 \right),
\end{equation*}
where $j = \{-1, -2, -3, \ldots \}$ and $k \in \mathbb{Z}$. }

\subsection{Model Definition}

Below, we define the \emph{Trend} Locally Stationary Wavelet (T-LSW) model, which is composed of a deterministic Lipschitz continuous trend component and a locally stationary wavelet component. Our T-LSW model, developing on the theory of locally stationary wavelet processes of \cite{nason2000wavelet}, allows for simultaneous inference on the time-varying mean and autocovariance of a time series. 

\begin{definition}
A \emph{trend locally stationary wavelet (T-LSW) process} $\{X_{t,T} \}$, $t = 0, \ldots , T-1$, and $T = 2^J \geq 1$ for $J \in \mathbb{N}$ is a doubly-indexed stochastic process with the following representation in the mean-square sense:
\begin{equation}\label{lsw_rep}
X_{t,T} = \mu \left( \frac{t}{T} \right) + \sum_{j = -\infty}^{-1} \sum_{k \in \mathbb{Z}} w_{j,k;T} \psi_{j,k-t} \xi_{j,k} ,
\end{equation}
where $\{\xi_{j,k} \}$ is a random, uncorrelated, zero-mean orthonormal increment sequence, $\{\psi_{j, k-t} \}_{j,k}$ is a set of discrete non-decimated wavelets, and $\{w_{j,k;T} \}$ is a set of amplitudes. The quantities in representation (\ref{lsw_rep}) possess the following properties:
\begin{enumerate}
\item The function $\mu: [0,1] \in \mathbb{R} $ is Lipschitz continuous with constant $K>0$, i.e
\begin{equation*}
\left| \mu \left( \frac{t}{T} \right) -  \mu \left( \frac{s}{T} \right) \right| \leq \frac{K}{T}, \ \forall \ s,t \in [0,T].
\end{equation*}
\item There exists, for each $j \leq -1$, a Lipschitz continuous function $W_j(z)$ for $z \in (0,1)$ which satisfies the following properties:
\begin{equation*}
\sum_{j = -\infty}^{-1} | W_j (z) |^2 < \infty \text{ uniformly in } z \in (0,1);
\end{equation*}
the Lipschitz constants $L_j$ are uniformly bounded in $j$ and $\sum_{j=-\infty}^{-1} 2^{-j} L_j < \infty$. There exists a sequence of constants $C_j$ such that for each $T$
\begin{equation*}
\sup_k \left| w_{j,k;T} - W_j \left( \frac{k}{T} \right) \right| \leq \frac{C_j}{T},
\end{equation*}
where for each $j \leq-1$ the supremum is over $k = 0, \ldots , T-1$, and where the sequence $\{C_j\}$ satisfies $\sum_{j=-\infty}^{-1} C_j < \infty$.
\end{enumerate}
\end{definition}

The model imposes the same assumptions on the LSW component as in \cite{nason2000wavelet}, allowing for locally stationary second-order structure, while also permitting nonstationary first-order behaviour by incorporating a smooth mean function $\mu$. Imposing a Lipschitz continuous trend assumption is not overly restrictive, given that trend functions are generally assumed to be smooth and slowly-evolving. In particular, polynomials (restricted to the interval $[0,1]$) are Lipschitz continuous, as are sinusoids. Such an assumption is commonly made in the literature, see for example \cite{vogt2012nonparametric} and \cite{khismatullina2019multiscale}. Furthermore, in Section \ref{n-order-diff} we will discuss the case when the trend is not Lipschitz.

\subsection{Background to LSW Processes}

In the original work of \cite{nason2000wavelet}, the trend function $\mu (t/T)$ in Equation (\ref{lsw_rep}) is assumed to be everywhere zero, which forces the process mean $\mathbb{E} (X_{t,T} )$ to be equal to zero for all $t$. Consequently, within the original LSW framework it is only possible to estimate time-varying second-order structure when the time series does not exhibit a trend. In order to discuss our proposed methodology in the setting where a trend \emph{is} present, we dedicate the rest of the section to recalling a number of definitions and results from \cite{nason2000wavelet}, which we will expand upon and adapt to the setting where a trend is present.

The second-order structure of a LSW process is \tcr{uniquely determined by its spectrum}. The evolutionary wavelet spectrum (EWS) of a LSW process is defined as $S_{j} (z) := | W_{j} (z) |^{2}$ for rescaled time $z = k/T \in (0,1)$, and measures the contribution to variance at a particular rescaled time $z$ and scale $j$. Since  the $W_{j}$ are assumed to be Lipschitz continuous, the spectrum $S_{j}$ is also Lipschitz continuous, which ensures it evolves slowly over time. Alterations to the LSW model that use different assumptions on the EWS can be found in \cite{fryzlewicz2006haar} and \cite{van2008locally}, which assume respectively that the EWS is piecewise constant and of bounded total variation. 

For ease of notation we now drop the dependence on $T$ in the subscripts of the model quantities. The EWS is estimated via the empirical wavelet coefficients of the time series, given by $d_{j,k} := \langle X_{t} , \psi_{j,k-t} \rangle = \sum_{t} X_{t} \psi_{j,k-t}$. As with Fourier approaches, the raw wavelet periodogram $I_{k}^j := | d_{j,k} |^2$ is a biased, inconsistent estimator of the EWS (\cite{nason2000wavelet}, Proposition 4):
\begin{eqnarray}\label{usual-exp}
\mathbb{E} \left( I_{k}^j \right) &=& \sum_{l} A_{jl} S_l (k/T) + \mathcal{O} (T^{-1} ),\\
\label{usual-var}
\text{Var}( I_{k}^j ) &=& 2 \left( \sum_l A_{jl} S_l (k/T) \right)^2 + \mathcal{O} (2^{-j} T^{-1}),
\end{eqnarray}
where the operator $A = (A_{jl})_{j, l <0}$ is given by $A_{jl} : = \langle \Psi_j , \Psi_l \rangle = \sum_\tau \Psi_j (\tau) \Psi_l (\tau)$, and the autocorrelation wavelets are defined by $\Psi_j (\tau) := \sum_{k \in \mathbb{Z}} \psi_{j,k} \psi_{j,k- \tau}  , \ j < 0, \tau \in \mathbb{Z}$. Hence, for the vector of periodograms $\textbf{I} (z) := \{ I_{\lfloor zT \rfloor}^l \}_{l = -1, \ldots , -J}$, and the vector of corrected periodograms $\textbf{L} (z) := \{L_{\lfloor zT \rfloor }^j \}_{j = -1, \ldots , -J}$ with $\textbf{L} (z) = A_J^{-1} \textbf{I} (z)$, where the $J$-dimensional matrix $A_J := (A_{jl})_{j, l=-1, \ldots , -J}$,
\begin{equation}\label{exp-cor}
\mathbb{E} \left( \textbf{L} (z) \right) = \mathbb{E} \left( A_J^{-1} \textbf{I} (z) \right) = \textbf{S} (z) + \mathcal{O} (T^{-1} ) \qquad \forall \ z \in (0,1),
\end{equation}
where $\textbf{S} (z) := \{ S_j (z) \}_{j = -1, \ldots , -J }$. The raw wavelet periodogram is first smoothed and then corrected to produce an asymptotically unbiased, consistent estimator. There are several approaches to smoothing for consistency, for example via a running mean as in \cite{nason2013test} or wavelet thresholding as in \cite{nason2000wavelet}. The correction is performed by premultiplying the (smoothed) raw wavelet periodogram by $A^{-1}_{J}$, as motivated by Equation \eqref{exp-cor}. The operator $A$ is shown to have bounded inverse for all Daubechies compactly supported wavelets in \cite{cardinali2017locally}.

The local autocovariance (LACV) function for a LSW process provides information about the covariance at a rescaled location $z = k/T \in (0,1)$. The LACV, $c(z, \tau)$, of a LSW process with EWS $\{S_j (z) \}$ is defined as $c(z, \tau) = \sum_{j = -\infty}^{-1} S_j (z) \Psi_j (\tau)$, for $\tau \in \mathbb{Z}$, $z \in (0,1)$. The LACV can be thought of as a decomposition of the autocovariance of a process over scales and rescaled time locations. The LACV is estimated by plugging in the smoothed, corrected estimate for the EWS into the definition of the LACV. Using wavelet thresholding of the EWS estimator, it is shown in \cite{nason2000wavelet} that the LACV estimator is consistent. The next section addresses the consistent estimation of these second-order quantities in the presence of first-order nonstationarity.

\section{LSW Estimation in the Presence of Trend via Differencing}

In this section, we discuss methodology for estimation of the evolutionary wavelet spectrum and local autocovariance function of a trend-LSW process. In order to estimate these quantities, we employ the ubiquitous practice of differencing to remove the trend, but crucially correct for the effect of this on the spectrum.

\subsection{Using Differencing to Detrend a Time Series}\label{diff-detrend-sec}

One of the most common methods for removing the trend in a time series is differencing, see for example \cite{chan1977note} and \cite{shumway2010time}. The time series $\{X_t\}$ can be, for example, first-differenced to obtain a new time series, $\{\nabla X_t = X_{t} - X_{t-1} \}$, upon which inference is then performed. One advantage of differencing is that no parameters are estimated in the differencing operation, which is not the case for \tcr{detrending achieved using an estimator of the trend}. An $n$-th order difference is capable of removing a polynomial trend of degree $n$ from the data.

One of the consequences of differencing is that the second-order statistical properties of the time series in Model (\ref{lsw_rep}) will change, sometimes quite drastically. Therefore, it is potentially problematic to directly use the differenced process for inference on the original process. In the context of ARIMA modelling, differencing is performed in order to induce stationarity, and estimation is then performed on the differenced series. However, this reasoning does not hold within our setting: if we difference a nonstationary LSW process, the result will still be a nonstationary process. Due to the potentially complex structure of the LSW process, properties of the differenced series are not necessarily representative of the original series. However, in order to estimate the trend component in (\ref{lsw_rep}), we require an estimate of the second-order structure of the original time series, not the differenced series.

It is straightforward to produce an example where spectral behaviour can be completely altered as a consequence of detrending using, for example, first-differencing. Consider the zero-mean LSW process of length $T = 2^{10} = 1024$, defined by the spectrum $S_j(z) = 1$ for $j=-1$, $0$ otherwise. This time series is referred to as the scale $-1$ Haar moving average process, and can be written as $X_t = ( \epsilon_t - \epsilon_{t-1} )/ \sqrt{2}$, where the $\epsilon_{t}$ are independent identically distributed (IID) random variables. Taking the first-difference, we obtain $\nabla {X}_t = X_{t} - X_{t-1} =  \left( \epsilon_{t} -2\epsilon_{t-1} + \epsilon_{t-2} \right) / \sqrt{2}$.  Computing the expectation of the raw wavelet periodogram of the differenced time series, we find that $\mathbb{E} (I^{-1}_{k}) = 5$, and for $j < -1$, $\mathbb{E} (I^{j}_{k})  = 3 \times 2^{j+1}$. Therefore, for all $z \in (0,1)$, the expected value of the LSW estimator at time $z$ is $\mathbb{E} ( \textbf{L} (z) ) = \mathbb{E} ( A_{10}^{-1} \textbf{I} (z) )$.

Having differenced the time series, a problem arises since $\mathbb{E} (I^{j}_{k} ) \neq \sum_l A_{jl} S_l (k/T) + \mathcal{O} (T^{-1})$. In particular, the expectation of the EWS estimate at scale $-2$ at any time point is given by $\mathbb{E} \left(\sum_l A_{-2,l}^{-1} I^{l}_{k} \right) = -0.79 < 0$. Therefore, the expectation of our estimate of the spectrum at level $-2$ is -0.79, while in the original time series we had $S_{-2} (z) = 0$ for all $z \in (0,1)$. In Figure \ref{extraplot} left, we see a plot of the original spectrum, while on the right, we see the expectation of the corrected periodogram estimate, showing a clear discrepancy.

\begin{figure}[]
\centering
\includegraphics[width = 0.87\textwidth]{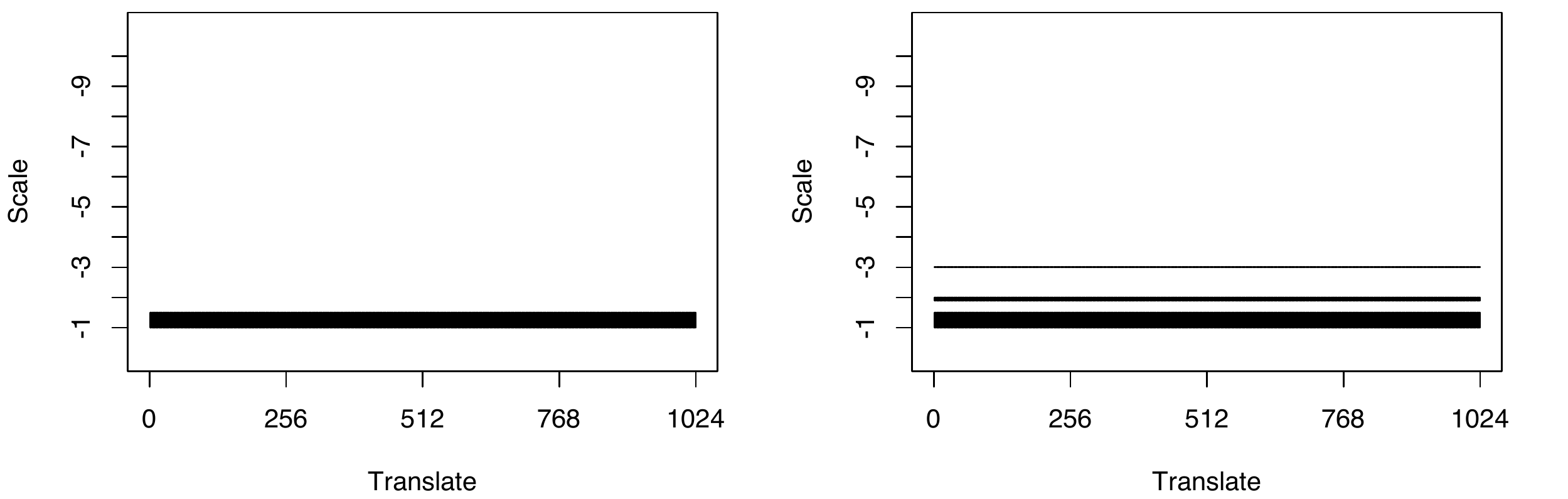}
\caption{Left: original spectrum of the scale $-1$ Haar moving average process. Right: expectation of spectral estimate of differenced time series.}
\label{extraplot}
\end{figure}

Differencing can also cause \tcr{vastly alter the sparsity structure of the spectrum of the error process}. The white noise process $Y_{t} = \epsilon_{t}$, where the $\epsilon_{t}$ are IID random variables, can be represented as an LSW process using Haar wavelets and a spectrum defined by $S_{j} (z) = 2^{j} \text{ for all } z \in (0,1)$. Differencing this time series gives $\nabla {X}_{t }= X_{t} - X_{t-1} = \epsilon_{t} - \epsilon_{t-1}$, which is a Haar LSW process with spectrum $S_{j} (z) = 2$ for $j = -1$, $S_{j} (z) = 0$ otherwise. This induces autocorrelation in the time series at lag 1, which is similar to what is observed when over-differencing in classical stationary time series. Therefore, it can be seen that we must take into account the differencing of the time series if we wish to say something meaningful about the original series.

\subsection{Intuition Behind Estimation Procedure}

Denote the LSW component of Model \eqref{lsw_rep} by $\epsilon_{t}$. Given that the trend component of the Model (\ref{lsw_rep}) is Lipschitz continuous, first-differencing the time series yields
\begin{align*}
\nabla {X}_{t} &= X_{t} - X_{t-1} = \mu \left( \frac{t}{T} \right) + \epsilon_{t} - \mu \left(  \frac{t-1}{T} \right) - \epsilon_{t-1} \\[1ex]
&= \epsilon_{t} - \epsilon_{t-1} + \mathcal{O} (T^{-1}).
\end{align*}
Hence, differencing the original series results in a differenced locally stationary wavelet process that is asymptotically zero-mean. We wish to estimate the evolutionary wavelet spectrum of the original time series $\{X_{t} \}$ using the differenced series $\{\nabla {X}_{t} \}$. Proceeding by using the standard estimation procedure of \cite{nason2000wavelet} by taking the squared wavelet coefficients and premultiplying by the inverse of the $A$ matrix, as in Equation (\ref{exp-cor}), is not appropriate, as we have seen in Section \ref{diff-detrend-sec}. Denote the empirical non-decimated wavelet coefficients of the first-differenced series by
\begin{equation*}
\tilde{d}_{j,k} := \sum_{t} \nabla {X}_{t} \psi_{j,k-t}.
\end{equation*}
We can compute the expectation of the raw wavelet periodogram $\tilde{I}_{k}^{j} := | \tilde{d}_{j,k} |^{2}$, which will yield an analogous result to Equations \eqref{usual-exp} and \eqref{usual-var}. Hence, our estimation strategy will follow that of \cite{nason2000wavelet}. However, the correction of the raw wavelet periodogram to achieve unbiasedness will require a different correction matrix, to account for the fact that the raw wavelet periodogram is calculated on the differenced series. 

Note that, while wavelets themselves are differencing operators, coarser wavelet scales accumulate large bias due to the trend and large filter length. Bias in the coarse scales will also affect the bias in the finer scales, due to the necessity for correcting the estimate using a correction matrix. This can be particularly noticeable for trends that display cusps. Differencing first removes the trend immediately, ensuring that the wavelet transform does not accumulate large bias.

\begin{remark}
%It is not immediately obvious that if we first-difference an LSW process, we obtain another LSW process. 
We can write the differenced process, in the mean-square sense, as
\begin{equation*}
\epsilon_{t} - \epsilon_{t-1} = \sum_{j} \sum_{k} w_{j,k} \psi_{j,k-t} \tilde{\xi}_{j,k} ,
\end{equation*}
where $\tilde{\xi}_{j,k} = \xi_{j,k} - \xi_{j,k-1}$. This process satisfies all of the required properties of a standard LSW process except one: the increments are no longer orthonormal. Instead, we have that
\begin{equation*}
\emph{Cov} (\tilde{\xi}_{j,k} , \tilde{\xi}_{l,m}) = \begin{cases}
2  & \text{ for } j=l, k = m, \\
-1  & \text{ for } j=l \text{ and } k=m+1 \text{ or } k+1=m, \\
0 & \text{otherwise}.
\end{cases}
\end{equation*}
Therefore, there is a distinction between the observed time series, which we assume to have LSW errors, and the differenced one, which we do not. This is in contrast to existing approaches utilising differencing for second-order estimation, such as ARIMA models, which instead assume that the differenced series follows the model form rather than the original series.
\end{remark}

\subsection{Asymptotic Behaviour of the Differenced Raw Wavelet Periodogram}

As motivated by the discussion in the previous section, we estimate the spectrum by calculating the raw wavelet periodogram of the first-differenced time series. We return to the case of general $n$-th differencing in Section \ref{n-order-diff}. For the purpose of theoretical results, we assume that the $\{ \xi_{j,k} \}$ in Model \eqref{lsw_rep} are Gaussian. In practice, this assumption is not required, and in the simulation study in Section \ref{sim-study} we also investigate the case where the $\{ \xi_{j,k} \}$ are exponential random variables. Before obtaining the result of the expectation of the raw wavelet periodogram, we define the following two operators, which involve the inner product of the autocorrelation wavelets at lag $1$.
\begin{definition}\label{A1_mat_def}
Let $B$ be the backshift operator such that $B \Psi_{j} (\tau) = \Psi_{j} (\tau-1)$. Define the operator $A^{1} = (A_{jl}^{1})_{j, l <0}$ by
\begin{equation*}
A_{jl}^{1} : = \langle \Psi_j  , B \Psi_l  \rangle = \sum_\tau \Psi_j (\tau) \Psi_l (\tau -1) ,
\end{equation*}
and the operator $D^{1} = (D_{jl}^{1})_{j, l <0}$ by
\begin{align*}
D_{jl}^{1} &: =  \langle (1 - B) \Psi_j , (1-B) \Psi_l  \rangle = 2 A_{jl} - 2 A^{1}_{jl} \\[1ex]
&= 2 \sum_\tau \Psi_j (\tau) \left( \Psi_{l} (\tau) - \Psi_l (\tau -1) \right) .
\end{align*}
Denote the $J$-dimensional matrices $A_J^{1} : = (A_{jl}^{1})_{j, l=-1, \ldots , -J}$ and $D_J^{1} : = (D_{jl}^{1})_{j, l=-1, \ldots , -J}$. 
\end{definition}
\begin{proposition}\label{prop1}
The matrix $D^{1}_{J}$ is invertible.
\end{proposition}
Intuitively, it is not surprising that the quantity $D^{1}$ will appear when we calculate the expectation of the squared wavelet coefficients of the first-differenced series. Indeed, we have the following result for the asymptotic behaviour of the raw wavelet periodogram of the first-differenced time series.
\begin{proposition}\label{A1exp}
Let $\tilde{I}^{j}_{k}= | \tilde{d}_{j,k}|^{2}$ be the wavelet periodogram of the first-differenced time series. Under the assumptions of Model \eqref{lsw_rep}, 
\begin{equation*}
\mathbb{E} (\tilde{I}_{k}^{j} ) =  \sum_{l} D_{jl}^{1} S_{l} \left( \frac{k}{T} \right) + \mathcal{O} (T^{-1} ) ,
\end{equation*}
\begin{equation*}
\emph{Var} (\tilde{I}_{k}^{j} ) = 2 \left( \sum_{l} D^{1}_{jl} S_{l} \left(\frac{k}{T} \right) \right)^{2}  + \mathcal{O}(2^{-j} T^{-1}).
\end{equation*}
\end{proposition}
Therefore, for the vector of periodograms $\textbf{I} (z) := \{ I_{\lfloor zT \rfloor}^l \}_{l = -1, \ldots , -J}$, and the vector of corrected periodograms $\textbf{L} (z) := \{L_{\lfloor zT \rfloor }^j \}_{j = -1, \ldots , -J}$ with $\textbf{L} (z) = (D^{1}_{J})^{-1} \textbf{I} (z)$,
\begin{equation*}
\mathbb{E} \left( \textbf{L} (z) \right) = \mathbb{E} \left( (D^{1}_{J})^{-1} \textbf{I} (z) \right) = \textbf{S} (z) + \mathcal{O} (T^{-1} ) \qquad \forall \ z \in (0,1),
\end{equation*}
where $\textbf{S} (z) := \{ S_j (z) \}_{j = -1, \ldots , -J }$ is the EWS of the original process. Thus we can bias correct the raw wavelet periodogram of the first-differenced time series by using the inverse of $D^{1}_{J}$, instead of the inverse of $A_{J}$ from \cite{nason2000wavelet}.

\begin{remark}\label{seasonal-remark}
\tcr{Proposition \ref{A1exp} can be generalised to the lag $L$ differenced time series $\{\nabla_L X_t := X_t-X_{t-L}\}_{t=L+1}^{T}$ for any fixed integer $L$. The expectation of the raw wavelet periodogram of the lag $L$ differenced series at location $k$ is given by
\begin{equation*}
\sum_l D^L_{jl} S_l \left( \frac{k}{T} \right) + \mathcal{O} (T^{-1}),    
\end{equation*}
where $D^L_{jl} = 2 \sum_{\tau} \Psi_j (\tau) ( \Psi_l (\tau) - \Psi_l (\tau - L ) )$. Then, the matrix $D^L_J$ is invertible by the same argument used to show invertibiltiy of $D^1_J$. This enables the raw wavelet periodogram of the $L$-differenced series to be corrected using the matrix $(D^L_J)^{-1}$. This technique can be used to extend the method to allow for seasonal trends, as shown empirically in Section \ref{seasonal-extension}.}
\end{remark}

\subsection{Bounded Invertibility of Haar and Shannon Operators}

In order to achieve an \emph{asymptotically} unbiased estimator of the EWS, we require boundedness of the inverse operator used when performing the bias correction step in the estimation procedure. In the original LSW work of \cite{nason2000wavelet}, boundedness of the inverse of the operator $A$ was shown in the case of the Haar and Shannon wavelets. \tcr{The members of the Daubechies compactly supported family of wavelets are characterised by the number of vanishing moments. The Haar wavelet is equivalently the Daubechies extremal phase wavelet with one vanishing moment, while the Shannon wavelet is the limiting wavelet in the Daubechies compactly supported family \citep{chui1997wavelets} as the number of vanishing moments tends to $\infty$. As noted in \cite{von1997adaptive}, Haar and Shannon wavelets can therefore be viewed as the lower and upper ``extremes'' of the family of Daubechies compactly support wavelets}. Therefore, proving bounded invertibility in these two cases intuitively suggests that bounded invertibility should hold in the case of all Daubechies' compactly supported wavelets. It is for these two wavelets that we prove bounded invertibility here; extensions to other wavelets are left for future work.

Proposition \ref{prop1} shows that the matrix $D^{1}_{J}$ is invertible with bounded inverse. Asymptotically, however, the inverse operator is unbounded, which mirrors a result in the setting of locally stationary Fourier time series \citep[Equation 5]{roueff2011locally}. Intuitively this is expected, since the differencing operator itself is asymptotically non-invertible. We can interpret the expectation result in Proposition \ref{A1exp} as a \textcolor{black}{quantification of the effect of differencing on the spectral structure of the time series.}

We can account for this theoretical issue by proving bounded invertibility for a related, rescaled operator \textcolor{black}{$P: \ell^2 (\mathbb{N}) \to \ell^2 (\mathbb{N})$, where the entries of $P$ are given by $P_{jl} = 2^{-j/2} D^{1}_{jl} 2^{-l/2}$. The use of $P$ is due to the fact that the inverse of the operator $D^1$ is unbounded, so our approach is to work with the operator $P$, which we can show has a bounded inverse.} Showing that $P$ possesses a bounded inverse enables us to show theoretical consistency properties of the wavelet periodogram-based estimator. Hence, we show for the Haar and Shannon wavelet families that $P$ possesses a bounded inverse, which enables consistent estimation of the EWS for Model \eqref{lsw_rep}. Note that, in the practical implementation of the methodology, we still use the inverse of the $D^{1}_{J}$ matrix for correcting the raw wavelet periodogram as it is an invertible matrix (Proposition \ref{prop1}).

\begin{theorem}\label{invertibility1}
Let $\lambda_{\text{min}} (P)$ denote the smallest eigenvalue of $P$, where the entries of $P$ are given by $P_{jl} = 2^{-j/2} D^{1}_{jl} 2^{-l/2}$. Then, for the Haar and Shannon wavelet families, there exists $\delta > 0$ such that  $\lambda_{\text{min}} (P) \geq \delta$ and hence $||P^{-1}|| < \infty$. That is, $P$ is positive-definite and has a bounded inverse.
\end{theorem}

\subsection{Smoothing and Estimation Theory}\label{smooth-est-theory}

As in the original LSW model, the wavelet periodogram is not a consistent estimator and must be smoothed. Smoothing to achieve consistency can be performed, for example, via wavelet thresholding or a running mean. For brevity, we provide theoretical results for wavelet thresholding, building on known results in the literature. 

%We note that in practice, a running mean smoother is more straightforward to implement as only one bin width parameter need be chosen, and can often provide better results. 

In order to utilise the result on boundedness of the operator $P^{-1}$ in Theorem \ref{invertibility1}, we rewrite the formula for the expectation of the wavelet periodogram. We can express the expectation in terms of $P$ and a scaled version of $S_{j}$, given by $\tilde{S}_{j} = 2^{j/2} S_{j}$, by rescaling the periodogram appropriately. Concretely, consider the auxiliary process $\epsilon_{t} = \sum_{j,k}  \tilde{w}_{jk} \tilde{\psi}_{j,k-t} \xi_{lm}$, where $\tilde{w}_{jk} = 2^{j/4} w_{jk}$ and $\tilde{\psi}_{j,k-t} = 2^{-j/4} \psi_{j,k-t}$. Then, the expectation of the raw wavelet periodogram (with respect to the rescaled wavelet $\tilde{\psi}_{j,k-t}$) is given by
\begin{equation*}
\mathbb{E} (\tilde{I}_{k}^{j} ) = \sum_{l} P_{jl} \tilde{S}_{l} (k/T) + \mathcal{O} (T^{-1}),
\end{equation*}
where $\tilde{S}_{j} (k/T) = 2^{j/2} S_{j} (k/T)$ and $P_{jl} = 2^{-j/2} D^{1}_{jl} 2^{-l/2}$. 

%To achieve consistency we take a similar approach to \cite{nason2000wavelet}. For each fixed scale $j$, the rescaled periodogram $\tilde{I}^j_{k}$ of a Gaussian LSW process (which is scaled $\chi^2$-distributed) is smoothed as a function of $z = k/T$ using, for example, discrete wavelet transform (DWT) shrinkage or translation invariant (TI) denoising of \cite{coifman1995translation}. Smoothing is achieved using non-linear thresholding of the empirical wavelet coefficients $\hat{v}_{rs}^{j}$ of $\tilde{I}^j (z)$.
%
%The theoretical argument to consistently estimate $S_{j}(k/T)$ is thus as follows. First, smooth the rescaled wavelet periodogram at each scale $j$ using wavelet thresholding. Next, use the (bounded) $P$-inverse matrix to correct the smoothed, rescaled periodogram. Finally, multiply the estimate at each scale by $2^{-j/2}$, since ${S}_{j} (k/T) = 2^{-j/2} \tilde{S}_{j} (k/T)$, yielding the final estimator, $\widehat{S}(k/T)$, of $S_{j}(k/T)$. Using the hard threshold $\lambda(j,r,s,T)^{2} = \text{Var} ( \hat{v}_{rs}^{j} ) \log^2 (T)$ when smoothing the periodogram via wavelet thresholding, we can show that the smoothed, corrected estimate $\widehat{S}_j (z)$ is mean square consistent in the $L_{2}$ sense. 

Then, to achieve consistency we take a similar approach to \cite{nason2000wavelet}. For each fixed scale $j$, the rescaled periodogram $\tilde{I}^j_{k}$ of a Gaussian LSW process (which is scaled $\chi^2$-distributed) is smoothed as a function of $z = k/T$ using, for example, discrete wavelet transform (DWT) shrinkage or translation invariant (TI) denoising of \cite{coifman1995translation}. \bernchange{Using the DWT, smoothing is performed with respect to an orthonormal wavelet basis $\{  \phi'_{r_{0},s} (z) , \psi'_{rs} (z)  \}$ of $L^{2}( [0,1])$. Here, $\psi'_{rs} (z) = 2^{r/2} \psi' (2^{r} z - s )$, where $r_{0}$ is the coarsest scale analysed and $s = 0, \ldots, 2^{r} - 1$. Smoothing is achieved using non-linear thresholding of the empirical wavelet coefficients $\hat{v}_{rs}^{j}$ of $\tilde{I}^j_{k}$.}

\bernchange{Explicitly, as described in \cite{von1997adaptive}, for levels $j = -1, \ldots, -J$, the wavelet expansion of the scaled periodogram can be written as 
\begin{equation*}
\tilde{I}^{j}_{\lfloor zT \rfloor} =  \sum_{r} \sum_{s} {v}^{j}_{rs} \psi_{rs}' (z),
\end{equation*}
where the ``true'' wavelet coefficients are given by $v^{j}_{rs} = \int_{0}^{1} \tilde{I}^{j}_{\lfloor zT \rfloor} \psi_{rs}'(z)\, dz$. As in \cite{von1997adaptive}, we employ a slight abuse of notation, with $\psi_{r_{0}-1,s}' = \phi'_{r_{0},s}$, in order to include the scaling coefficient at the coarsest scale $r_{0}$ of the second wavelet scheme. The empirical analogues of the wavelet coefficients are given by
\begin{equation}\label{wav-thresh}
\hat{v}^{j}_{rs} = T^{-1} \sum_{n=0}^{T-1} \tilde{I}^{j}_{n,T} \psi_{rs}' (n/T), \quad \text{for } r = r_{0}, \ldots, \log_{2} (T), \quad s =0, \ldots , 2^{r} - 1.
\end{equation}
Analogously, we can build non-decimated wavelet coefficients for TI denoising. Let $\psi_{r}' (z) = 2^{r/2} \psi' (2^{r} z)$, and let
\begin{equation*}
\hat{v}^{j}_{rs} = T^{-1} \sum_{n=0}^{T-1} \tilde{I}^{j}_{n,T} \psi_{rs}' ((n-s)/T), \quad \text{for } r = r_{0}, \ldots, \log_{2} (T), \quad s =0, \ldots , T - 1.
\end{equation*}
Then, denoising is achieved by applying non-linear hard wavelet thresholding to the wavelet coefficients $\hat{v}^{j}_{rs}$. The resulting estimator is obtained by inverting the wavelet transform using only the coefficients which remain after thresholding:}
\begin{equation*}
\hat{I}^{j}_{\lfloor zT \rfloor} = \sum_{r} \sum_{s} \tilde{v}^{j}_{rs} \psi_{rs}' (z), \quad z \in (0,1),
\end{equation*}
\bernchange{where $\tilde{v}^{j}_{rs} = \hat{v}^{j}_{rs} \mathbb{I} (|\hat{v}^{j}_{rs}| > \lambda)$ are the hard thresholded wavelet coefficients with threshold $\lambda$. The particular threshold is specified in Theorem \ref{smooth_theorem4}, as is the appropriate set of indices $r$ over which to perform the summation.}

The theoretical argument to consistently estimate $S_{j}(k/T)$ is thus as follows. First, smooth the rescaled wavelet periodogram at each scale $j$ using wavelet thresholding. Next, use the (bounded) $P$-inverse matrix to correct the smoothed, rescaled periodogram. Finally, multiply the estimate at each scale by $2^{-j/2}$, since ${S}_{j} (k/T) = 2^{-j/2} \tilde{S}_{j} (k/T)$, yielding the final estimator, $\widehat{S}_j(k/T)$, of $S_{j}(k/T)$. Using the hard threshold $\lambda(j,r,s,T)^{2} = \text{Var} ( \hat{v}_{rs}^{j} ) \log^2 (T)$ when smoothing the periodogram via wavelet thresholding, we can show that the smoothed, corrected estimate $\widehat{S}_j (z)$ is \tcr{consistent with respect to the $L_{2}$ error}.

\begin{theorem}\label{smooth_theorem4}
Let ${\psi}'$ be a wavelet of bounded variation, with $2^r = o (T)$ for wavelet coefficients $\hat{v}_{rs}^j$. For a Gaussian trend-LSW process and using the threshold $\lambda^2 (j, r, s, T) = \emph{Var} ( \hat{v}_{rs}^{j} ) \log^2 (T)$, for each fixed $j$,
\begin{equation}\label{cons-spec}
 \mathbb{E} \left[ \int_0^1  \left( \widehat{S}_j (z) - S_j(z) \right)^2 dz \right] = \mathcal{O} \left( 2^{-j} T^{-2/3} \log^2 (T)\right).
\end{equation}
\end{theorem}
The rate obtained in Equation (\ref{cons-spec}) is a consequence of known results on wavelet thresholding estimators, utilised in \cite{neumann1995wavelet}, and the multiplication of $2^{-j/2}$ that occurs in the estimation procedure. The rate highlights the fact that differencing the time series has resulted in an ``information loss", with spectral estimation in coarser scales resulting in slower rates of convergence. As a consequence, we can only consistently estimate the wavelet spectrum for a proportion of the finest scales $j$. \tcr{In particular, we have that we need $2^{-j} = \mathcal{O} (T^{-2/3-\delta} )$ for some $\delta>0$ in order for the mean-squared error of the EWS estimator to be $o(1)$}. Next, we tackle estimation of the local autocovariance via the EWS estimate. \begin{proposition}\label{A1-c-est}
Define $\hat{c} (z,\tau)$ by replacing $S_{j} (z)$ by $\widehat{S}_{j} (z)$ in the equation for the local autocovariance and replacing the lower limit in the sum from $j=-\infty$ to $j=-J_{0}$, i.e.
\begin{equation*}
\hat{c} (z,\tau) = \sum_{j=-J_{0}}^{-1} \widehat{S}_{j} (z) \Psi_{j} ( \tau) .
\end{equation*}
Let $T \rightarrow \infty$ and let $J_{0} = \alpha \log_{2} T$ for $\alpha \in (0,1)$. Assume that $S_{j} (z) \leq D 2^{\gamma j}$ for some positive constant $D$, where $\gamma = (3\alpha)^{-1} - 1/2$. Then,
\begin{equation*}
\mathbb{E} \left[  \int_{0}^{1} \left( \hat{c} (z,\tau) - c(z,\tau) \right)^{2} dz   \right] = \mathcal{O} \left(T^{\alpha -2/3} \log^{2}(T)\right),
\end{equation*}
i.e. $\hat{c} (z,\tau)$ is a consistent estimator of $c(z,\tau)$ for each fixed $\tau \in \mathbb{Z}$, provided that $T^{\alpha-2/3} \log^{2} (T) \rightarrow 0$.
\end{proposition}

We reiterate that the rescaling argument used to achieve consistency in Theorem \ref{smooth_theorem4} and Proposition \ref{A1-c-est} is performed purely for theoretical reasons; practical considerations for estimation are discussed at the end of the section. The results in Theorem \ref{smooth_theorem4} and Proposition \ref{A1-c-est} show that in the case where the trend is Lipschitz continuous, we can consistently estimate both the EWS and LACV of the original process using first-order differences and a modified bias correction.

The assumption placed on the decay rate of the EWS in Proposition \ref{A1-c-est} is a purely technical one, utilised in order to ensure mean square consistency of the estimator. The assumption controls for the fact that the local autocovariance is estimated using the finest $J_{0}$ scales, instead of across infinite scales which are not available in practice. The specific form of the decay rate is calculated in order to balance with the error rate of the wavelet thresholding procedure.

%(Note that although this value for $\alpha$ does not satisfy the conditions of Proposition \ref{A1-c-est}, in practice it performs well). 

\subsection{$n$-th Order Differencing}\label{n-order-diff}

In some cases, a first-difference may not be enough to remove a trend. Further differencing can be performed, although it is usually only necessary to at most second-difference a time series \citep{brockwell1991time}. If we assume that the $(n-1)$-th derivative of $\mu$ is Lipschitz, then the $n$-th difference of the time series will be (asymptotically) free of trend. We denote the $n$-th difference of a time series as $\{ \nabla^{n} X_{t} \}$. 

To calculate the expectation of the squared non-decimated wavelet coefficients of the $n$-th differenced series, we can argue in a similar fashion to the case of first-differencing. Denote the operator $A^{n} = (A_{jl}^{n})_{j, l <0}$ by $A_{jl}^{n} : = \langle \Psi_j  , B^{n} \Psi_l  \rangle = \sum_\tau \Psi_j (\tau) \Psi_l (\tau -n)$, and the $J$-dimensional matrix $A_J^{n} : = (A_{jl}^{n})_{j, l=-1, \ldots , -J}$. The entries of the matrix $A^{n}_{J}$ are given by the inner product of the autocorrelation wavelet with the autocorrelation wavelet at lag $n$. If we difference a time series $n$ times, then the expectation of the squared wavelet coefficients will involve linear combinations of inner product autocorrelation wavelet matrices from lag 0 (i.e. the standard $A$-matrix) up to lag $n$ (the matrix $A^{n}$). The result for the expectation of the wavelet periodogram of the $n$-th differenced time series mirrors that of the first-differenced series, and is given by the following proposition.
\begin{proposition}\label{nth-diff-exp}
Let $\tilde{d}_{j,k} = \sum_{t} \nabla^{n} X_{t} \psi_{j,k-t}$ be the non-decimated wavelet coefficients of $\Delta^{n} X_{t}$, and let $\tilde{I}^{j}_{k}:= | \tilde{d}_{j,k}|^{2}$. If the $(n-1)$-th derivative of $\mu$ is Lipschitz, then
\begin{equation*}
\mathbb{E} (\tilde{I}^{j}_{k}) = \sum_{l} D^{n}_{jl} S_{l} (k/T) + \mathcal{O}(T^{-1}) ,
\end{equation*} 
where
\begin{equation*}
D^{n}_{jl} = {2n \choose n} A_{jl} + 2 \sum_{\tau=1}^{n} (-1)^{\tau} {2n \choose n+\tau}A_{jl}^{\tau}.
\end{equation*} 
\end{proposition}
\begin{corollary}
For second differences,
\begin{equation*}
\mathbb{E} (\tilde{I}^{j}_{k}) = \sum_{l} \left( 6 A_{jl} - 8 A_{jl}^{1} + 2 A_{jl}^{2}  \right) S_{l} \left( \frac{k}{T} \right) + \mathcal{O}(T^{-1}) .
\end{equation*} 
\end{corollary}
Note that Proposition \ref{nth-diff-exp} generalises Proposition 4 in \cite{nason2000wavelet}, in which $n=0$. As in the case of first-differences, the bias operator $D^{n}$ does not possess a bounded inverse. Intuitively, for higher order differences, the eigenvalues of $D^{n}$ decay to 0 at increasingly faster rates. As such, correcting the estimation in a similar fashion as was described in Theorem \ref{smooth_theorem4} yields much slower rates of convergence. 

For the remainder of the section, we focus on the situation where a second-difference suffices to remove the trend. We can show that using first-differences suffices to obtain a consistent spectral estimate, even though the trend has not been fully removed. The first-differences of the trend are Lipschitz continuous, and as such the magnitude of the wavelet coefficients can be bounded using the wavelet characterisation of H{\"o}lder spaces \citep{daubechies1992ten}. Using this bound, we can bound the error of the first-differenced estimator in terms of the error due to estimation and error due to the squared wavelet coefficients of the first-differenced trend. Using this argument we are able to show consistency of both the spectrum and local autocovariance estimator as follows.
\begin{theorem}\label{2nd-diff-spec}
Assume that the first derivative of $\mu$ is Lipschitz, and let $J_{1} = \beta \log_{2} T$ for $\beta \in (0,1)$. Further assume that the smoothed raw wavelet periodogram is corrected across the finest $J_{1}$ scales only. Under the same conditions as Theorem \ref{smooth_theorem4}, for each fixed $j$, $\widehat{S}_{j} (z)$ is a consistent estimator of $S_{j} (z)$, provided that $T^{7\beta-4} \rightarrow 0$ as $T \rightarrow \infty$, since
\begin{equation*}
 \mathbb{E} \left[ \int_0^1  \left( \widehat{S}_j (z) - S_j(z) \right)^2 dz \right] = \mathcal{O} \left( 2^{-j} T^{7\beta-4}  \right)   + \mathcal{O} \left( 2^{-j} T^{-2/3} \log^2 (T) \right) + \mathcal{O} \left( 2^{-j} T^{-\beta} \right).
\end{equation*}
Define $\hat{c} (z,\tau)$ by replacing $S_{j} (z)$ by $\widehat{S}_{j} (z)$ in the equation for the local autocovariance and replacing the lower limit in the sum from $j=-\infty$ to $j=-J_{0}$, i.e.
\begin{equation*}
\hat{c} (z,\tau) = \sum_{j=-J_{0}}^{-1} \widehat{S}_{j} (z) \Psi_{j} ( \tau) ,
\end{equation*}
where $J_{0} = \alpha \log_{2} T$ for $\alpha < \beta \in (0,1)$. Under the assumptions of Proposition \ref{A1-c-est}, provided that $T^{\alpha + 7\beta-4} \rightarrow 0$ and $T^{\alpha-2/3} \log^{2} (T) \rightarrow 0$, $\hat{c} (z,\tau)$ is a consistent estimator of $c(z,\tau)$, since for each fixed $\tau \in \mathbb{Z}$,
\begin{equation*}
\mathbb{E} \left[  \int_{0}^{1} \left( \hat{c} (z,\tau) - c(z,\tau) \right)^{2} dz   \right] = \mathcal{O} \left(T^{\alpha + 7 \beta -4} \right) + \mathcal{O} \left( T^{\alpha-2/3} \log^{2}(T)\right) + \mathcal{O}(T^{\alpha - \beta}). 
\end{equation*}
\end{theorem}

Hence, we can still consistently estimate the EWS and LACV via first-order differences when the trend of the time series has a Lipschitz continuous first derivative. Therefore, we argue that in most practical scenarios, it is sufficient to only perform one difference in order to estimate the evolutionary wavelet spectrum and local autocovariance in the presence of a trend. 

\textcolor{black}{In practice, we set $J_0=J_1$ when performing EWS and LACV estimation. That is, we use $J_1$ scales for estimating the EWS, and the same number of scales for estimating the LACV. Performing estimation in this fashion ensures that the spectral estimate is well-behaved, and is an approach commonly adopted in the LSW literature. We suggest using $J_{1} = \lfloor \beta \log_{2} (T) \rfloor$, with $\beta = 7/10$, motivated by extensive numerical results given in Appendix \ref{alpha-beta-parameter-choice}.} Furthermore, this choice is in agreement with other discussion in the literature, see for example \cite{sanderson2010estimating}. \textcolor{black}{Finally, an algorithmic description of the spectral estimation procedure, where smoothing is carried out either via wavelet thresholding or using a running mean, is detailed in Algorithm \ref{alg-method-spec}. }

\begin{algorithm}[t]
\caption{\textcolor{black}{Spectral Estimation Procedure}}
\label{alg-method-spec}
\DontPrintSemicolon
\SetAlgoLined
\SetKwData{return}{return}

\SetKwProg{Fn}{Function}{:}{}

\KwIn{Data $\{X_t\}_{t = 1}^n$, spectral estimation wavelet $\psi^0$, maximum scale $J_1$, smoothing wavelet $\psi'$ or smoothing bin width parameter $W$}

1. Compute wavelet periodogram of first-differenced time series:
\begin{equation*}
\tilde{I}_k^j  = \left( \sum_t \nabla X_t \psi_{j,k-t}^0 \right)^2, \ j = -1, \ldots , -J_1. 
\end{equation*}

2. Compute the smoothed wavelet periodogram $\hat{I}^j_k$ for $j = -1, \ldots , -J_1$, either by
\begin{enumerate}[(A)]
\item Wavelet thresholding:
\begin{equation*}
\hat{I}^j_{k} = \sum_{r} \sum_s \tilde{v}_{rs}^j \psi'_{rs} (k/T),
\end{equation*}
where $\tilde{v}^j_{rs} = \hat{v}^j_{rs} \mathbb{I} ( \hat{v}^j_{rs} > \lambda(j,r,s,T))$, with $\tilde{v}^j_{rs}$ given in Equation~\eqref{wav-thresh}, and $\lambda(j,r,s,T) = \hat{\sigma}_j \log (T)$ with $\hat{\sigma}_j$ computed as the median absolute deviation estimator of the $\hat{v}_{rs}^{j}$.
\item Running mean :
\begin{equation*}
    \hat{I}^j_k = \frac{1}{2W+1} \sum_{w=-W}^{W} \tilde{I}^j_{k+w}.
\end{equation*}
\end{enumerate}
3. Compute the corrected, smoothed wavelet periodogram as
\begin{equation*}
\hat{\bm{S}}(k/T) = (D^{1}_{J_1} )^{-1} \hat{\bm{I}}_k,
\end{equation*}
where $\hat{\bm{S}}(k/T) = \{ \hat{S}_j (k/T) \}_{j=-1}^{-J_1}$, $\hat{\bm{I}}_k = \{ \hat{I}^l_k \}_{l=-1}^{-J_1}$, and $D^1_{J_1}$ is defined in Definition \ref{A1_mat_def}.

\KwOut{Spectrum estimate $\{ \hat{S}_j (k/T)\}_{j=-1}^{-J_1}$ for $k=0, \ldots, T-1$.}
\end{algorithm}

\section{Trend Estimation Using the Spectral Estimate}

In this section, we discuss a wavelet thresholding approach for the estimation of the trend component of Model (\ref{lsw_rep}). If a first (or second) difference is capable of removing the trend from a time series, we have shown that we can consistently estimate the time-varying evolutionary wavelet spectrum using the smoothed, corrected raw wavelet periodogram of the differenced time series. We now wish to use this estimate in order to estimate the trend of the time series.

\subsection{Wavelet Thresholding Estimator}

The approach we take is to use the spectral estimate directly within a wavelet thresholding estimation procedure. In \cite{von2000non}, the authors describe a wavelet thresholding methodology for consistent curve estimation in the presence of locally stationary errors, subject to mild regularity conditions. The authors propose to use a local median absolute deviation pre-estimate for the variance of the wavelet coefficients, which is used in the threshold. We instead incorporate the consistent spectral estimate into the thresholding procedure. This yields an analogous version of Theorem 1 in \cite{von2000non}, for the specific case of Lipschitz continuous trend functions. 

With a slight abuse of notation, denote the estimated wavelet coefficients of the time series by $\hat{v}_{rs}$. The results in this section apply to the commonly used soft and hard thresholding rules, given respectively by
\begin{align*}
 \hat{v}_{rs}^{S} &= \text{sgn} (d_{r,s}) \left( |d_{r,s} | - \lambda ) \mathbb{I}(|d_{r,s} | > \lambda \right), \\
\hat{v}_{rs}^{H} &= d_{r,s} \mathbb{I}(|d_{r,s} | > \lambda ),
\end{align*}
where $\lambda = \lambda(r,s,T)$ is the threshold, and $\mathbb{I}$ is the indicator function. Asymptotic normality of the empirical wavelet coefficients, established in \cite{von2000non}, permits the use of a coefficient-dependent universal threshold $\lambda(r,s,T) = \sigma_{rs} \sqrt{2 \log(T)}$, where $\sigma_{rs}^{2}$ is the variance of the wavelet coefficients. This yields the following result for the wavelet thresholding estimator $\hat{\mu}$ obtained using the \tcr{DWT and the} threshold $\lambda(r,s,T)$, with either soft or hard thresholding, \bernchange{calculated in the same way as described in Section \ref{smooth-est-theory}.}  \tcr{We note here however, that the result also holds for TI wavelet denoising (outlined in Section 3.5) since the non-decimated wavelet coefficients can be seen as a set of DWT coefficients computed from cyclic shifts of the data.}

\begin{proposition}\label{smooth_theorem5}
Let $\tilde{\psi}$ be a wavelet of bounded variation, with $2^r = o (T)$ for wavelet coefficients $\hat{v}_{rs}$. For a trend-LSW process with Lipschitz continuous trend, and using the threshold $\lambda (r, s, T) = \sigma_{rs} \sqrt{2 \log(T)}$, the estimator $\hat{\mu}$ satisfies
\begin{equation*}
 \mathbb{E}\left[  \int_0^1  \left( \hat{\mu}(z) - \mu(z) \right)^2  dz \right] = \mathcal{O} \left( \left( \frac{\log (T) }{T} \right)^{2/3} \right).
\end{equation*}
\end{proposition}

Note that this result is subject to mild regularity assumptions on the LSW component of the model, all of which are satisfied. The innovations $\{ \xi_{j,k} \}$ are not restricted to be Gaussian, and can for example be exponential, gamma, or inverse Gaussian distributed. To estimate the variance $\sigma_{rs}^{2}$, which is necessary to choose the threshold $\lambda(r,s,T)$, we use an estimate of the variance of the empirical wavelet coefficients. If the LSW process is generated by wavelet $\psi^{0}$, and the wavelet used for thresholding is denoted $\psi^{1}$, the variance of the empirical wavelet coefficients $d_{j,k} := \sum_{t} X_{t} \psi^{1}_{j,k-t}$ is given by
\begin{equation}\label{emp-wav-var}
\text{Var} (d_{j,k}) = \sum_{l} C^{1,0}_{jl} S_{l} (k/T) + \mathcal{O}(T^{-1}),
\end{equation}
where $C^{(1,0)}_{jl} = \sum_{\tau} \Psi^{0}_{j} (\tau) \Psi^{1}_{l} (\tau)$, and where $\Psi_{j}^{0} (\tau)$ and $\Psi_{j}^{1} (\tau)$ are autocorrelation wavelets with respect to wavelets $\psi^{0}$ and $\psi^{1}$. By plugging in the estimate $\widehat{S}_{j} (z)$, obtained in Section \ref{smooth-est-theory}, into the expression \eqref{emp-wav-var}, this yields the universal-type threshold $\lambda(r,s,T) = \hat{\sigma}_{r,s} \sqrt{2 \log (T)}$, where $\hat{\sigma}_{r,s}^{2} = \sum_{l} C^{1,0}_{rl} \widehat{S}_{l} (s/T)$. 

\subsection{Practical Considerations} 

In alignment with discussion in \cite{von2000non}, in practice we analyse approximately the finest $7/10$ scales of the time series, the same as in the spectral estimation procedure. In practice, we have found that applying hard thresholding yields better performance. We recommend the use of translation invariant (TI) thresholding over a standard discrete wavelet transform. We have found that it offers stronger practical performance, in terms of the mean squared error of the estimator. As noted in \cite{nason2010wavelet}, use of a non-decimated transform ensures that there is ``more chance'' of the wavelet coefficients picking up the signal of the time series. 

Note that it is possible to obtain negative estimates of the variance of the wavelet coefficients, although based upon our empirical analyses this is rare. In this case we replace the negative values by the nearest neighbouring positive value, which was found to have no discernible impact on the trend estimation procedure. We recommend the use of the Daubechies Least Asymmetric wavelet with 4 vanishing moments, as we have found empirically that it works well for estimation purposes and also helps to minimise the number of negative variance estimates. Note that Proposition \ref{smooth_theorem5} holds for non-Gaussian time series, while theoretical results concerning the second-order estimation require an assumption of normality. In practice, our approach still performs well in the presence of non-Gaussian noise, as we show in the simulation study.

To complete the discussion, we provide algorithmic pseudo code for the trend estimation procedure using TI thresholding in Algorithm \ref{alg-method-trend}.

\begin{algorithm}[htp]
\caption{\textcolor{black}{Trend Estimation Procedure}}
\label{alg-method-trend}
\DontPrintSemicolon
\SetAlgoLined
\SetKwData{return}{return}

\SetKwProg{Fn}{Function}{:}{}

\KwIn{Data $\{X_t\}_{t = 1}^n$, spectrum estimate $\{ \hat{S}_j (k/T)\}_{j=-1}^{-J_1}$ for $k=0, \ldots, T-1$, spectrum estimation wavelet $\psi^0$, trend estimation wavelet $\psi^1$}

1. Compute non-decimated wavelet transform of data:
\begin{equation*}
v_{r,s} = \sum_t X_t \psi_{r,s-t}^1, \ j=-1, \ldots , - \log_2 (T) . 
\end{equation*}
2. Compute variance of wavelet coefficients as 
\begin{equation*}
\hat{\sigma}_{r,s}^2 = \text{Var} (v_{r,s}) = \sum_{l=-1}^{-J_1} C^{1,0}  \widehat{S}_l (s/T) .
\end{equation*}
3. Obtain the thresholded wavelet coefficients $\hat{v}_{r,s}$ using the soft or hard thresholding rules
\begin{align*}
 \hat{v}_{rs}^{S} &= \text{sgn} (d_{r,s}) \left( |d_{r,s} | - \lambda (r,s,T) ) \mathbb{I}(|d_{r,s} | > \lambda (r,s,T) \right), \\
\hat{v}_{rs}^{H} &= d_{r,s} \mathbb{I}(|d_{r,s} | > \lambda (r,s,T) ),
\end{align*}
where
\begin{equation*}
\lambda (r,s, T) =   \hat{\sigma}_{r,s} \sqrt{2 \log (T)}, s= -1, \ldots , -J_1.
\end{equation*}
4. Invert the thresholded wavelet coefficients $\hat{v}_{r,s}$ using basis averaging to obtain $\hat{\mu}(t/T)$.

\KwOut{Trend estimate $\hat{\mu}(t/T)$}
\end{algorithm}

\section{Simulation Study}\label{sim-study}

In this section we illustrate the ability of our proposed methodology to jointly estimate the mean and EWS of a trend-LSW process by performing a simulation study. For each set of simulations, we use the three EWS shown in Figure \ref{all-specs-plot}, which represent spectra with distinct characteristics. The spectra are explicitly defined in the supplementary material. Spectrum $S^{1}$, studied in \cite{nason2010wavelet}, displays coarse-scale, slowly-evolving sinusoidal behaviour with a fine-scale burst in power at time point 800. Spectrum $S^{2}$ is a concatenation of moving average processes and contains power moving from fine to coarser scales, and was examined in \cite{nason2000wavelet}. Spectrum $S^{3}$ contains slowly-evolving power at fine scales. 

\begin{figure}[]
\centering
\includegraphics[width = 0.9\textwidth]{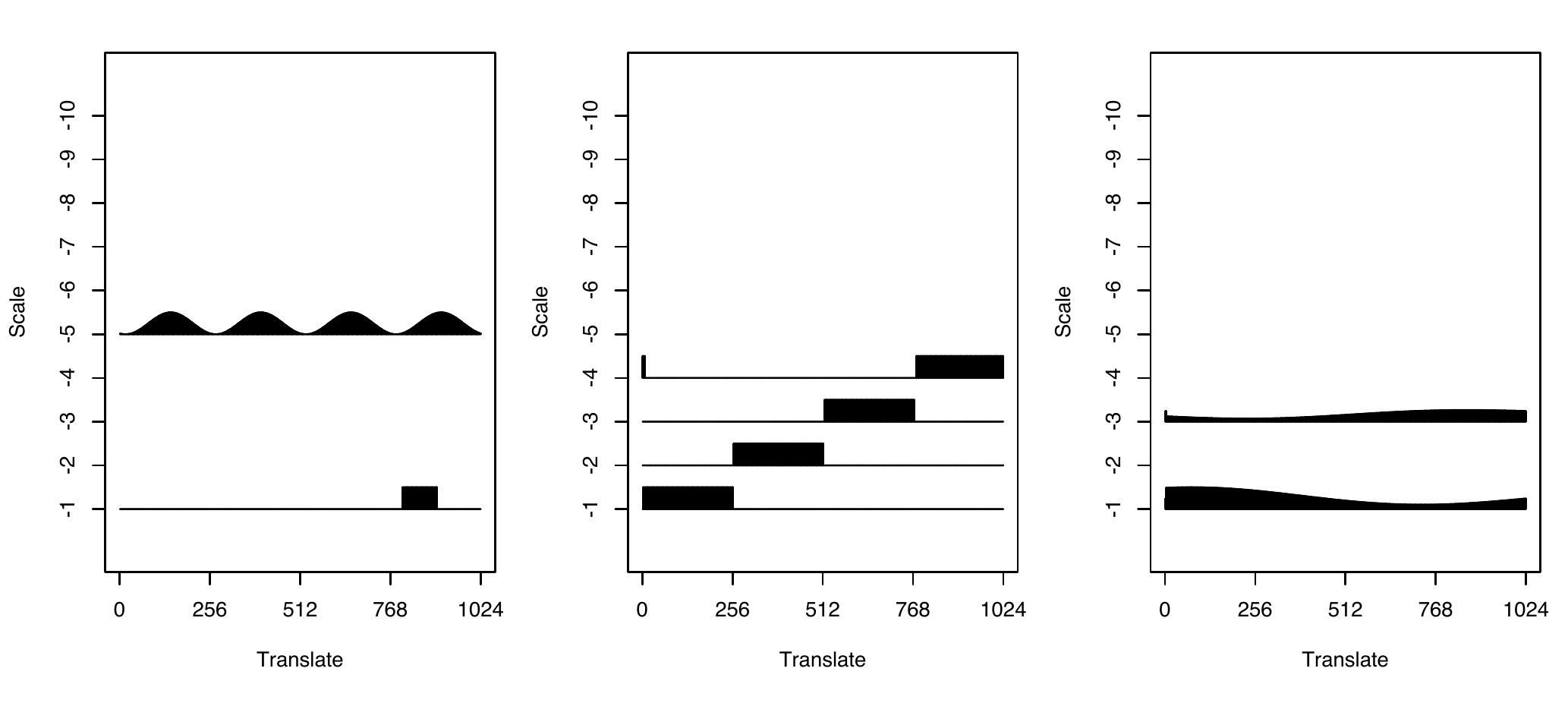}
\caption{Left: spectrum $S^{1}$, sinusoid with ``burst''. Centre: $S^{2}$, concatenated moving average process. Right: $S^{3}$, slowly-evolving fine-scale power.}
\label{all-specs-plot}
\end{figure}

We simulate 100 realisations of time series $\{ X_t \}_{t=0}^{T-1}$ of length $T = 2^{10} = 1024$ from LSW processes with those spectra, with different additive trend functions. The trends used in the simulation study are a linear, sinusoidal, logistic and piecewise quadratic trend, denoted in Figure \ref{reals} by $\mu_{li}$, $\mu_{s}$, $\mu_{lo}$ and $\mu_{q}$ respectively, and defined explicitly in the supplementary material. These functions are Lipschitz continuous with varying degrees of smoothness, with $\mu_{q}$ also being non-differentiable at two time-points. All LSW processes were simulated using the Daubechies Extremal Phase wavelet with 4 vanishing moments. Similar results, which can be found in the supplementary material, were reported for other wavelets. Example realisations from the simulation study are shown in Figure \ref{reals}. All simulations were performed in \verb!R! with estimation procedures implemented using modifications to code in the \verb!locits! \citep{locits} and \verb!wavethresh! \citep{nason2010wavethresh} packages. 

\begin{figure}[]
\centering
\includegraphics[width = 0.84\textwidth]{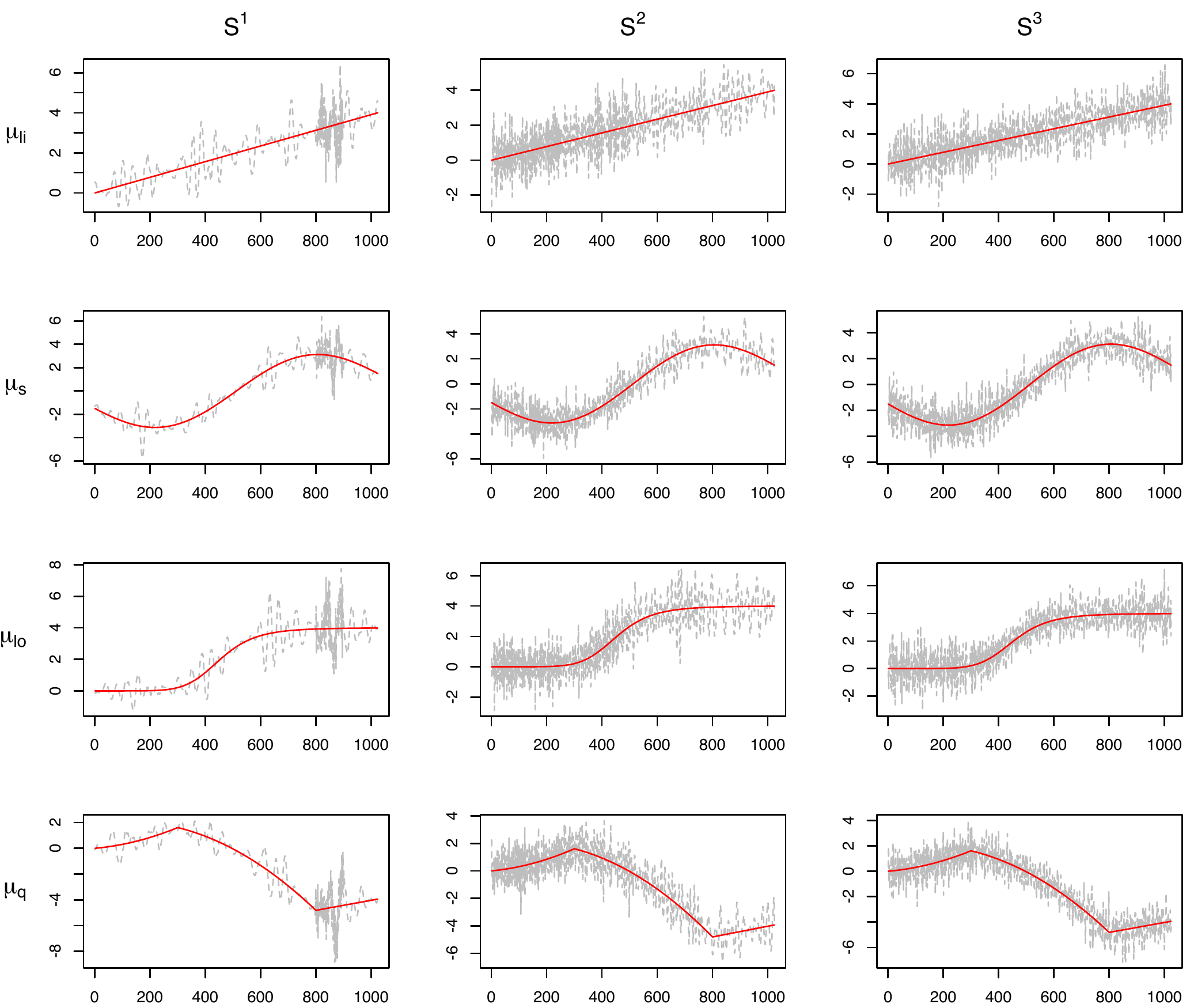}
\caption{Example realisations from each trend and spectrum scenario. Dashed line shows time series with true underlying trend shown in solid line. Left: spectrum $S^{1}$, sinusoid with ``burst'' . Centre: $S^{2}$, concatenated moving average process. Right: $S^{3}$, slowly-evolving fine-scale power.}
\label{reals}
\end{figure}

\subsection{Spectral Estimation Performance}\label{spec-est-sims}

In this simulation we show that by first-differencing to remove the trend, we can obtain an unbiased EWS estimate, which in turn can be used to obtain a trend estimate that performs well. For each realisation, the un-smoothed estimate of the EWS was calculated, which was then used to obtain an averaged estimate for the EWS across the 100 realisations. In alignment with the discussion in Section \ref{smooth-est-theory}, we correct the raw wavelet periodogram across the finest 7 scales.

We report the mean squared error of the averaged spectrum compared with the true spectrum in Table \ref{table1}. We compare these values with the mean squared error obtained by using the standard LSW estimation procedure of \cite{nason2000wavelet}. In this case, no trend is present, the estimate is performed using the original $A$ matrix inverse for bias correction, and no differencing is performed. This is calculated using the \verb!ewspec3! command in the \verb!locits! \verb!R! package. This is reported in the ``None'' row in Table \ref{table1}, and represents the `best-case' performance with which to compare. We see that despite the presence of a trend, differencing enables us to approximately remove it in order to accurately estimate the spectrum. The estimation error is smaller using our methodology than the ``None'' case for spectrum 2, while it is marginally worse for the other two spectra. \tcr{Visual inspection of the resulting estimators shows they are generally satisfactory, with examples given in Figure \ref{trend-spec-plot}. Note that, while the spectral estimation could also have been performed using a second difference instead of first, this over-differencing results in higher estimation error. A numerical comparison between first and second differences, showing the effect of over-differencing, is given in Appendix \ref{over-diff-sec}.}

\begin{table}[H]
\centering
\begin{tabular}{c|c|c|c}
Trend               & Spectrum 1                 & Spectrum 2                 & Spectrum 3                 \\ \hline
None                & 3.13                      & 4.88                         & 1.87                 \\ 
Linear              & 3.32                   & 4.63                     & 2.76                     \\ 
Sine                & 3.32                      & 4.63                     & 2.76                     \\ 
Logistic            & 3.32                   & 4.63                    & 2.76                     \\ 
Piece. Quad. & 3.32                     & 4.67                      & 2.79                   \\ 
\end{tabular}
\caption{Mean squared error comparison for the averaged spectrum estimate, multiplied by $10^{3}$, across the spectrum and trend scenarios.}
\label{table1}
\end{table}

%
%\begin{figure}[H]
%\centering
%\begin{tabular}{|l|l|l|l|l|l|}
%\hline
%                       & \multicolumn{5}{c|}{Sum of squared error (SSE)}              \\ \hline
%Spectrum & No trend       & Linear        & Sine   &  Logistic & Piecewise quadratic  \\ \hline
%$S^{1}$, burst                  & \textbf{14.69}            &   18.06      &  18.06    &   18.07 &   18.45 \\ \hline
%$S^{2}$, moving average   &   {21.00}   & \textbf{19.66}   &  19.66 & 19.66 & 20.01 \\ \hline
%$S^{3}$, slowly-evolving     &    \textbf{4.81}         &   11.06     & 11.06 & 11.06 & 11.34 \\ \hline
%\end{tabular}
%\caption{SSE comparison across the three defined spectra in Figure \ref{all-specs-plot} and trend scenarios given in ??}
%\label{611}
%\end{figure}

 \begin{figure}[]
\centering
\includegraphics[width=0.84\textwidth]{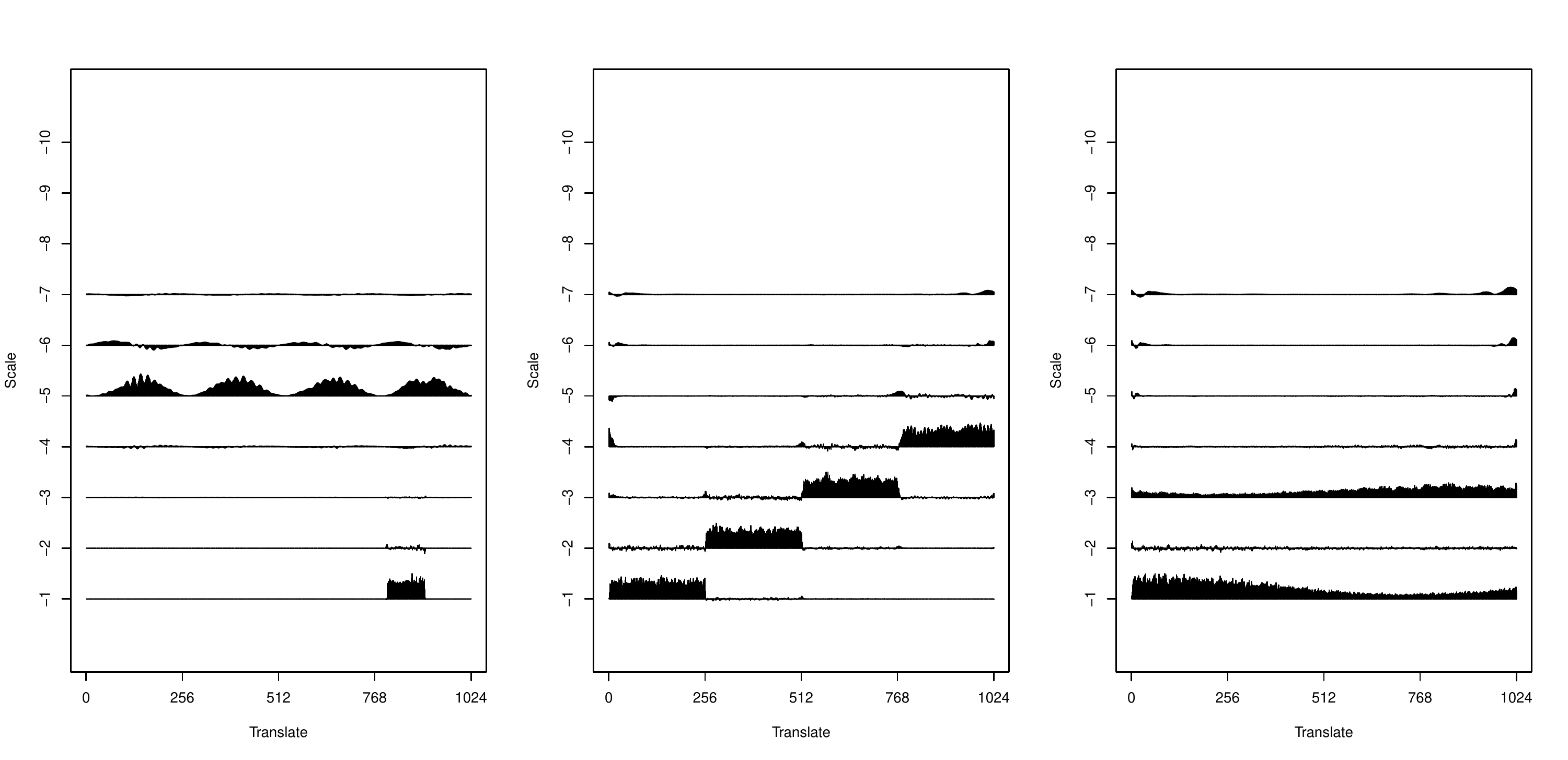}
\caption{\tcr{Example averaged spectral estimates. Left: spectrum 1, linear trend. Middle: spectrum 2, logistic trend. Right: spectrum 3, piecewise quadratic trend.}}
\label{trend-spec-plot}
\end{figure}

\subsection{Trend Estimation Performance}

We next assess the trend estimation procedure. The trend estimate is computed using the Daubechies Least Asymmetric wavelet with 4 vanishing moments, analysing the finest 6 scales. The spectrum estimate used in the thresholding procedure is smoothed with a running mean of bin width size 128. We compute the average mean squared error across the 100 realisations, as well as the standard deviation of the errors. For all simulations, a hard threshold of the form $\hat{\sigma}_{r,s} \sqrt{2 \log (1024)}$ is used.

We compare our method to two other wavelet-based trend estimation methods. In order to make a fair comparison, we use the same hyper-parameters where applicable. Firstly, we compare to the standard wavelet thresholding method based upon a \tcr{global, time-independent universal threshold $\hat{\sigma}_{r} \sqrt{2 \log (1024)}$, computed using the \texttt{wavethresh} package in R}. \tcr{This estimator assumes that the error process is second-order stationary, and is referred to as stationary wavelet thresholding (SWT) in the tables.} Secondly, we compare to the wavelet-based trend estimation procedure in \cite{von2000non}. No code is publicly available for this method, and so we have implemented the method utilising \verb!wavethresh! and following the description of the computation of the threshold in Section 2.5 of \cite{von2000non}. In the tables, the abbreviation LSWT refers to our method of locally stationary wavelet thresholding, and MVSWT refers to the wavelet-based method of \cite{von2000non}.

We repeat this simulation twice; firstly using Gaussian innovations in the LSW process, and secondly using exponentially distributed innovations. Further results examining the performance in other, non-LSW error scenarios can be found in the supplementary material. The results of the simulation for Gaussian LSW processes are reported in Table \ref{tableEP4t1}, while the results for exponentially distributed innovations are shown in Table \ref{tableEP4t2}. Values in bold are the lowest values for each trend and spectrum value across the three methods. 

From the tables we observe that the performance of our estimator in the presence of exponentially distributed random innovations is comparable to the Gaussian case. The error is consistent across the various time series scenarios, and the standard deviation is low. We see that our method outperforms the two other wavelet-based methods across all trend and spectrum scenarios.

\begin{table}[H]
\centering
\begin{tabular}{c|c|c|c|c}
\multirow{2}{*}{Trend} & \multirow{2}{*}{Spectrum} & \multicolumn{3}{c}{Method}    \\ \cline{3-5}
      &        & LSWT      & SWT        & VSMWT                \\ \hline
\multirow{3}{*}{Linear}     &  1      &  \textbf{0.024} (0.012)     &   0.519 (0.112)     &  0.303   (0.100)            \\
&2        & \textbf{0.030} (0.019)      & 0.762 (0.078)       & 0.225   (0.068)             \\
      &  3      & \textbf{0.028} (0.024)      & 0.441 (0.040)         &   0.160 (0.038)        \\ \hline
\multirow{3}{*}{Sine}     &  1      & \textbf{0.022} (0.010)         &  0.526 (0.110)       &   0.286 (0.109)             \\
&  2        &  \textbf{0.026} (0.014)    & 0.744 (0.075)         &  0.215  (0.068)             \\
      & 3      &   \textbf{0.022} (0.017)     &   0.447 (0.043)       &  0.156  (0.037)              \\ \hline
\multirow{3}{*}{Logistic}     &  1      &  \textbf{0.023} (0.015)        &  0.494 (0.109)      & 0.261 (0.093)    \\ 
&  2        & \textbf{0.033} (0.019)     & 0.753 (0.083)         & 0.226 (0.071)               \\
      &  3      & \textbf{0.027} (0.024)     & 0.449 (0.037)       & 0.164 (0.035)                \\ \hline
 \multirow{3}{*} {Piece. quad.}    &  1      & \textbf{0.022} (0.012)      &  0.517 (0.127)        &  0.290 (0.105)           \\
& 2        &  \textbf{0.032} ( 0.018) &   0.720 (0.087)       &  0.216  (0.072)              \\
      &  3      & \textbf{0.028} (0.024)      & 0.366 (0.046)       &   0.167 (0.039)           
\end{tabular}
\caption{Average mean squared error and standard deviation in brackets of trend estimate over 100 realisations generated using Gaussian innovations.}
\label{tableEP4t1}
\end{table}

\begin{table}[H]
\centering
\begin{tabular}{c|c|c|c|c}
\multirow{2}{*}{Trend} & \multirow{2}{*}{Spectrum} & \multicolumn{3}{c}{Method}    \\ \cline{3-5}
      &        & LSWT      & SWT        & VSMWT              \\ \hline
\multirow{3}{*}{Linear}     &  1      &  \textbf{0.030} (0.018)       &    0.521 (0.148)      &  0.305 (0.127)             \\
&2        & \textbf{0.035} (0.024)     & 0.779 (0.108)         &  0.328 (0.092)              \\
      &  3      & \textbf{0.040} (0.026)       & 0.497 (0.052)        &  0.318 (0.057)              \\ \hline
\multirow{3}{*}{Sine}     &  1      &\textbf{0.027} (0.025)         &   0.520 (0.119)    & 0.308 (0.102)            \\
&  2        &  \textbf{0.033} (0.020)    &  0.782 (0.110)        &  0.322  (0.085)             \\
      & 3      &   \textbf{0.037} (0.022)     & 0.495 (0.060)         &  0.309 (0.064)                  \\ \hline
\multirow{3}{*}{Logistic}     &  1      &  \textbf{0.030} (0.018)        &  0.508 (0.126)     &0.291 (0.108)                 \\ 
&  2        & \textbf{0.036} (0.022)    &   0.788 (0.098)       &    0.332  (0.085)           \\
      &  3      & \textbf{0.044} (0.031)      &  0.506 (0.061)        &   0.327 (0.058)             \\ \hline
 \multirow{3}{*} {Piece. quad.}    &  1      & \textbf{0.031} (0.016)      & 0.524 (0.116)       & 0.319 (0.106)               \\
& 2        & \textbf{0.038} (0.022) &       0.748 (0.098)     &0.326 (0.100)            \\
      &  3      & \textbf{0.045} (0.031)      &  0.432 (0.055)        &  0.318 (0.055)          
\end{tabular}
\caption{Average mean squared error and standard deviation in brackets of trend estimate over 100 realisations generated using exponential innovations.}
\label{tableEP4t2}
\end{table}

\subsection{Extension to Seasonal Time Series}\label{seasonal-extension}

\tcr{Our approach can be extended empirically to time series that display an additive deterministic seasonal component. From the discussion in Remark \ref{seasonal-remark}, if the time series possesses a (known) number of seasons $K$, lag $K$-differencing can be performed to remove the seasonality in order to estimate the EWS of the original time series. To illustrate this, we investigated the performance of the method in the presence of seasonal and smooth trends.}

\tcr{We simulated 100 realisations of time series according to the various trend and spectrum scenarios described in Section \ref{spec-est-sims}. To each of these scenarios, we added a (monthly) stationary seasonal component with $K=12$, where the value of each of the 12 seasonal components $K_1, \ldots, K_{12}$ is generated from a uniform random variable on the interval $[0,10]$. That is, the time series is generated by
\begin{equation*}
X_t = \mu \left(  \frac tT \right) +s_t +\epsilon_t,   
\end{equation*}
where $\mu$ is the smooth trend function, the seasonal component $s_t =K_1$ for $t \mod 12 =1$, $\ldots , s_t = K_{12}$ for $t \mod 12 = 0$, and $\epsilon_t$ is the LSW component.}

\tcr{Lastly, we considered the case where the seasonal component can be time-varying, with the smooth trend component $\mu_t = 0$. In this case, the seasonal component is given by a linear trend where the initial values are simulated from the uniform distribution on the interval $[0,10]$, with slope parameters generated as a uniform random variable on the interval $[-0.05,0.05]$. }

\begin{table}[H]
\centering
\begin{tabular}{c|c|c|c}
Trend               & Spectrum 1                 & Spectrum 2                 & Spectrum 3                 \\ \hline
Seasonal + No Trend & 4.76  & 8.44  &  2.54         \\ 
Seasonal + Linear  &  4.76   &  8.44 &   2.54   \\ 
Seasonal + Sine &   4.76 &  8.43 &     2.54    \\ 
Seasonal + Logistic & 4.76   &    8.44  & 2.54  \\ 
Seasonal + Piece. Quad. & 4.79 & 8.46  &  2.55  \\ 
Time-Varying Seasonal + No Trend & 4.76 & 8.43 & 2.53
\end{tabular}
\caption{Mean squared error comparison for the averaged spectrum estimate, multiplied by $10^{3}$, across the spectrum, trend, and seasonality scenarios.}
\label{seasonal-table}
\end{table} 
 
 \tcr{The results of these experiments are given in Table \ref{seasonal-table}, which can be compared to the results for the lag 1 differences in Table \ref{table1}. In particular, the performance is actually slightly better for spectrum 3; the smoothest spectrum. The other 2 spectra yield worse results, however: the performance drop-off is most noticeable for spectrum 2, which contains piecewise stationary components. This is due to the increased bias caused in estimating the time-varying spectra using lag $12$ differences instead of lag $1$ differences. Examples of the averaged spectrum estimates are given in Figure \ref{seasonal-spec-plot}, qualitatively showing strong similarities to the examples in Figure \ref{trend-spec-plot}.}
 
 \begin{figure}[]
\centering
\includegraphics[width=0.84\textwidth]{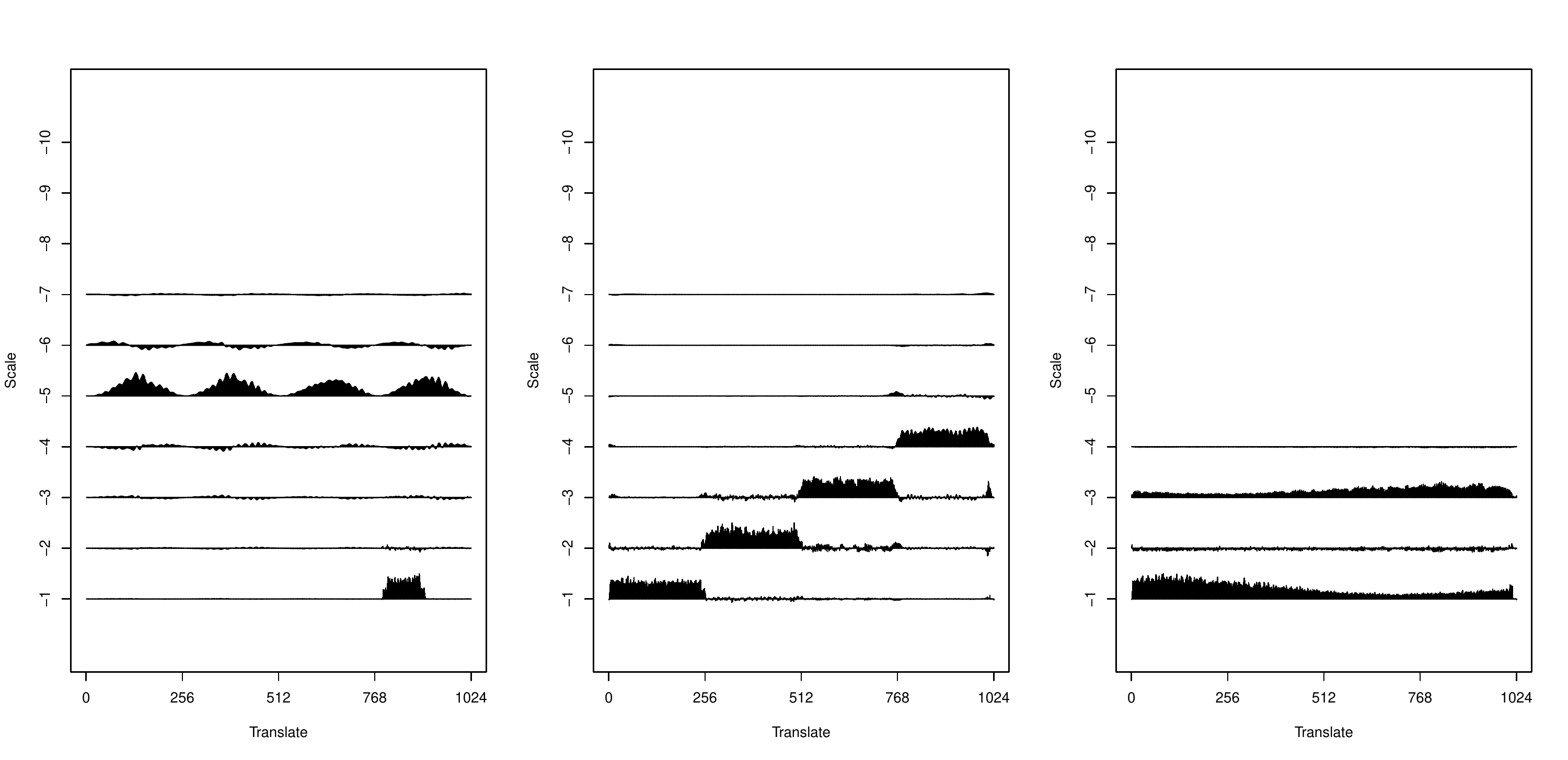}
\caption{\tcr{Example averaged spectral estimates. Left: spectrum 1, seasonal trend only. Middle: spectrum 2, seasonal and sinusoidal trend. Right: spectrum 3 and time-varying seasonal trend only.}}
\label{seasonal-spec-plot}
\end{figure}

\section{Data Application: Baby Electrocardiogram Data}

Next, we discuss a data application that highlights the benefits of our proposed methodology. We re-analyse the baby electrocardiogram (ECG) data set first examined in \cite{nason2000wavelet}. Furthermore, in the supplementary material, we analyse a wave height data set collected from a buoy in the Atlantic Ocean that was studied in \cite{killick2012optimal}. 

In Figure \ref{baby-ecg} left, we see a time series for the ECG reading of a 66-day-old infant. The series is available in the \verb!wavethresh! package and was collected by the Institute of Child Health at the University of Bristol. The series can be seen to exhibit both nonstationary trend and second-order behaviour. As \cite{nason2000wavelet} note, one reason for the nonstationarity is that the ECG varies considerably over time and changes significantly between periods of sleep and waking. In \cite{nason2000wavelet}, the authors first-difference the series to remove the trend, as shown in Figure \ref{baby-ecg} right, resulting in an approximately zero-mean time series. Then, LSW analysis is performed on the differenced series.

\begin{figure}[H]
\centering
\includegraphics[width = 0.95\textwidth]{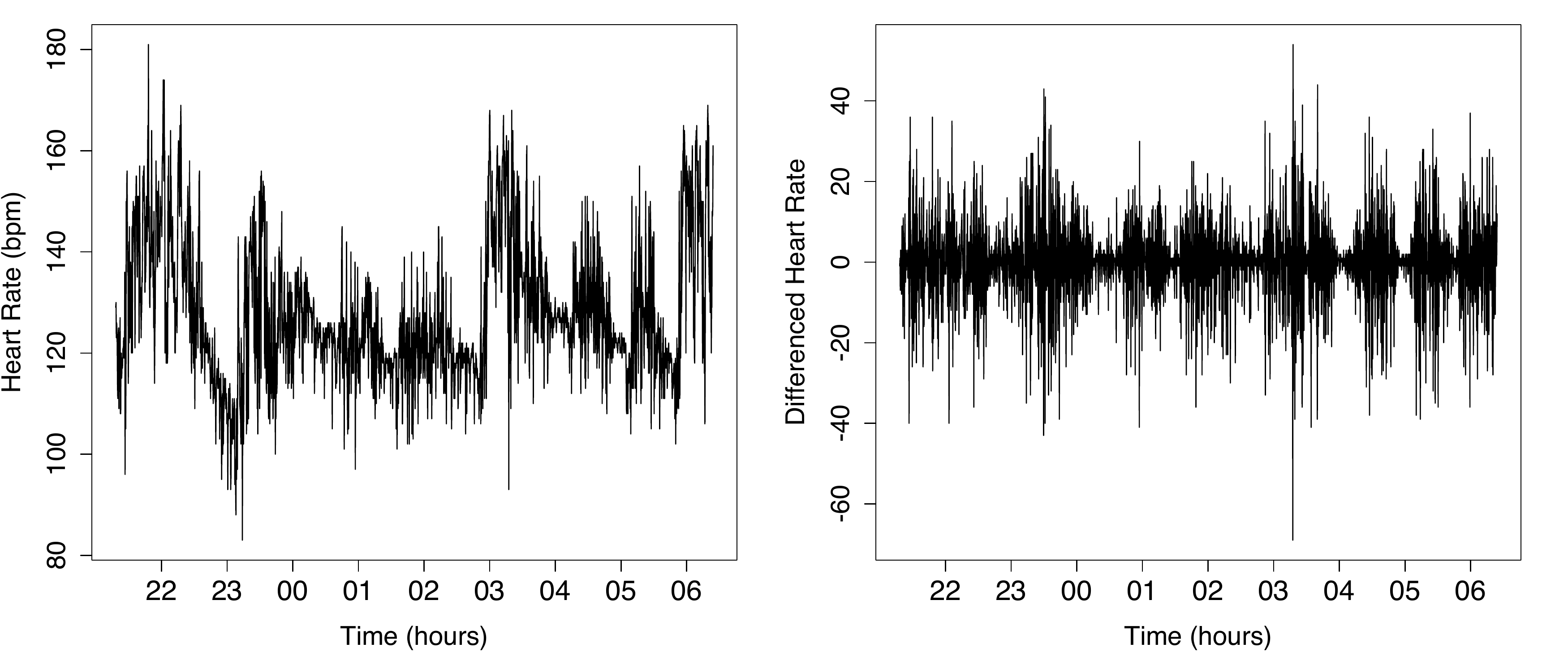}
\caption{Left: original ECG data. Right: first-differenced ECG data.}
\label{baby-ecg}
\end{figure}

We instead perform analysis on the original series, and note the similarities and differences to the original analysis in \cite{nason2000wavelet}. We use the Daubechies Least Asymmetric wavelet with 10 vanishing moments to analyse both series and the same bin width of size 128 for the running mean smoother, but as discussed we use a different matrix to correct for bias. In Figure \ref{baby-ecg-spec-comp} left, we see our analysis on the original series. The top plot is not scaled, while the bottom plot is individually scaled at each level. On the right, we see (our version of) the original analysis of \cite{nason2000wavelet}, again with top plot unscaled and bottom plot scaled individually. In line with previous discussion, only the finest 7 scales were used for analysis. From the top plots, we can see the effect that differencing has had on the resulting spectral estimate. In our estimate based on the original time series, we see that the power is spread fairly evenly across scales. However, in the spectral estimate of the differenced series, power is concentrated in the finest scales. From the bottom plots, we see that both estimates exhibit similar overall behaviour, with both evolving over time in a similar fashion.

Due to the difference in magnitude between the estimated spectra at different scales, some features of the original series can be identified more easily in our analysis. In \cite{nason2000wavelet}, an accompanying data set, the sleep state of the observed infant, is also used in the analysis. The sleep state is judged by a trained human observer, and is measured as either quiet (1), between quiet and active (2), active (3), or awake (4). A strong association between sleep state and spectral value is observed for scales $-1$ to $-5$ in \cite{nason2000wavelet}. However, using the original series instead of the differenced series can allow this association to become more apparent: in particular, consider the spectrum estimate at scale $-4$, enlarged in Figure \ref{sleep-state-plot}. In solid line is the spectral estimate of the original series, while the dashed line shows the spectral estimate for the differenced series. The dotted line shows the sleep state of the infant. Our estimate correlates strongly with the sleep state, especially in highlighting periods of being awake (sleep state 4). In general, the ``signal'' of variability appears stronger in our analysis, which is due to the differencing lowering the level of autocorrelation within the series. Furthermore, our estimate contains fewer negative spectral values, perhaps suggesting that the original time series is better represented as an LSW process as opposed to the first-differenced series.
\begin{figure}[h!]
\centering
\includegraphics[width = \textwidth]{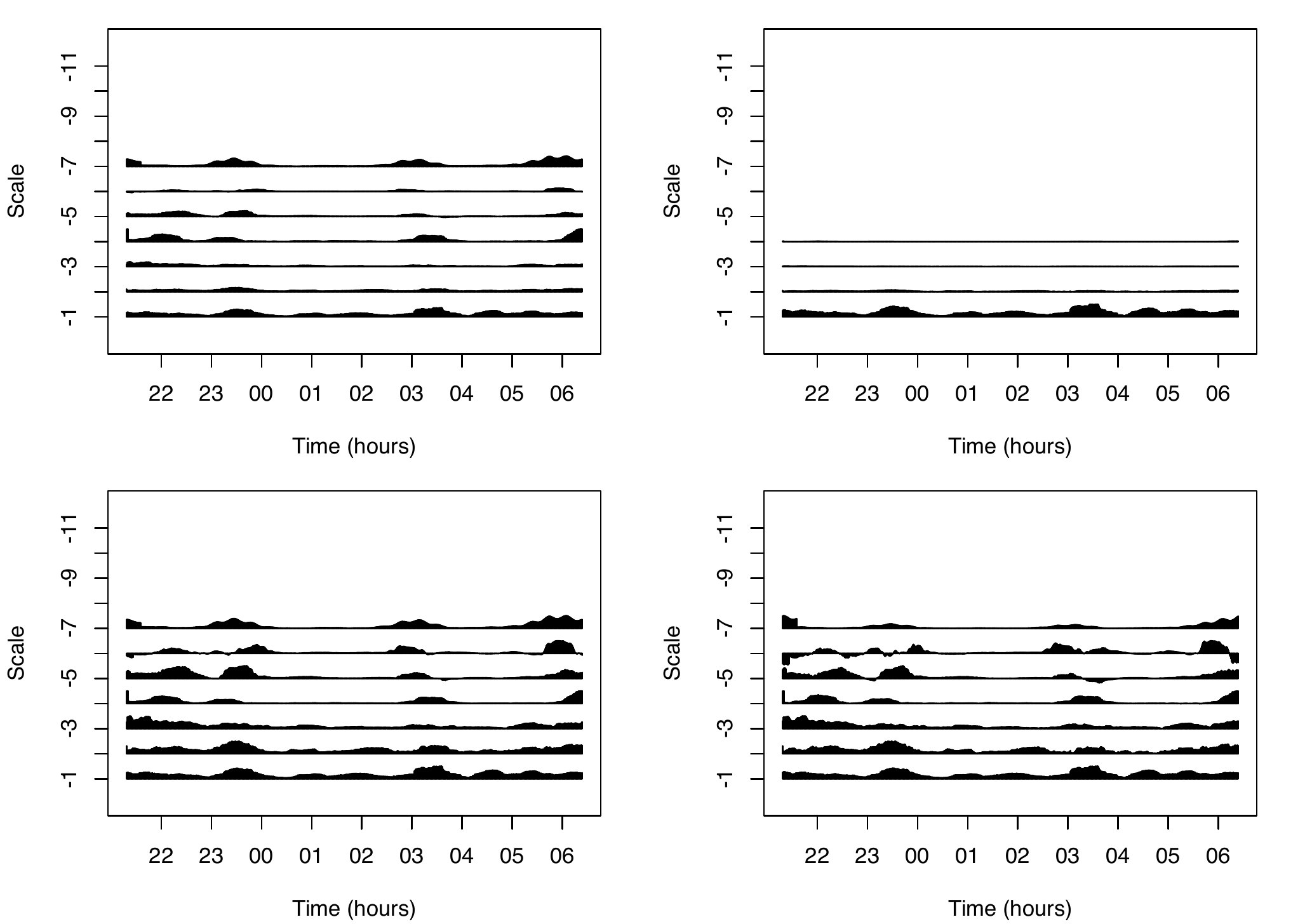}
\caption{Comparison of the spectral estimates. Top left: spectrum estimate using our methodology, unscaled, bottom left: scaled individually. Top right: spectrum estimate of the differenced series, unscaled, bottom right: scaled individually.}
\label{baby-ecg-spec-comp}
\end{figure}

\begin{figure}[h!]
\centering
\includegraphics[width = 0.85\textwidth]{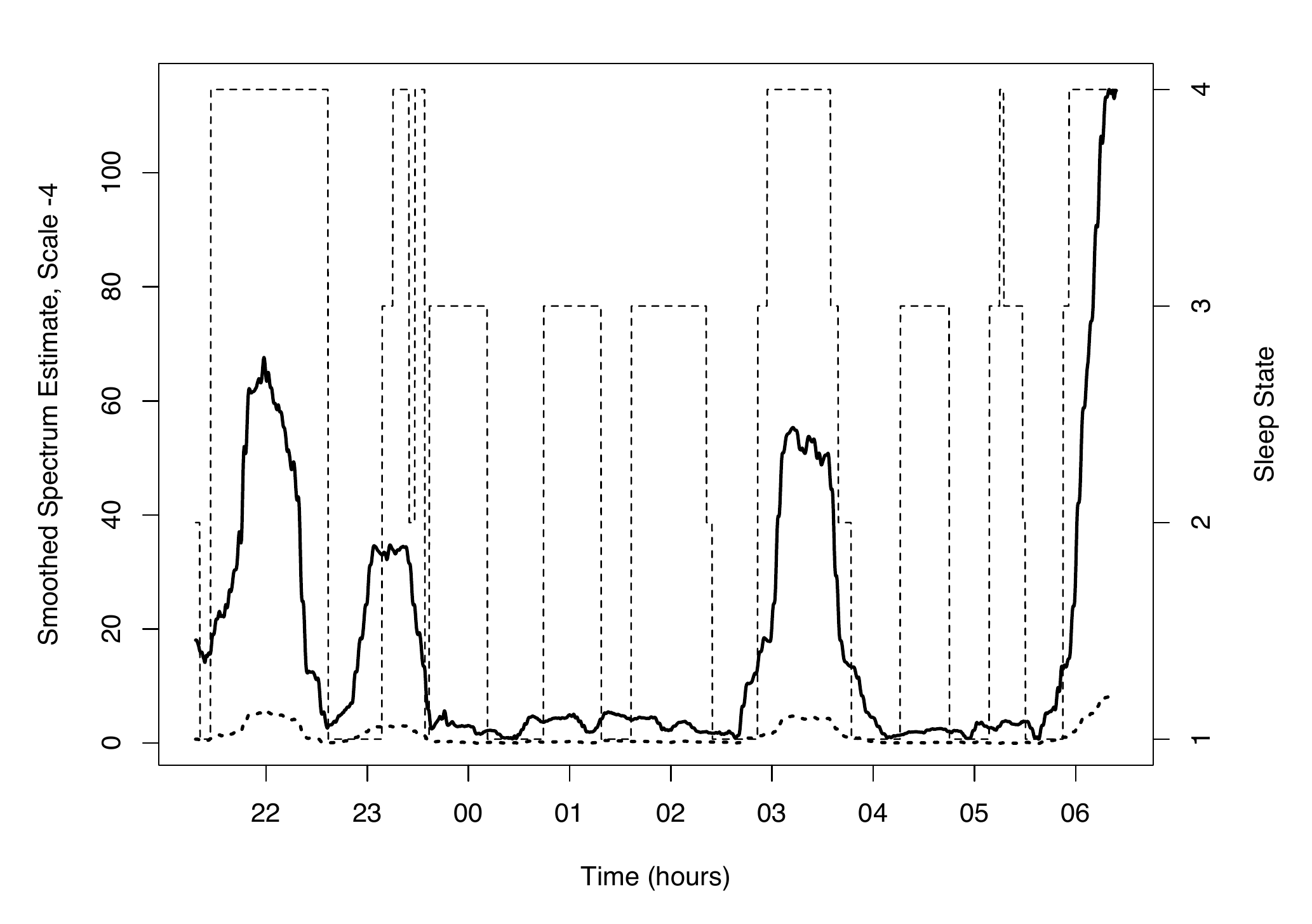}
\caption{Estimate of the EWS at scale $-4$: solid line is for the original series, while the dashed line is for the differenced series. Sleep state shown in dotted line.}
\label{sleep-state-plot}
\end{figure}

We perform TI wavelet thresholding using the Daubechies Least Asymmetric wavelet with 4 vanishing moments to estimate the trend of the Baby ECG series. We analyse the finest 7 scales of the series, using a hard universal threshold of $\hat{\sigma}_{r,s} \sqrt{2 \log (2048)}$, where $\hat{\sigma}_{r,s}$ is calculated using the spectral estimate in Figure \ref{baby-ecg-spec-comp} left. The trend estimate is shown in Figure \ref{ecg-trend-plot} in the solid line. We see that in general the estimate is quite smooth, with more rapid changes in mean occurring at approximately 23:00, 03:00, and 06:00, corresponding to changes in sleep state. The estimate again correlates with the underlying sleep state, highlighting the benefit of performing a joint first and second-order analysis of the data.

\begin{figure}[H]
\centering
\includegraphics[width = 0.7\textwidth]{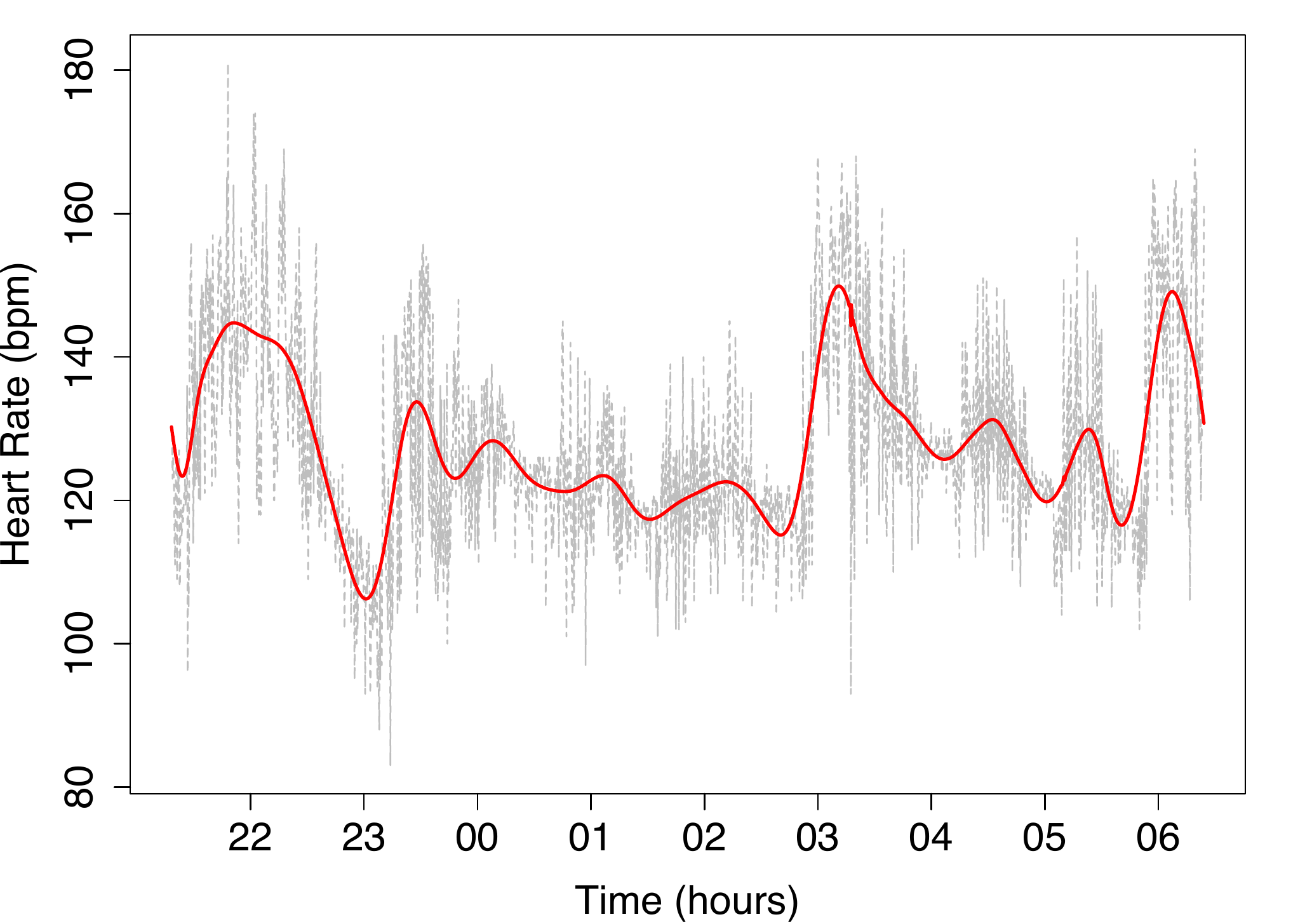}
\caption{Trend estimate for the Baby ECG data shown in solid line, with data shown in dashed line.}
\label{ecg-trend-plot}
\end{figure}

%\begin{figure}[H]
%\centering
%\includegraphics[width = 0.9\textwidth]{spec-comps-ecg.pdf}
%\caption{Comparison of the spectral estimates. Left: spectrum estimate using our methodology, right: estimate obtained from the differenced time series.}
%\label{baby-ecg-spec-comp}
%\end{figure}

\section{Concluding Remarks}

In this article, we have considered the problem of jointly modelling time-varying first and second-order properties of nonstationary time series. Our model employs the locally stationary wavelet model of \cite{nason2000wavelet}, used for modelling second-order nonstationarity, and adapts it to incorporate a nonstationary trend component. Using the common statistical technique of differencing, we have shown that we can consistently estimate the evolutionary wavelet spectrum and local autocovariance of the original series. Using these results, we have proposed a wavelet thresholding approach for the nonparametric estimation of the trend of the time series. The trend estimation methodology benefits from the information provided by the ability to consistently estimate the second-order structure of the time series, despite the presence of the trend component. 

Using our methodology we have analysed two time series that highlight the strength of the method to provide information about the original time series, which can become masked when only considering the differenced time series. In particular, our analysis of the Baby ECG data set highlights the benefit of modelling the stochastic component of the original series for the purpose of identifying sleep states more readily.

\bibliographystyle{apalike} 
\bibliography{EJSBib}

\begin{appendices}

\section{Additional Numerical Results}\label{appA}

\subsection{Spectra and Trend Functions Used}

In this section we provide additional information for the discussion in Section 5 of the main text. The three evolutionary wavelet spectra used in the simulation study in Section 5 are defined as

\begin{align*}
S^{1}_j(z) &= \begin{cases}
\sin^2 (4 \pi z) & \text{for } j = -5, z \in(0,1), \\
1 & \text{for } j = -1, z \in(800/1024, 900/1024), \\
0 & \text{otherwise},
\end{cases} \\
S^{2}_j(z) &= \begin{cases}
1 & \text{for } j = -1, z \in(0/1024, 256/1024), \\
1 & \text{for } j = -2, z \in(256/1024, 512/1024), \\
1 & \text{for } j = -3, z \in(512/1024, 768/1024), \\
1 & \text{for } j = -4, z \in(768/1024, 1), \\
0 & \text{otherwise},
\end{cases} \\
S^{3}_j(z) &= \begin{cases}
\frac{1}{2} + \frac{1}{4} \sin( \pi z) - \frac{1}{2} \cos( 3\pi z/2) & \text{for } j = -1, z \in(0,1), \\
\frac{1}{2} -\frac{1}{8} \sin(2 \pi z)- \frac{1}{4} \cos( \pi z/2) & \text{for } j = -3, z \in(0,1), \\
0 & \text{otherwise}.
\end{cases}
\end{align*}
The four trend functions used in the simulation study are defined as
\begin{align*}
\mu_{li}(z)& = 4z,\\
\mu_{s}(z)& = -2 \sin(2 \pi z) -\frac{3}{2} \cos (\pi z)  ,\\
\mu_{lo} (z) &=  \frac{4}{1 + \exp(4-7\log 4z)},\\[1ex]
\mu_{q}(z)& = \begin{cases}
12z^{2} +2z  & \text{for } z \in(0,300/1024), \\
1.81 - 16z^{2} +4z & \text{for } z \in(300/1024, 800/1024), \\
4z-7.94 & \text{for } z \in(800/1024, 1),
\end{cases}
\end{align*}
where $z = t/T$.

\subsection{Results for Daubechies LA10 Wavelet LSW Processes}

In this section we report the results of the simulations when the Daubechies Least Asymmetric wavelet with 10 vanishing moments (LA10) is used to simulate the LSW processes. Table \ref{table-sup} reports the mean squared error comparison for the averaged spectral estimates, where the LSW process has been generated using the LA10 wavelet. The estimation error for spectrum 1 is almost identical comparing to ``None'' using our methodology with a trend present, while for spectrum 2 and 3 there is only a marginal increase. The results show that the method performs well irrespective of the wavelet that generates the LSW process.

\begin{table}[]
\centering
\begin{tabular}{c|c|c|c}
Trend               & Spectrum 1                 & Spectrum 2                 & Spectrum 3                 \\  \hline 
None                & $1.25  $                     & $3.32  $                       &$ 1.30 $               \\ 
Linear              & $1.29       $                & $4.09    $                  &$ 2.34   $               \\ 
Sine                & $1.29       $            & $4.09    $                 & $ 2.34 $                      \\ 
Logistic            & $1.29        $           & $4.09 $                    &$ 2.34   $                      \\ 
Piece. Quad. & $1.29             $       & $4.11   $                   &$ 2.34  $                     \\ 
\end{tabular}
\caption{Mean squared error, multiplied by $10^{3}$, comparing across the spectrum and trend scenarios.}
\label{table-sup}
\end{table}

Table \ref{tableLA10t1} reports the average mean squared error and standard deviation for the trend estimate where the 100 LSW processes have been generated using Gaussian innovations and the LA10 wavelet. Lastly, Table \ref{tableLA10t2} reports the average mean squared error and standard deviation for the trend estimate where the 100 LSW processes have been generated using exponential innovations and the LA10 wavelet. For both cases, a hard threshold of $\hat{\sigma}_{r,s} \sqrt{2 \log (1024)}$ is used. Again, the error is consistent across the various time series scenarios. We see that the trend estimate performs well no matter the generating wavelet of the process.

\begin{table}[]
\centering
\begin{tabular}{c|c|c|c|c}
\multirow{2}{*}{Trend} & \multirow{2}{*}{Spectrum} & \multicolumn{3}{c}{Method}    \\ \cline{3-5}
      &        & LSWT      & SWT        & VSMWT             \\ \hline
\multirow{3}{*}{Linear}     &  1      &  \textbf{0.019}   (0.013)       &   0.516 (0.117)       &   0.275 (0.091)                    \\
&2        & \textbf{0.038}  (0.024)      &   0.746 (0.083)       & 0.207 (0.0644)                   \\
      &  3      & \textbf{0.029}    (0.027)       & 0.494 (0.042)          &   0.171 (0.039)                \\ \hline
\multirow{3}{*}{Sine}     &  1      & \textbf{0.021}    (0.012)          & 0.547 (0.130)       &  0.287 (0.111)                  \\
&  2        &  \textbf{0.031}  (0.016)    & 0.752 (0.077)         &  0.214  (0.067)               \\
      & 3      &   \textbf{0.024}    (0.018)      &   0.485 (0.045)       & 0.156 (0.040)                    \\ \hline
\multirow{3}{*}{Logistic}     &  1      &  \textbf{0.022}    (0.013)        &  0.544 (0.112)      &  0.306 (0.107)                \\ 
&  2        & \textbf{0.042} (0.028)     & 0.764 (0.088)          & 0.214 (0.066)                \\
      &  3      & \textbf{0.028}  (0.025)       &  0.491 (0.038)       & 0.170 (0.041)                       \\ \hline
 \multirow{3}{*} {Piece. quad.}    &  1      & \textbf{0.025}   (0.014)      &  0.537 (0.127)        & 0.280 (0.110)                      \\
& 2        & \textbf{0.041}  (0.024) &  0.736 (0.076)        &  0.211 (0.080)                 \\
      &  3      & \textbf{0.034}    (0.028)         &   0.414 (0.043)     &  0.162 (0.042)            
\end{tabular}
\caption{Average mean squared error and standard deviation in brackets of trend estimate over 100 realisations generated using Gaussian innovations and the Daubechies LA10 wavelet.}
\label{tableLA10t1}
\end{table}

\begin{table}[]
\centering
\begin{tabular}{c|c|c|c|c}
\multirow{2}{*}{Trend} & \multirow{2}{*}{Spectrum} & \multicolumn{3}{c}{Method}    \\ \cline{3-5}
      &        & LSWT      & SWT        & VSMWT                \\ \hline
\multirow{3}{*}{Linear}     &  1      &  \textbf{0.027} (0.016)       &  0.547(0.140)   & 0.312 (0.107)                   \\
&2        & \textbf{0.050} (0.028)      & 0.793 (0.113)     & 0.337(0.096)                    \\
      &  3      & \textbf{0.040} (0.029)       & 0.529 (0.066)   &  0.311 (0.068                            \\ \hline
\multirow{3}{*}{Sine}     &  1      & \textbf{0.027} (0.014)         &   0.540 (0.138)     & 0.311  (0.121)                \\
&  2        &  \textbf{0.041} (0.021)    & 0.788 (0.118)         &  0.340  (0.103)                \\
      & 3      &   \textbf{0.036} (0.019)     &   0.542 (0.059)       &  0.308 (0.067)                    \\ \hline
\multirow{3}{*}{Logistic}     &  1      &  \textbf{0.026} (0.017)        & 0.544 (0.143)          &  0.315  (0.118)                 \\ 
&  2        & \textbf{0.046} (0.026)     &   0.801 (0.112)        & 0.254  (0.091)                      \\
      &  3      & \textbf{0.045} (0.030)      &  0.541 (0.053)       &  0.327 (0.059)                 \\ \hline
 \multirow{3}{*} {Piece. quad.}    &  1      & \textbf{0.032} (0.020)      &   0.515 (0.122)       &   0.296  (0.102)            \\
& 2        & \textbf{0.044} (0.028) &   0.774 (0.111)       & 0.347 (0.100)               \\
      &  3      & \textbf{0.042} (0.029)      &   0.473 (0.057)     &  0.319 (0.061)     
\end{tabular}
\caption{Average mean squared error and standard deviation in brackets of trend estimate over 100 realisations generated using exponential innovations and the Daubechies LA10 wavelet.}
\label{tableLA10t2}
\end{table}

\subsection{Effect of Over-Differencing}\label{over-diff-sec}

\tcr{In the main text, we recommended using a first difference for removing the trend. Here, we compare the estimation performance of the spectral estimate computed using first- and second-differences, to highlight the effect of over-differencing the time series. We perform the same simulations as in Section \ref{spec-est-sims} of the main text, and compute the averaged spectral estimate using both first- and second-differences. To compare the two results, we compute the relative mean squared error of the first-differenced estimate and second-differenced estimate, in each of the trend and spectrum scenarios. To control for boundary effects, we only compute the RMSE on the non-boundary wavelet coefficients. The results are reported in Table \ref{rel-mse}.  We see that over-differencing yields a higher MSE for spectral estimation. The RMSE is above 1 in all scenarios, showing that the first-difference estimator achieves a uniformly stronger performance.} 

\begin{table}[H]
\centering
\begin{tabular}{c|c|c|c}
Trend               & Spectrum 1                 & Spectrum 2                 & Spectrum 3                 \\ \hline
Linear              & 1.562                   & 1.251                    & 1.461        \\ 
Sine                & 1.562                     & 1.251                     & 1.461     \\ 
Logistic            & 1.561                  & 1.251                    &   1.461  \\ 
Piece. Quad. & 1.562                     & 1.251                      & 1.459   \\ 
\end{tabular}
\caption{Relative mean squared error comparison for the averaged spectrum estimate, across the spectrum and trend scenarios.}
\label{rel-mse}
\end{table}

\subsection{Further Trend Estimation Simulations}

Next, we assess the performance of the trend estimation procedure in the presence of non LSW-type error structures to highlight the versatility of our methodology. In particular, we simulate 100 realisations of time series using the previously defined trends, with (Gaussian) errors simulated from the following models:
\begin{enumerate}[(A)] 
\item Time-varying AR(2) model with parameters $\phi_{1} (z)  = 0.8\cos(1.5-\cos(4\pi z))$, $\phi_{2} (z) = -0.2 + 0.4z$, $z \in (0,1)$.
\item AR(1) model with parameter $\phi = 0.6$.
\item Time-varying AR(1) model with  parameter $\phi (z) = 0.7$ for $z \in (0,600/1024)$, $\phi (z) = -0.3$ for $z \in (600/1024,1)$.
\item ARMA(1,3) model with AR parameter $\phi = 0.4$ and MA parameters $\theta_{1} = 0.8$, $\theta_{2} = -0.3$, and $\theta_{3} = 0.4$.
\end{enumerate}
These scenarios represent some common stationary and nonstationary error structures. Model A is an AR(2) process with slowly-evolving autoregressive parameters, and is a variant of a process studied in \cite{von2000non}. Model C represents a piecewise stationary AR(1) process with a single changepoint in the autoregressive parameter. The results of these simulations are given in Table \ref{new-trend-sims1}, again with bold values showing the lowest reported mean squared error. We see that our methodology is able to perform well across the four scenarios, working well in both stationary and nonstationary second-order settings. It outperforms the other three methods in all of the trend and error scenarios. Overall, we see that our proposed methodology offers strong practical performance across a variety of time series models.

\begin{table}[]
\centering
\begin{tabular}{c|c|c|c|c}
\multirow{2}{*}{Trend} & \multirow{2}{*}{Model} & \multicolumn{3}{c}{Method}    \\ \cline{3-5}
      &        & LSWT      & SWT        & VSMWT             \\ \hline
\multirow{4}{*}{Linear}     &  A      & \textbf{0.140} (0.061)     &   0.376 (0.108)      &   0.482 (0.139)                   \\
&B        & \textbf{0.147} (0.040)     & 0.265 (0.070)       &   0.500 (0.086)         \\
      &  C      & \textbf{0.129} (0.048)      & 0.435 (0.080)      & 0.532 (0.095)      \\
      &D       & \textbf{0.216} (0.059)            & 0.349 (0.160)             &   0.840 (0.129)                    \\ \hline
\multirow{4}{*}{Sine}     &  A      &  \textbf{0.124} (0.069)     &  0.389 (0.134)    &   0.494  (0.130)             \\
&  B        & \textbf{0.134} (0.033)   & 0.241 (0.091)        & 0.499 (0.087)               \\
      & C      & \textbf{0.122}  (0.040)      & 0.421 (0.083)        & 0.516 (0.100)               \\ 
      &D       & \textbf{0.174} (0.045)            &  0.256 (0.146)            &  0.806 (0.135)                     \\ \hline
\multirow{4}{*}{Logistic}     &  A      & \textbf{0.129} (0.048)      & 0.370 (0.091)      &   0.474 (0.100)        \\ 
&  B        & \textbf{0.150}  (0.033)  &  0.267 (0.077)       & 0.495 (0.092)                \\
      &  C      & \textbf{0.129} (0.047)       & 0.431 (0.093)    & 0.531 (0.098)                 \\
      &D       &    \textbf{0.214} (0.051)         & 0.357 (0.147)             &       0.837 (0.143)                 \\ \hline
 \multirow{4}{*} {Piece. quad.}    &  A      &   \textbf{0.141}  (0.051) &  0.362 (0.094)     & 0.467 (0.115)              \\
& B        &  \textbf{0.141}  (0.040) &  0.275 (0.073)     & 0.500 (0.099)           \\
      &  C      & \textbf{0.131}  (0.046)       &  0.420 (0.084)    &  0.525 (0.098)      \\    
      &D       &  \textbf{0.221} (0.057)          & 0.332 (0.138)             &    0.834 (0.145)                  
\end{tabular}
\caption{Average mean squared error and standard deviation in brackets of trend estimate over 100 realisations generated from Models A -- D.}
\label{new-trend-sims1}
\end{table}

\subsection{Choice of $\alpha$ and $\beta$ Parameters}\label{alpha-beta-parameter-choice}

\tcr{In this section we investigate the effect that the choice of the parameters $\alpha$ and $\beta$ have on estimation quality, in order to inform the best choice of the parameters in practice. Recall that $\beta$ is the proportion of scales used when estimating the spectrum, whilst $\alpha$ is the number of scales used for estimating the LACV.}

\tcr{Firstly, the choice of $\beta$ was investigated. We simulated 1000 realisations of LSW processes, for three wavelets: the Haar, Daubechies EP4, and Daubechies EP10 wavelet. For each realisation the spectrum was given by $S_j = U_j 2^{\gamma j}$, where $U_j \sim \text{Uniform}(0,5)$ and $\gamma>0$ is varied. This provides a more complete picture of the impact of the choice of $\beta$ than only using the three previously defined spectra. The decaying spectral structure is used in order to control for the fact that estimation of the EWS is more difficult in coarse scales, and mirrors the assumption in Proposition \ref{A1-c-est} in the main text. In addition, we use a global mean to smooth the raw wavelet periodogram at each level, in order to focus on the effect of $\beta$, rather than the bin width parameter.}

\tcr{We calculated the smoothed EWS estimate, while varying the number of scales used to perform EWS correction, for $\gamma = 1/2, 1/4, 1/8$,  and time series length $T=512, 1024, 2048$. We computed the average mean squared error over the 1000 realisations between the estimated and true spectra. For sake of brevity, the results presented here are for the case where the time series has no trend: further simulations in each of the trend scenarios considered in the main text produced almost identical results. The results for $T=1024$ are shown in Figure \ref{sensitivity-plot1}. Note that the results in Figure \ref{sensitivity-plot1} suggest that the choice of $\beta$ is dependent on the wavelet used: the higher the number of vanishing moments, the higher the optimal choice of $\beta$. Across all (27) scenarios, the mean value of $\beta$ that produces the lowest average mean squared error in each scenario is 0.721. This motivates setting $J_1 =  \lfloor \frac{7}{10} \log_2 (T) \rfloor$ as a flexible rule-of-thumb. In practice, different values of $\beta$ can also be compared visually to assess the fit.}

\begin{figure}[]
\centering
\includegraphics[width =\textwidth]{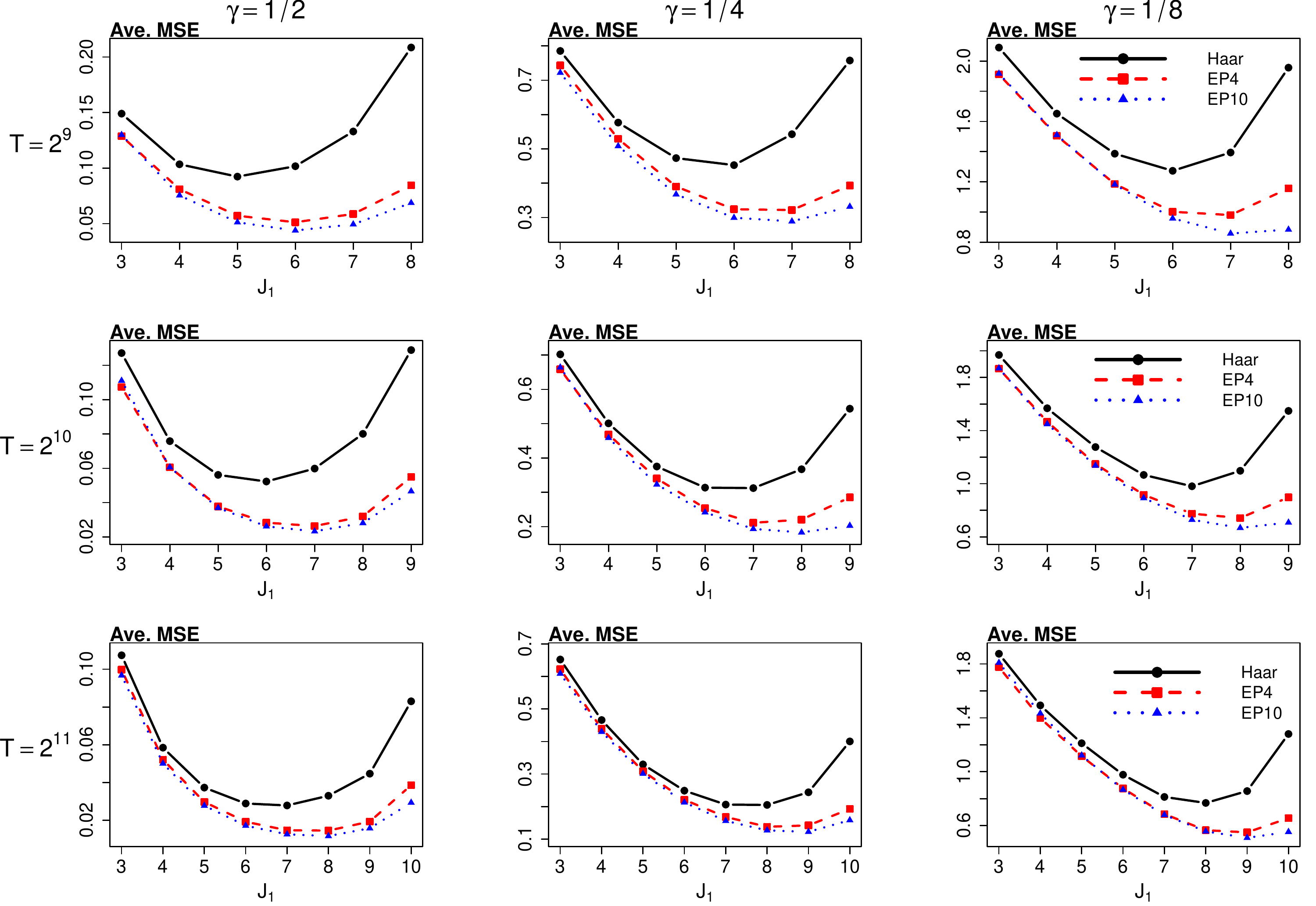}
\caption{Average mean squared error (MSE) for different choices of $J_1$, for varying values of $\gamma$ and $T$, using the Haar, EP4, and EP10 wavelets.}
\label{sensitivity-plot1}
\end{figure}

\tcr{Next, given the above choice of $\beta =  \lfloor \frac{7}{10} \log_2 (T) \rfloor$, we perform the same set of simulations as above, this time with varying $\alpha$, the number of scales used to estimate the LACV function. Note that necessarily $\alpha \leq \beta$. We compute the average MSE (averaged across lags) of the estimated autocovariance function, computed at the first $2^{J-3}$ lags where $J = \log_2 (T)$. The results of the simulation are shown in Figure \ref{sensitivity-plot2}. These show that taking $\alpha=\beta$ yields the best practical performance in all scenarios.}

\begin{figure}[]
\centering
\includegraphics[width =\textwidth]{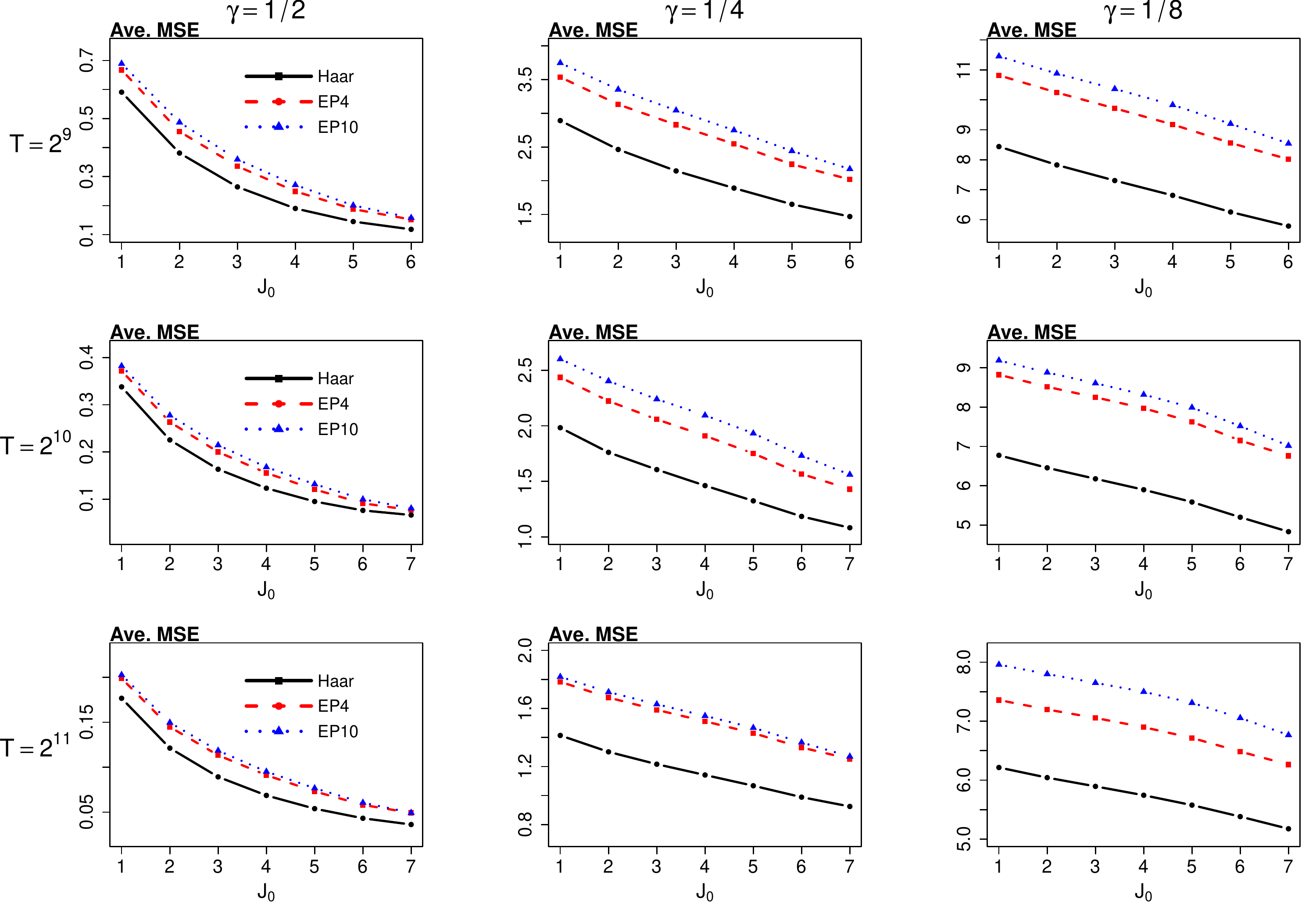}
\caption{Average mean squared error (MSE) for different choices of $J_0$, for varying values of $\gamma$ and $T$, using the Haar, EP4, and EP10 wavelets.}
\label{sensitivity-plot2}
\end{figure}

\subsection{Canadian Wave Height Data}

Here we examine publicly available wave height data for a location in the North Atlantic. The data we use are obtained from Fisheries and Oceans Canada, East Scotian Slop buoy ID C44137. The data measures wave heights collected at hourly intervals, from approximately mid-June 2005 to mid-May 2006, giving a time series of length $2^{13}=8192$, and is available in the \verb!changepoint! package in \verb!R! \citep{killick2014changepoint}. \cite{killick2012optimal} also perform nonstationary analysis on the series, using the first-differences of a longer version of the series to detect changepoints in variance only. The data are plotted in Figure \ref{wave-data} left, from which we see larger wave heights in the winter and smaller wave heights in the summer, as well as increased variability in the winter. Figure \ref{wave-data} right shows the first-differenced time series. \cite{killick2012optimal} take first-differences to remove the trend, and we also find this is sufficient.

For the spectral estimate, we use the Daubechies Least Asymmetric wavelet with 10 vanishing moments. For simplicity, the periodogram is smoothed using a running mean with bin width 512, corresponding to a time length of roughly 3 weeks. The estimate is shown in Figure \ref{wave-spec-plot}, where each level in the plot is scaled individually for clarity. We can see strong nonstationarity in the spectrum, with more variability during the winter months and less in the summer, as found in \cite{killick2012optimal} using a changepoint approach. 

\begin{figure}[]
\centering
\includegraphics[width =\textwidth]{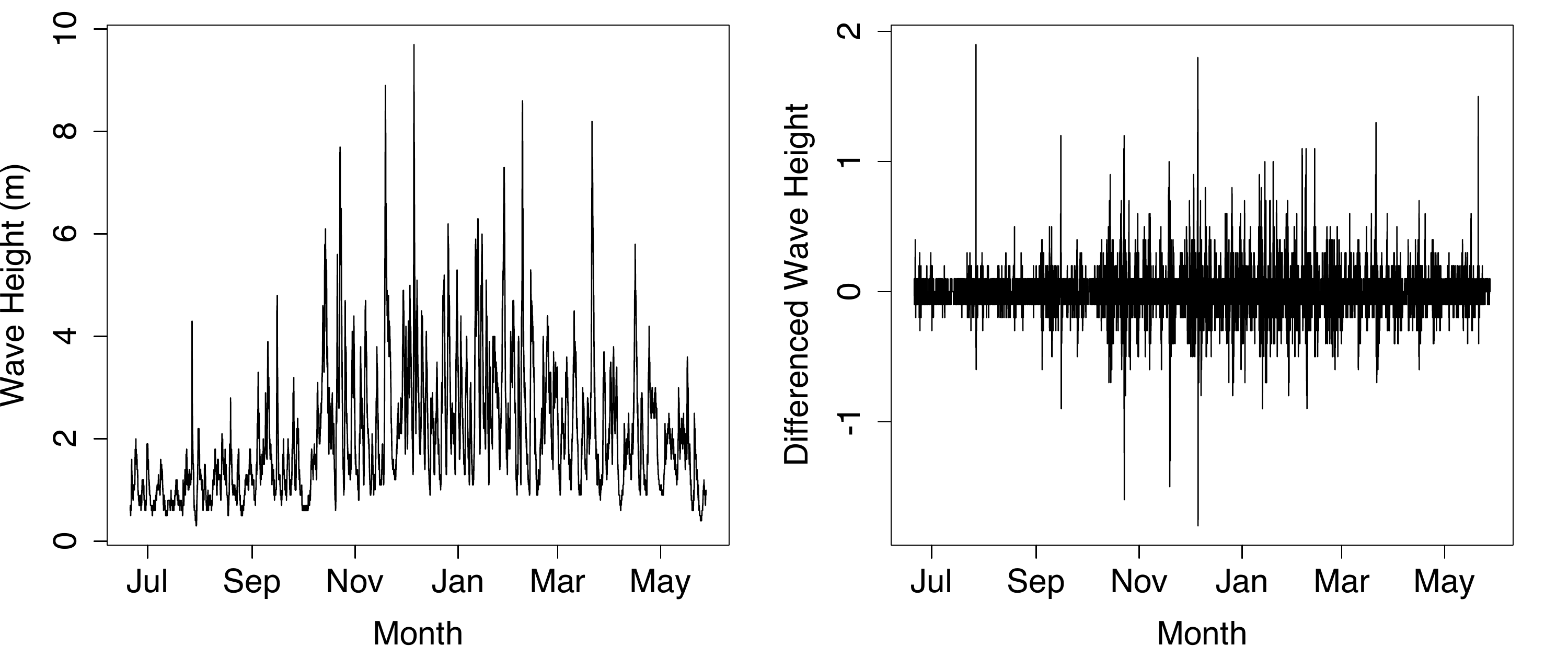}
\caption{Left: North Atlantic wave heights recorded hourly between June 2005 and May 2006. Right: first-differenced wave heights.}
\label{wave-data}
\end{figure}

\begin{figure}[]
\centering
\includegraphics[width = 0.9\textwidth]{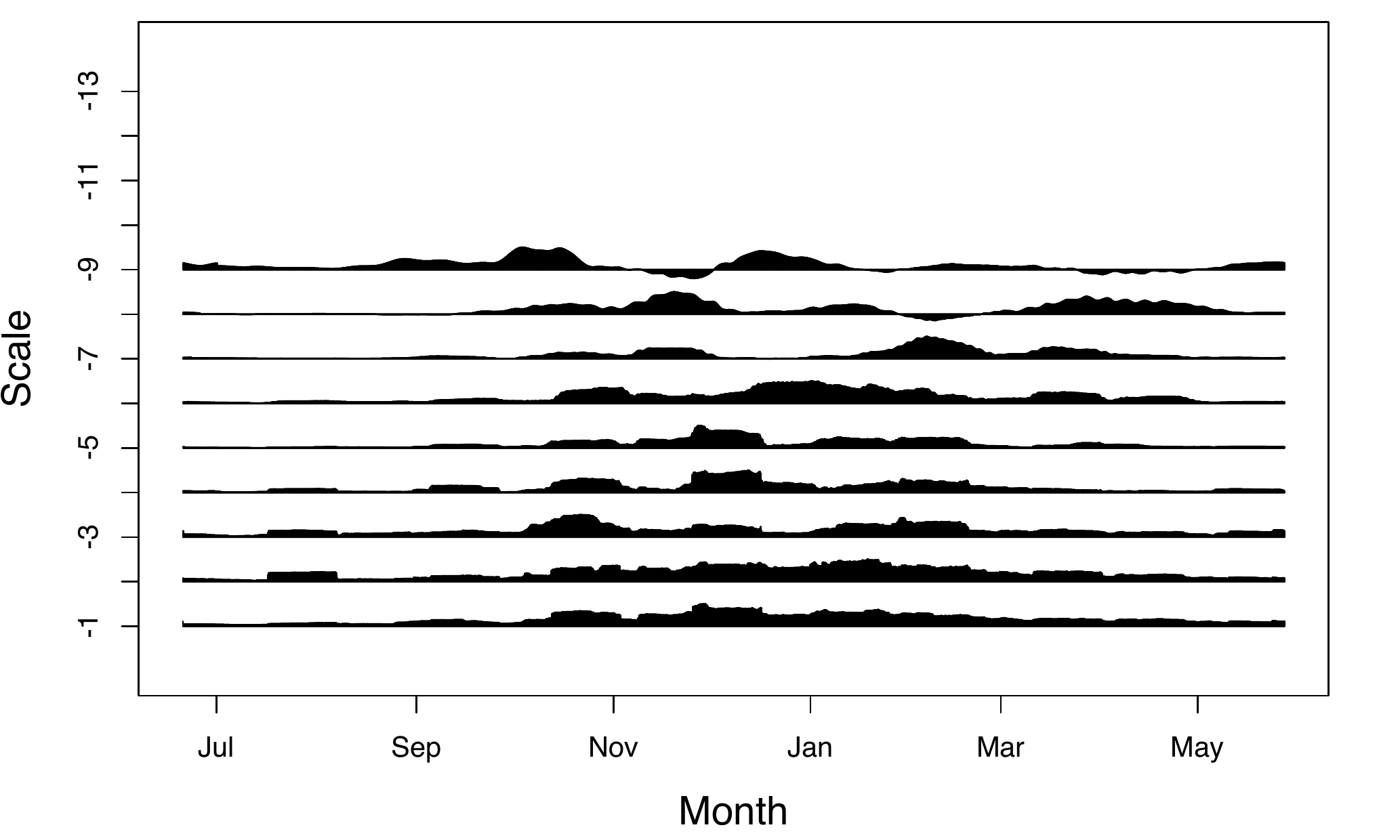}
\caption{Smoothed spectral estimate for the wave height data in Figure \ref{wave-data}, performed using first-differencing of the data.}
\label{wave-spec-plot}
\end{figure}

To estimate the trend of the series, we perform TI wavelet thresholding using the Daubechies Least Asymmetric wavelet with 4 vanishing moments. In line with previous discussion, we analyse the finest 9 scales of the series. We use a soft universal threshold of $\hat{\sigma}_{r,s} \sqrt{ 2 \log (8192)}$ to account for likely non-Gaussianity, where $\hat{\sigma}_{r,s}$ is calculated using the spectral estimate in Figure \ref{wave-spec-plot}. The trend estimate is shown in the solid line in Figure \ref{wave-trend-plot}. We see that the estimated trend function is relatively smooth with occasional sharp changes, with the mean wave height larger during the winter and smaller in the summer. There are several locations where perhaps the wavelet coefficients were not thresholded correctly, leading to the sharp changes. Using a scale-dependent smoothing of the raw wavelet periodogram is one possible way to remedy this.

\begin{figure}[]
\centering
\includegraphics[width = 0.9\textwidth]{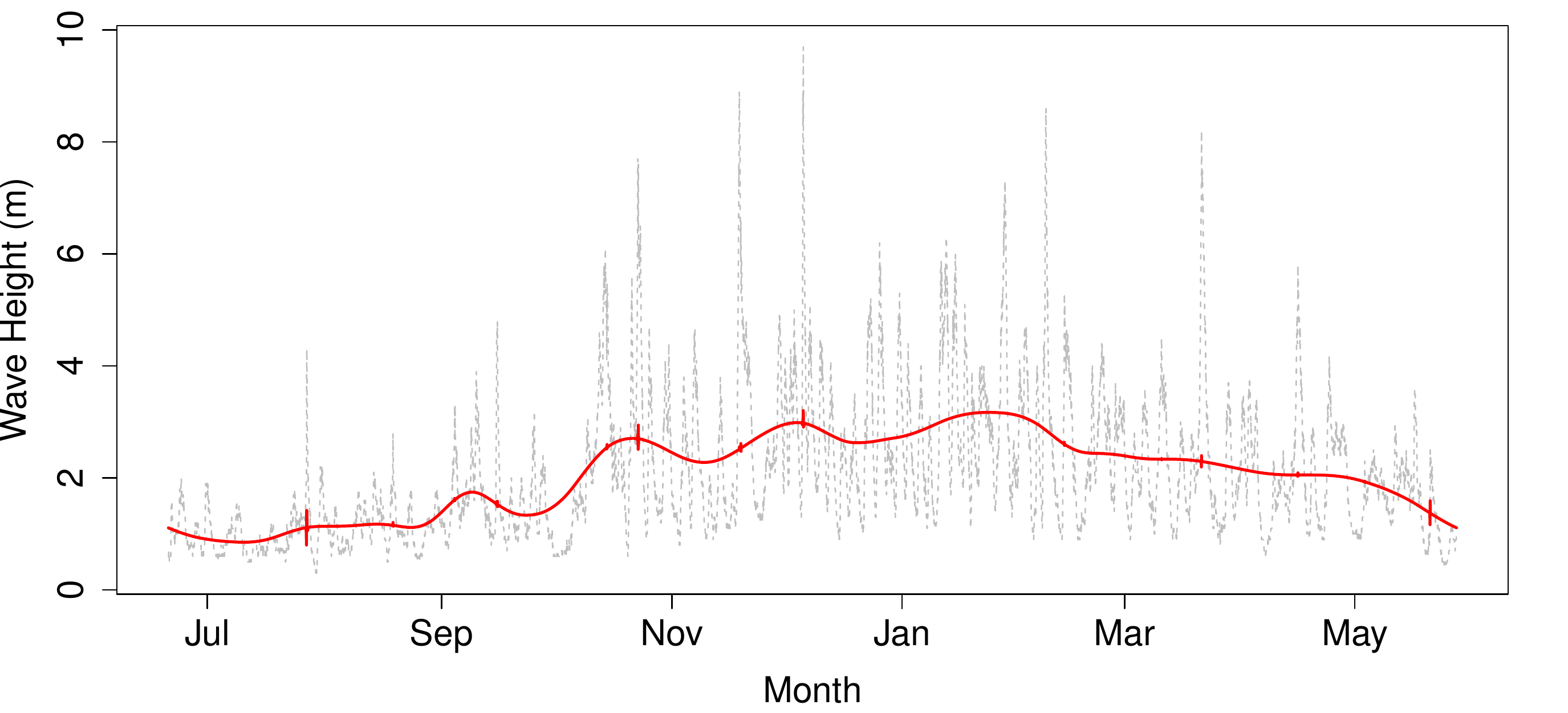}
\caption{Trend estimate for the North Atlantic wave data shown in solid line, with data shown in dashed line.}
\label{wave-trend-plot}
\end{figure}

Using the spectral estimate in Figure \ref{wave-spec-plot}, we calculated the local autocovariance estimate of the series. In Figure \ref{wave-var-plot}, we see in the solid line the estimate using our methodology for the local variance function of the time series. The nonstationary nature of the variance is clear to see, with the summer months containing low variability, while winter months display much higher variability. As one way of determining the performance of our method, we can compare the variance estimate with the estimate obtained by using the detrended wave height data. The detrended wave height data is obtained by subtracting the wavelet thresholding trend estimate from the data. We then perform standard LSW inference on this series, using the same parameters as in our approach. In dashed line, we see the local variance estimate obtained in this way. We see that the two estimates agree, which is reassuring on two counts: first that our spectral estimate obtained using the differenced data is accurate, and secondly that the trend estimate obtained was accurate.

\begin{figure}[]
\centering
\includegraphics[width = 0.9\textwidth]{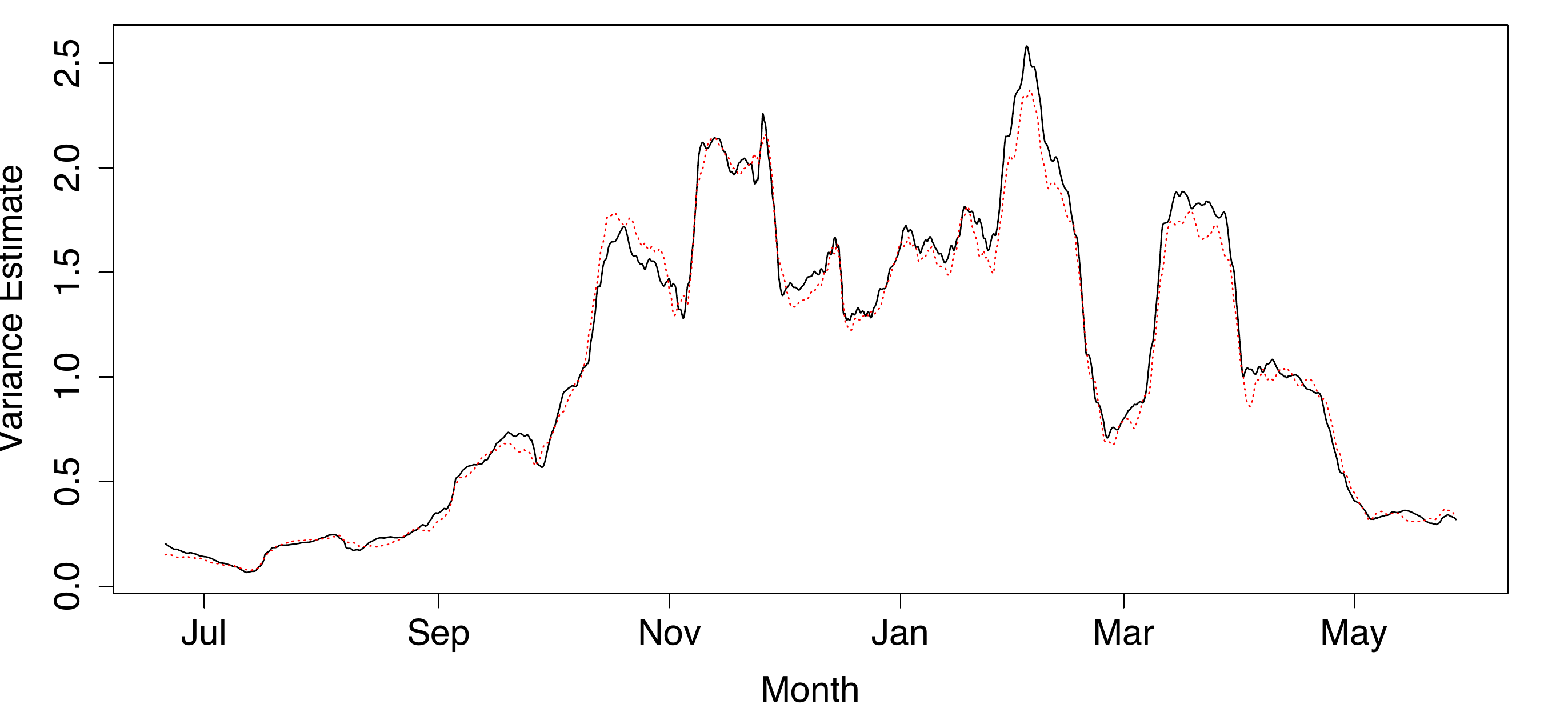}
\caption{Local variance estimate for the wave height data. Solid line: obtained using our methodology. Dashed line: obtained from the detrended data.}
\label{wave-var-plot}
\end{figure}

Finally, we plot the autocorrelation function across 4 time points, which highlights the second-order nonstationary nature of the data, and showing how the structure of the autocorrelation varies considerably over time. In Figure \ref{wave-acf-plot}, we see that estimate for the local autocorrelation function for the first observation recorded in the months of July 2005, October 2005, January 2006 and April 2006. These are shown in the solid, dashed, dotted, and dashed and dotted lines respectively. We see that in general the series exhibits strong autocorrelation, which we may expect due to the observations being recorded at hourly intervals. Furthermore, we note that the shape of the autocorrelation function changes across the four months plotted.

\begin{figure}[]
\centering
\includegraphics[width = 0.9\textwidth]{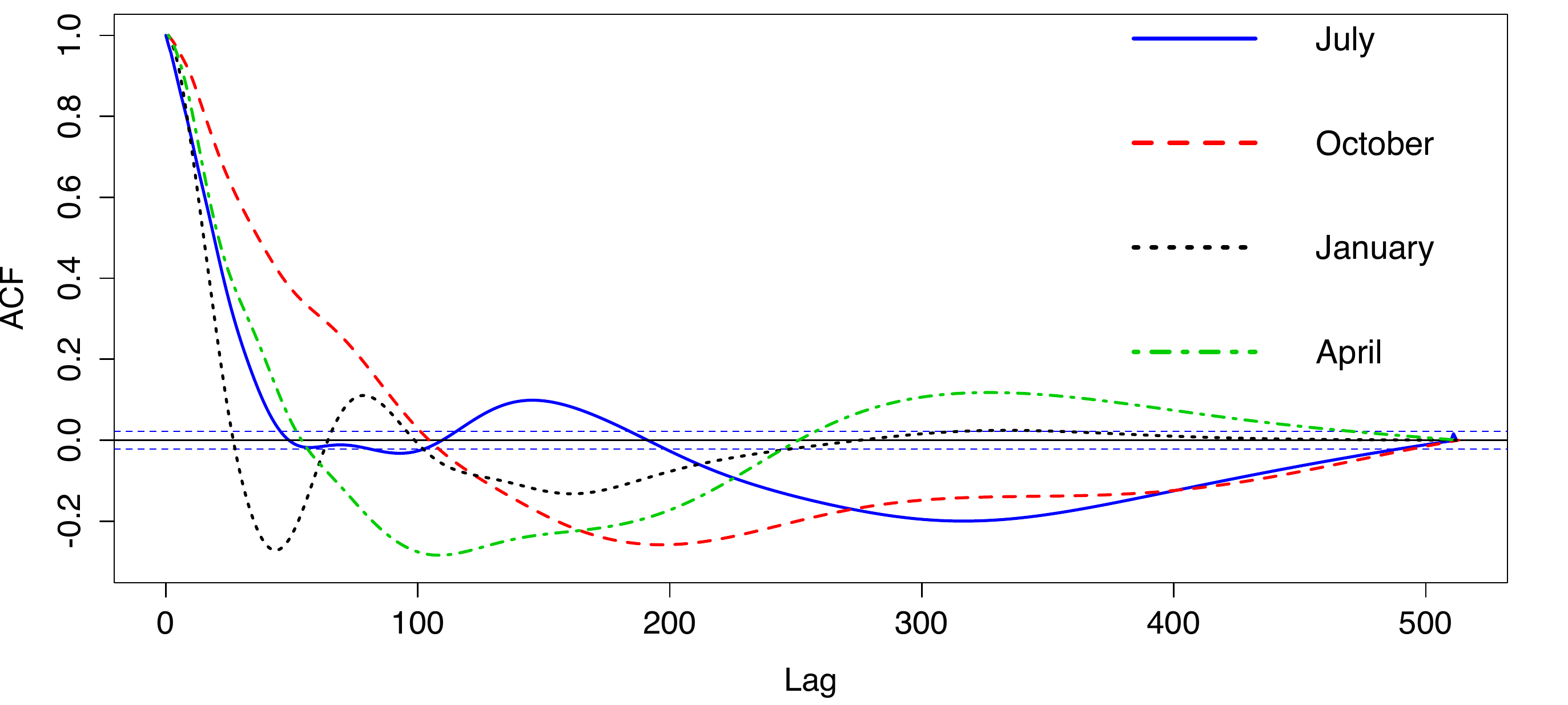}
\caption{Local autocorrelation function estimate for the wave height data, at 4 different time points. Solid: July, dashed: October, dotted: January, dotted and dashed: April.}
\label{wave-acf-plot}
\end{figure}

\section{Proofs of Results}\label{appB}

\subsection{Proof of Proposition \protect\ref{prop1}}

\tcr{The matrix $D^{1}_{J}$ is symmetric and positive semi-definite, since it can be expressed as a Gram matrix of vector inner products:
\begin{equation*}
D^1_J =  \langle \nabla \Psi_j , \nabla \Psi_l \rangle,   
\end{equation*}
where $\Psi_j$ is the autocorrelation wavelet vector at scale $j$, and $\nabla$ is the first-differencing matrix with diagonal entries $1$, above the main diagonal entries equal to $-1$, and all other entries equal to zero. The matrix $\nabla$ is invertible (with inverse given by the upper triangular matrix with non-zero entries equal to 1).  By \cite{nason2000wavelet} Theorem 1, the family of vectors $\{ \Psi_{j}\}_{j=-J}^{-1}$ is linearly independent. This implies that $\{ \nabla \Psi_{j} \}_{j=-J}^{-1}$ is also linearly independent, since this family is given by the invertible matrix transform of a family of linearly independent vectors. Therefore, $D^1_J$ is invertible since it is the Gram matrix of a set of linearly independent vectors.}

\subsection{Proof of Proposition \protect\ref{A1exp}} 
The expectation is given by
\begin{align*}
\mathbb{E} (\tilde{I}_{k}^{j}) &= \mathbb{E} \left[ \left( \sum_{t} \left[ \epsilon_{t} - \epsilon_{t-1}  + \mathcal{O} ( T^{-1}) \right]\psi_{j,k-t} \right)^{2} \right] \\[1ex]
&= \mathbb{E}  \left[ \left( \sum_{t} \epsilon_{t} \psi_{j,k-t} \right)^{2} \right] + \mathbb{E}  \left[ \left( \sum_{t} \epsilon_{t-1} \psi_{j,k-t} \right)^{2} \right] \\[1ex]
&\qquad - 2 \mathbb{E} \left( \sum_{s} \sum_{t} \epsilon_{t} \epsilon_{s-1} \psi_{j,k-t} \psi_{j,k-s}  \right)  + \mathcal{O} ( T^{-1}) \\[1ex]
&:= I + II + III  + \mathcal{O} ( T^{-1}),
\end{align*}
where the remainder term can come out of the inner bracket since, for fixed $j$, the sum is finite as the wavelet is compactly supported. Now we evaluate each expectation individually. The term $I$ is simply the expectation of the raw wavelet periodogram of the original LSW model. Hence,
\begin{equation*}
I = \sum_{l} A_{jl}S_{l} \left( \frac{k}{T} \right) + \mathcal{O} ( T^{-1}) .
\end{equation*}
Next, observe that the term $II$ is equal to $\mathbb{E} (I_{k-1}^{j})$, since $d_{j,k-1} = \sum_{s} X_{s} \psi_{j,k-1-s}$, and by substituting $s = t-1$, we obtain $\sum_{t} X_{t-1} \psi_{j,k-t} = \sum_{s} X_{s} \psi_{j,k-1-s}$. Next, setting $u = k-1$, we obtain
\begin{align*}
\mathbb{E} (I^{j}_{u}) = \sum_{l} \sum_{m} w_{lm}^{2} \left( \sum_{t} \psi_{l,m-t} \psi_{j,u-t}\right)^{2}.
\end{align*}
Now, following a similar argument to \cite{nason2000wavelet}, we obtain 
\begin{align*}
II &=  \sum_{l} S_{l} \left( \frac{u+1}{T} \right) \sum_{n} \sum_{t} \sum_{s} \psi_{j,-s} \psi_{j,-t} \psi_{l,n-s} \psi_{l,n-t} + \mathcal{O} (T^{-1})\\[1ex]
&= \sum_{l} A_{jl} S_{l} \left( \frac{k}{T} \right) + \mathcal{O} (T^{-1}).
\end{align*}
Finally,

%substitute $m = n+u$, and subtracting the $u$ from the wavelet subscripts since the sum is infinite, we obtain
%\begin{align*}
%II &= \sum_{l} \sum_{n} w_{l,n+u}^{2} \left( \sum_{t} \psi_{j,n+u-t} \psi_{j,u-t}\right)^{2}\\[1ex]
%&= \sum_{l} \sum_{n} \left[ S_{l}\left(\frac{n+u}{T} \right) + \mathcal{O}(T^{-1}) \right] \left( \sum_{t} \psi_{j,n-t} \psi_{j,-t}\right)^{2} \\[1ex]
%&= \sum_{l} \sum_{n} \left[ S_{l}\left(\frac{u+1}{T} \right) + \mathcal{O} \left((n-1)T^{-1}\right) \right] \left( \sum_{t} \psi_{j,n-t} \psi_{j,-t}\right)^{2}.
%\end{align*}
%The remainder term can come out of the inner brace because the number of terms in the wavelet product is finite because of the compact support and the Lipschitz constant summability. Hence,
%\begin{align*}
%II &=  \sum_{l} S_{l} \left( \frac{u+1}{T} \right) \sum_{n} \sum_{t} \sum_{s} \psi_{j,-s} \psi_{j,-t} \psi_{l,n-s} \psi_{l,n-t} + \mathcal{O} (T^{-1})\\[1ex]
%&= \sum_{l} A_{jl} S_{l} \left( \frac{k}{T} \right) + \mathcal{O} (T^{-1}),
%\end{align*}
%where the last equality follows from the last few lines of the proof of Proposition 4 from \cite{nason2000wavelet}. Finally,
\begin{align*}
III &= -2 \mathbb{E} \left( \sum_{s} \sum_{t} \epsilon_{t} \epsilon_{s-1} \psi_{j,k-t} \psi_{j,k-s}  \right)\\[1ex]
&= -2 \sum_{s} \sum_{t} \psi_{j,k-t} \psi_{j,k-s} \mathbb{E} (\epsilon_{t} \epsilon_{s-1} )\\[1ex] 
&= -2 \sum_{s} \sum_{t} \psi_{j,k-t} \psi_{j,k-s} \sum_{l} S_{l} \left( \frac{t+s-1}{2T} \right) \Psi_{l}(s-t-1) + \mathcal{O} (T^{-1}),
\end{align*}
since 
\begin{equation*}
\mathbb{E} (\epsilon_{t} \epsilon_{s-1}) = c \left( \frac{t+s-1}{2T}, s-t-1 \right) + \mathcal{O}(T^{-1}).
\end{equation*}
The remainder of the proof is similar to that of the proof of Theorem 1 of \cite{nason2013test}. Let $u = k-t$ and $v = k-s$, substituting into the above to obtain
\begin{align*}
III &= -2 \sum_{l} \sum_{u} \sum_{v} \psi_{j,u} \psi_{j,v} S_{l} \left(\frac{k}{T} - \frac{u+v-1}{2T} \right) \Psi_{l}(u-v-1) + \mathcal{O} (T^{-1}) \\[1ex]
&= -2 \sum_{l} \sum_{u} \sum_{v} \psi_{j,u} \psi_{j,v}\left[ S_{l} \left(\frac{k}{T} \right) + \mathcal{O} \left( \frac{ |u+v-1|}{2T} \right) \right] \Psi_{l}(u-v-1) + \mathcal{O} (T^{-1}) .
\end{align*}
The remainder term in the above expression can be shown to be $\mathcal{O}(T^{-1})$ (A.2.1 in \cite{nason2013test}). Hence, we are left with 
\begin{align*}
III &= -2 \sum_{l} S_{l} \left(\frac{k}{T} \right) \sum_{u} \sum_{v} \psi_{j,u} \psi_{j,v} \Psi_{l}(u-v-1) + \mathcal{O} (T^{-1}).
\end{align*}
Finally, substituting $r= u-v$, we obtain
\begin{align*}
C &= -2 \sum_{l} S_{l} \left(\frac{k}{T} \right) \sum_{u} \sum_{r} \psi_{j,u} \psi_{j,u-r} \Psi_{l}(r-1) + \mathcal{O} (T^{-1})\\[1ex]
&= -2 \sum_{l} S_{l} \left(\frac{k}{T} \right) \sum_{r} \Psi_{j} (r) \Psi_{l}(r-1) + \mathcal{O} (T^{-1}) \\[1ex]
&= -2 \sum_{l} A^{1}_{jl} S_{l} \left(\frac{k}{T} \right)  + \mathcal{O} (T^{-1}).
\end{align*}
Hence, we obtain
\begin{align*}
\mathbb{E} ( \tilde{I}_{k}^{j}  ) &= I + II + III \\
&= \sum_{l} A_{jl} S_{l} \left(\frac{k}{T} \right)  + \sum_{l} A_{jl} S_{l} \left(\frac{k}{T} \right)  -2 \sum_{l} A^{1}_{jl} S_{l} \left(\frac{k}{T} \right)  + \mathcal{O} (T^{-1}) \\[1ex]
&= 2 \sum_{l} \left(A_{jl} - A^{1}_{jl} \right) S_{l} \left(\frac{k}{T} \right) + \mathcal{O} (T^{-1}) .
\end{align*}
For the variance part, we follow the same argument as in the proof of Proposition 4 of \cite{nason2000wavelet}. The wavelet coefficients of the differenced series are asymptotically Gaussian. Hence, the wavelet periodograms are asymptotically scaled $\chi^{2}$-distributed, and hence the variance is asymptotically proportional to the expectation squared.

\subsection{Proof of Theorem \protect\ref{invertibility1}} 

(Proof for the Haar wavelet). The proof follows the same strategy to that of the proof of Theorem 2 in \cite{nason2000wavelet}. We show that there exists $\delta >0$ such that $\lambda_{\text{min}} (P) \geq \delta$, by using the following property from Toeplitz matrix theory. Let $T$ be Toeplitz (and Hermitian) with elements $\{t_{0}, t_{1}, \ldots \}$. Let $f(z) = \sum_{n=-\infty}^{\infty} t_{n} z^{n}$ for $z \in \mathbb{C}$ be the symbol of the operator associated with $T$. If $\sum_{n} |t_{n} | < \infty$, then $f(z)$ is analytic in the open unit disc $D$ in the complex plane and continuous in the closed unit disc $\Delta = D \cup S$, where $S$ denotes the unit circle. The spectrum $\Lambda$ of the (Laurent) operator $T$ is $\Lambda(T) = f(S)$. If $T$ is symmetric then an estimate of the smallest eigenvalue of $T$ is 
\begin{equation*}
\min_{|z|=1} \{ f(z) \} = \min_{|z|=1} \left\{ t_{0} + 2 \text{Re} \left( \sum_{n=1}^{\infty} t_{n} z^{n} \right) \right\} .
\end{equation*}
(\cite{reichel1992eigenvalues}, theorem 3.1 (i)). For ease of notation, indices will now run from $1$ to $\infty$ instead of $-1$ to $-\infty$. Using straightforward algebra, we can derive explicit formulae for the entries of $P$, which are given by the following lemma.
\begin{lemma}
In the case of the Haar wavelet, the elements of the matrix $P$ are given by
\begin{equation*}
P_{jj} = 10,
\end{equation*}
\begin{equation*}
P_{j,j+m} = 6 \times 2^{-m/2}, \text{ for } l = j+m, m>0.
\end{equation*}
\end{lemma}
\begin{proof}
In general, we can derive the discrete autocorrelation wavelets $\Psi_{j} (\tau)$ by discretising the continuous wavelet autocorrelation function $\Psi (\tau)$, using the relationship
\begin{equation*}
\Psi_{j} (\tau) = \Psi \left(  \frac{ | \tau | }{ 2^{j}} \right).
\end{equation*}
For Haar wavelets, the continuous wavelet autocorrelation function is given by
\begin{align*}
\Psi (\tau) &= \int_{-\infty}^{\infty} \psi_{H} (x) \psi_{H} (x - \tau) dx  = \begin{cases}
1 - 3 | \tau | & \text{for } | \tau | \in [0, 1/2], \\
| \tau | - 1 & \text{for } | \tau | \in (1/2,1],
\end{cases}
\end{align*}
where $\psi_{H} (x)$ is the Haar mother wavelet. The discretisation formula holds for $\tau = -(2^{j}-1), \ldots, 0, \ldots , (2^{j}-1)$, and is equal to zero for all other values of $\tau$. By \cite{nason2000wavelet}, we have that the elements of $A$ are given by
\begin{align*}
A_{jj} &= \frac{1}{3} 2^{j} + \frac{5}{3} 2^{-j}, \qquad A_{jl} = 2^{2j-l-1} + 2^{-l} , \quad l > j >0.
\end{align*}
We must therefore calculate the elements of $A^{1}$, from which we can obtain the elements of $D^{1}$, and hence $P$. Note that replacing the $(\tau -1)$ terms with $(\tau+1)$ results in an equivalent definition of the operators $A^{1}$ and $D^{1}$. When $l=j$, we have that
\begin{align*}
A^{1}_{jj} &= \sum_{\tau = -(2^{j}-1)}^{2^{j}-1}\Psi_{j} (\tau)  \Psi_{j} (\tau+1) \\[1ex]
& = 2 \left[ \sum_{\tau = 1}^{2^{j-1}-1} \Psi_{j} (\tau) \Psi_{j} (\tau+1)+ \sum_{\tau = 2^{j-1}+1}^{2^{j}-1}  \Psi_{j} (\tau)  \Psi_{j} (\tau+1)  \right] -2^{-j-1} + \frac{1}{4} - 2^{-j-1} + \frac{1}{4}\\[1ex]
&= 2 \left[ \sum_{\tau = 1}^{2^{j-1}-1}  \left( 1 - \frac{3\tau}{2^{j}} \right)  \left( 1 - \frac{3(\tau+1)}{2^{j}} \right) + \sum_{\tau = 2^{j-1}+1}^{2^{j}-1} \left( \frac{\tau}{2^{j}} - 1 \right)  \left( \frac{\tau+1}{2^{j}} - 1 \right)   \right] -2^{-j} + \frac{1}{2} \\[1ex]
& = \frac{1}{3} 2^{j} -\frac{10}{3} 2^{-j}.
\end{align*}
Hence $D^{1}_{jj} = 2A_{jj} - 2A^{1}_{jj} = 10 \times 2^{-j}$, and therefore $P_{jj} = 10$. When $j\neq l$, without loss of generality assume $l>j$. We have that 
\begin{align*}
A^{1}_{jl} &= \sum_{\tau = -(2^{j}-1)}^{2^{j}-1}\Psi_{j} (\tau)  \Psi_{l} (\tau+1) = \sum_{\tau = -(2^{j}-1)}^{2^{j}-1} \Psi_{j} (\tau) \left( 1 - \frac{3|\tau+1|}{2^{l}} \right) \\[1ex]
&= 1 - 3 \times 2^{-l} +  \sum_{\tau = -(2^{j}-1)}^{-1} \Psi_{j} (\tau) \left( 1 - \frac{3|\tau+1|}{2^{l}} \right) + \sum_{\tau = 1}^{ 2^{j}-1} \Psi_{j} (\tau) \left( 1 - \frac{3(\tau+1)}{2^{l}} \right),
\end{align*}
since the value of $|\tau+1|$ is always less than or equal to $1/2$, and so we use the $1-3|\tau|$ part of the autocorrelation wavelet formula. The first sum is equal to
\begin{align*}
& \sum_{\tau = -(2^{j}-1)}^{ -2^{j-1}-1} \Psi_{j} (\tau) \left( 1 + \frac{3(\tau+1)}{2^{l}} \right) + \sum_{\tau = -2^{j-1}}^{ -1} \Psi_{j} (\tau) \left( 1 + \frac{3(\tau+1)}{2^{l}} \right) \\[1ex]
&= \sum_{\tau = -(2^{j}-1)}^{ -2^{j-1}-1} \left( -1 - \frac{\tau}{2^{j}} \right)  \left( 1 + \frac{3(\tau+1)}{2^{l}} \right) + \sum_{\tau = -2^{j-1}}^{ -1}  \left( 1 + \frac{3\tau}{2^{j}} \right)  \left( 1 + \frac{3(\tau+1)}{2^{l}} \right) \\[1ex]
&= \left( -2^{j-l-2} - 2^{j-3} + 2^{-l-1} + \frac{1}{4} - 2^{j-l-1} + 2^{2j-l-2} \right) +  \left( 3 \times 2^{j-l-2} + 2^{j-3} -3 \times 2^{-l-1} -\frac{3}{4} \right)  \\[1ex]
& = 2^{2j-l-2} -2^{-l} - \frac{1}{2}. 
\end{align*}
The second sum is given by
\begin{align*}
& \sum_{\tau = 1}^{ 2^{j-1}} \Psi_{j} (\tau) \left( 1 - \frac{3(\tau+1)}{2^{l}} \right) + \sum_{\tau = 2^{j-1}+1}^{ 2^{j}-1} \Psi_{j} (\tau) \left( 1 - \frac{3(\tau+1)}{2^{l}} \right) \\[1ex]
&= \sum_{\tau = 1}^{ 2^{j-1}} \left( 1 - \frac{3 \tau}{2^{j}} \right) \left( 1 - \frac{3(\tau+1)}{2^{l}} \right) + \sum_{\tau = 2^{j-1}+1}^{ 2^{j}-1}  \left( \frac{\tau}{2^{j}} - 1 \right) \left( 1 - \frac{3(\tau+1)}{2^{l}} \right) \\[1ex]
&= \left( 2^{j-3} + 3\times 2^{-l} -\frac{3}{4}  \right) + \left( 2^{2j-l-2} -2^{j-3} -2^{-l} + \frac{1}{4} \right) \\
&= 2^{2j-l-2} +2^{1-l} -\frac{1}{2} .
\end{align*}
From this, we finally obtain that 
\begin{align*}
A_{jl}^{1} &= \left(2^{2j-l-2} -2^{-l} - \frac{1}{2} \right) + \left( 2^{2j-l-2} +2^{1-l} -\frac{1}{2} \right) + \left( 1 - 3 \times 2^{-l} \right) \\[1ex]
&= 2^{2j-l-1} -2^{1-l}.
\end{align*}
Therefore, we have that $D^{1}_{jl} = 2A_{jl} - 2 A^{1}_{jl} =  6 \times 2^{-l}$, from which the required form for $P$ follows.
\end{proof}
Returning to the main proof, we have that $P$ is symmetric, however the formula above only refers to the upper triangular portion of the matrix $P$. Now, $P$ is a Toeplitz matrix with $p_{0} = 10$ and $p_{m} = 6 \times 2^{-m/2}$. We can now show that $\lambda_{\text{min}} (P) \geq \delta >0$. Substituting in the formula for the symbol of the symmetric Toeplitz $P$, we obtain
\begin{align*}
\min_{|z|=1} \{ f(z) \} &= \min_{|z|=1} \left\{ 10 + 2 \text{Re} \left( \sum_{n=1}^{\infty} 6 \times 2^{-n/2} z^{n} \right) \right\} \\[1ex]
&= 10 + 12 \min_{|z|=1} \left\{  \text{Re} \left( \frac{-\sqrt{2}z}{\sqrt{2}z -2}  \right)  \right\}\\[1ex]
&= 10 + 12 \min_{|z|=1} \left\{  \frac{ 2\sqrt{2} \text{Re} (z) - 2 }{6-4\sqrt{2} \text{Re} (z)}  \right\}.
\end{align*}
This function is strictly monotonically increasing on $-1 \leq \text{Re} (z) \leq 1$, therefore it follows that $\min_{|z|=1} \{ f(z) \} = f(-1)$. Hence,
\begin{equation*}
\lambda_{\text{min}}  (P) \geq f(-1) = 10 + \frac{12(-2\sqrt{2}-2)}{6+4\sqrt{2}} = \frac{18+8\sqrt{2}}{3+2\sqrt{2}} > 0.
\end{equation*}
(Proof for the Shannon wavelet). Note that the indices now run over the negative integers. We can compute the entries of $P$ using a variant of the Fourier domain formula shown in Equation (\ref{parceval}) below, which is a consequence of Parseval's relation.
\begin{equation}\label{parceval}
A_{jl} = \sum_{\tau} \Psi_{j} (\tau) \Psi_{l} (\tau) = \frac{1}{2\pi} \int \hat{\Psi}_{j} (\omega)  \hat{\Psi}_{l} (\omega) d \omega ,
\end{equation}
where $\hat{\Psi}_{j} (\omega)$ denotes the Fourier transform of $\Psi_{j} (\tau)$, which is equal to the squared modulus of the Fourier transform of the non-decimated wavelet coefficients, $\hat{\psi_{j} }(\omega )$. Explicitly,
\begin{equation}\label{autofourier}
\hat{\Psi}_{j} (\omega) = \left| \hat{\psi}_{j} (\omega) \right|^{2} = 2^{-j} \left| m_{1}\left( 2^{-(j+1)} \omega \right) \right|^{2} \prod_{l=0}^{-(j+2)} \left| m_{0} \left( 2^{l} \omega \right) \right|^{2},
\end{equation}
where $m_{0} (\omega) = 2^{-1/2} \sum_{k} h_{k} \exp (-i \omega k)$ is the transfer function; $\{ h_{k} \}$ is is the high-pass quadrature mirror filter with $\sum_{k} h_{k}^{2}=1$ and $\sum_{k} h_{k} = 2^{1/2}$; and $|m_{1} (\omega)|^{2} = 1 - |m_{0} (\omega) |^{2}$. 

The formula for the Fourier transform of the non-decimated wavelets given in Equation (\ref{autofourier}) and the corresponding formulae for $m_{0}( \omega )$ and $m_{1}( \omega )$ for the Shannon wavelet can be found using the Fourier transform of the continuous time mother and father wavelets, which can be found in \cite{chui1997wavelets}, pages 46 and 64. Define the set $C_{j}$, for $j<0$, to be
\begin{equation*}
C_{j} = \left[ -\frac{\pi}{2^{-j-1}}, - \frac{\pi}{2^{-j}} \right] \cup  \left[ \frac{\pi}{2^{-j}},  \frac{\pi}{2^{-j-1}} \right].
\end{equation*}
As in \cite{nason2000wavelet}, the Fourier transform of the non-decimated Shannon wavelets is given by
\begin{equation*}
\hat{\psi_{j}} (\omega) = - 2^{-j/2} \exp (-2^{-j-1} i \omega ) \mathbb{I}_{C_{j}} (\omega ),
\end{equation*}
where $\mathbb{I}_{C_{j}}$ is the indicator function on the set $C_{j}$. From this the Fourier transform of the autocorrelation wavelets can be obtained as 
\begin{equation*}
\hat{\Psi}_{j} (\omega) = \left| \hat{\psi}_{j} (\omega )\right|^{2} = 2^{-j} \mathbb{I}_{C_{j}} (\omega) .
\end{equation*}
Using Parseval's relation and the shifting property of the Fourier transform, we have that
\begin{equation*}
D^{1}_{jl} = 2 \sum_{\tau} \Psi_{j} (\tau) (\Psi_{l} (\tau) - \Psi_{l}(\tau + 1) ) = \frac{1}{2\pi}  \int \hat{\Psi}_{j} (\omega)  \hat{\Psi}_{l} (\omega) d \omega - \frac{1}{2\pi}  \int e^{-i \omega} \hat{\Psi}_{j} (\omega)  \hat{\Psi}_{l} (\omega) d \omega.
\end{equation*}
The first term in the sum is exactly the entries of the original $A$-matrix which are given by $A_{jj} = 2^{-j}$ for $j<0$, $A_{jl} = 0 $ for $j \neq l$. When $j \neq l$, the second term is equal to zero since the supports of different $\hat{\Psi}_{j} (\omega)$ do not overlap. Thus, the matrix is diagonal, and when $j=l$, the second term is given by
\begin{align*}
-\sum_{\tau} \Psi_{j} (\tau) \Psi_{j}(\tau + 1)  &= - \frac{1}{2\pi}  \int e^{-i \omega} \hat{\Psi}_{j} (\omega)^{2} d \omega \\
&= -\frac{2^{-2j-1}}{\pi} \int e^{-i \omega} \mathbb{I}_{C_{j}} (\omega) d \omega \\[1ex]
&= -\frac{2^{-2j-1}}{\pi} \left(\int_{-2^{j+1}\pi}^{-2^{j}\pi} e^{-i \omega} d \omega + \int_{2^{j}\pi}^{2^{j+1}\pi} e^{-i\omega}   d \omega \right) \\[1ex]
&= -\frac{2^{-2j-1}}{\pi} \left( \left[i e^{-i \omega} \right]_{-2^{j+1}\pi}^{-2^{j}\pi} +  \left[i e^{-i \omega} \right]_{2^{j1}\pi}^{2^{j+1}\pi}  \right)\\
& = -\frac{2^{-2j-1}i}{\pi} \left( \exp(2^{j} \pi i) -\exp(2^{j+1} \pi i)+\exp(-2^{j+1} \pi i) -\exp(-2^{j} \pi i)     \right).
\end{align*}
Now, expanding the complex exponential terms into trigonometric functions, we obtain
\begin{align*}
\begin{split}
-\sum_{\tau} \Psi_{j} (\tau) \Psi_{j}(\tau + 1)  &= -\frac{2^{-2j-1}i}{\pi} \left\{  \cos(2^{j}\pi ) + i\sin (2^{j}\pi) -\cos(2^{j+1}\pi ) - i\sin (2^{j+1}\pi) \right. \\
& \left. + \cos(-2^{j+1}\pi ) + i\sin (-2^{j+1}\pi) - \cos(-2^{j}\pi ) - i\sin (-2^{j}\pi)\right\}.
\end{split}
\end{align*}
After much simplification in which the cosine terms will vanish, we obtain
\begin{equation*}
-\sum_{\tau} \Psi_{j} (\tau) \Psi_{j}(\tau + 1) = \frac{2^{-2j}}{\pi} \left( \sin(2^{j} \pi) - \sin(2^{j+1} \pi) \right).
\end{equation*}
Hence, the diagonal entries of the matrix $D^{1}$ are given by 
\begin{equation*}
D^{1}_{jj} = 2^{-j+1} + \frac{2^{-2j+1}}{\pi} \left( \sin(2^{j} \pi) - \sin(2^{j+1} \pi) \right).
\end{equation*}
Hence, $P_{jj} = 2^{-2j+1} + 2^{-3j+1}\left[ \sin(2^{j} \pi) - \sin(2^{j+1} \pi) \right]/ \pi$, while the off-diagonal terms are zero. Next, approximate the diagonal terms using a Maclaurin series:
\begin{align*}
P_{jj} & \approx 2^{-2j+1} + \frac{2^{-3j+1}}{\pi} \left( 2^{j}\pi - \frac{(2^{j}\pi)^{3}}{6} - 2^{j+1}\pi + \frac{(2^{j+1}\pi)^{3}}{6} \right) \\
&= 2^{-2j+1} +2^{-2j+1} - 2^{-2j+2} - \frac{\pi^{2}}{3} + \frac{8\pi^{2}}{3} \\
&=\frac{7 \pi^{2}}{3}.
\end{align*}
Therefore, the matrix is diagonal, with all diagonal entries being uniformly bounded away from 0, with $P_{jj} \geq P_{-1,-1} \approx 13.09$ for all $j$. Hence, we have shown that  there exists some $\delta>0$ such that $\lambda_{\text{min}} (P) \geq \delta$. 

\subsection{Proof of Theorem \protect\ref{smooth_theorem4}}

As $T \rightarrow \infty$, the eigenvalues of $D^{1}$ tend to zero and hence, when viewed as an operator acting on the sequence space $\ell^{2}(\mathbb{N})$, its inverse is unbounded. In the locally stationary Fourier time series setting, a similar relationship is found, as given in Equation (5) of \cite{roueff2011locally}. Loosely speaking, the original spectrum at frequency $\omega$ is related to the differenced one through multiplication of the term $|1-e^{-i \omega}|^{-2}$. As $\omega \rightarrow 0$ (corresponding to low frequencies) the equation blows up. This mirrors our scenario, where, as the correction matrix grows in size -- and we consider coarse-scale (low frequency) behaviour -- the inverse matrix norm becomes larger. 

We can account for the unboundedness of the inverse of $D^{1}$ by using a rescaling of the LSW process itself. Consider the auxiliary process $\epsilon_{t} = \sum_{j,k}  \tilde{w}_{jk} \tilde{\psi}_{j,k-t} \xi_{lm}$, where $\tilde{w}_{jk} = 2^{j/4} w_{jk}$ and $\tilde{\psi}_{j,k-t} = 2^{-j/4} \psi_{j,k-t}$. Then, the expectation of the raw wavelet periodogram (with respect to the rescaled wavelet) is given by
\begin{equation}\label{beans}
\mathbb{E} (\tilde{I}_{k}^{j} ) = \sum_{l} P_{jl} \tilde{S}_{l} (k/T) + \mathcal{O} (T^{-1}),
\end{equation}
where $\tilde{S}_{j} (k/T) = 2^{j/2} S_{j} (k/T)$ and $P_{jl} = 2^{-j/2} D^{1}_{jl} 2^{-l/2}$. We can therefore use the (bounded) $P$-inverse matrix to correct the smoothed, rescaled periodogram, and then multiply by $2^{-j/2}$, since ${S}_{j} (k/T) = 2^{-j/2} \tilde{S}_{j} (k/T)$.  To determine the appropriate threshold for the wavelet-based estimator, we use the following lemma:
\begin{lemma}\label{smooth_theorem1}
For a Gaussian trend LSW process and using a wavelet ${\psi}'$ of bounded variation, the wavelet coefficients $\hat{v}_{rs}^j$, with $2^r = o (T)$, obey uniformly in $s$,
\begin{equation*}
\mathbb{E} ( \hat{v}_{rs} ) - \int_0^1 \sum_n P_{jn} \tilde{S}_n (z) {\psi}'_{rs} (z) dz = \mathcal{O} \left( 2^{r/2} T^{-1} \right),
\end{equation*}
and
\begin{equation*}
\text{Var} ( \hat{v}_{rs} ) = 2 T^{-1} \int_0^1 \left( \sum_n P_{jn} \tilde{S}_n (z) \right)^2 {\psi}'^2_{rs} (z) dz + \mathcal{O} \left( 2^{r} T^{-2} \right).
\end{equation*}
\end{lemma}
Lemma \ref{smooth_theorem1} is analogous to Theorem 3 of \cite{nason2000wavelet}. The result of mean square consistency follows due to a combination of Equation \eqref{beans} and Theorem 4 of \cite{nason2000wavelet}. The mean squared error of the smoothed, corrected wavelet periodogram is given by
\begin{align*}
 & \mathbb{E} \left[  \int_{0}^{1}  \left(\widehat{S}_{j} (z) - S_{j} (z) \right)^{2} dz \right]  = 2^{-j}  \mathbb{E} \left[  \int_{0}^{1}  \left(\hat{\tilde{S}}_{j} (z) - \tilde{S}_{j} (z) \right)^{2} dz \right] \\[1ex]
 & \leq  2^{-j+1} \mathbb{E} \left[  \int_{0}^{1} \left( \sum_{l=-J}^{-1} \left(\hat{I}^{l}_{\lfloor zT \rfloor} - \Lambda_{l} (z) \right) P_{jl}^{-1} \right)^{2} dz \right]  + 2^{-j+1} \int^{1}_{0} \left( \sum_{l < -J} \Lambda_{l} (z) P_{jl}^{-1} \right)^{2} dz,
 \end{align*}
 where $\hat{I}^{l}_{\lfloor zT \rfloor}$ is the smoothed estimate of the rescaled raw wavelet periodogram and $\Lambda_{l} (z) = \sum_{n} P_{nl} \tilde{S}_{n} (z)$. The first term can be bounded as
 \begin{align*}
 I & \leq 2^{-j+1}  \left( \sum_{l=-J}^{-1} P_{jl}^{-1} \left( \mathbb{E} \left[ \int_{0}^{1} \left(\hat{I}^{l}_{\lfloor zT \rfloor} - \Lambda_{l} (z) \right)^{2} dz \right] \right)^{1/2} \right)^{2} \\
& \leq 2^{-j+1}  \left( \sum_{l=-J}^{-1} P_{jl}^{-1}  \mathcal{O} \left( T^{-2/3} \log^{2} (T)  \right)^{1/2} \right)^{2} =  \mathcal{O} \left( 2^{-j} T^{-2/3} \log^{2} (T)\right) ,
\end{align*}
which follows since $P$ possesses a bounded inverse with exponentially decaying entries, and using the rate of convergence of the mean squared error of the wavelet thresholding estimator derived in \cite{neumann1995wavelet}, Theorem 3.1 A. The second term is asymptotically dominated by the first, since it can be bounded as
\begin{align*}
II & \leq2^{-j+1} \left(  \sum_{l< -J} \Lambda_{l} (z) P_{jl}^{-1} \right)^{2} \leq 2^{-j+1} \left(  \sum_{l< -J} P_{jl}^{-1} \sum_{n=-\infty}^{-1} 2^{-l/2} D^{1}_{ln} S_{n} (z) \right)^{2} \\[1ex]
& \leq2^{-j+1} \left(  \sum_{l< -J} \mathcal{O} (2^{l/2}) \sum_{n=-\infty}^{-1}  S_{n} (z) \right)^{2} = \mathcal{O} (2^{-j} \times 2^{-J}) = \mathcal{O} ( 2^{-j} T^{-1}) , 
\end{align*}
which follows since $P^{-1}_{jl}$ is bounded, $D^{1}_{ln} = \mathcal{O}(2^{l})$, $\sum_{n} S_{n} (z) < \infty$, and $T = 2^{J}$. Hence, the mean squared error is given by $\mathcal{O} \left( 2^{-j} T^{-2/3} \log^{2} (T)\right)$.

\subsection{Proof of Proposition \protect\ref{A1-c-est}} 
We can write the mean squared error as
\begin{equation*}
\mathbb{E} \left[  \int_{0}^{1} \left( \hat{c} (z,\tau) - c(z,\tau) \right)^{2} dz \right] \leq 2 \mathbb{E} \left[  \int_{0}^{1} \left( \sum_{j=-J_{0}}^{-1} \left(\widehat{S}_{j} (z) - S_{j} (z) \right) \Psi_{j} (\tau) \right)^{2} dz \right]  + 2R_{J_{0}},
\end{equation*}
where $R_{J_{0}}$ can be bounded as
\begin{align*}
R_{J_{0}} &=  \left( \sum_{j < - J_{0}} S_{j} (z) \Psi_{j} (\tau) \right)^{2} \leq  \left( \sum_{j < - J_{0}} S_{j} (z)  \right)^{2} \\[1ex]
& \leq  \left( \sum_{j < - J_{0}}  \mathcal{O} (2^{\gamma j} ) \right)^{2} = \mathcal{O}(2^{-2\gamma J_{0}}) = \mathcal{O} (T^{-2 \gamma \alpha}) = \mathcal{O} (T^{\alpha - 2/3}).
\end{align*}
For the first term, we obtain
\begin{align*}
&\mathbb{E} \left[  \int_{0}^{1} \left( \sum_{j=-J_{0}}^{-1} \left(\widehat{S}_{j} (z) - S_{j} (z) \right) \Psi_{j} (\tau) \right)^{2} dz \right] \\
%& \leq \sum_{j} \sum_{l} \mathbb{E} \left[ \left( \int_{0}^{1} \left(\hat{S}_{j} (z) - S_{j} (z) \right)^{2} dz \right)^{1/2}   \left( \int_{0}^{1} \left(\hat{S}_{l} (z) - S_{l} (z) \right)^{2} dz \right)^{1/2}      \Psi_{j} (\tau) \Psi_{l} (\tau) \right] \\
&\leq \left( \sum_{j=-J_{0}}^{-1} \Psi_{j} ( \tau ) \left( \mathbb{E} \left[ \int_{0}^{1}  \left(\widehat{S}_{j} (z) - S_{j} (z) \right)^{2} dz \right] \right)^{1/2} \right)^{2}\\
&\leq \left( \sum_{j=-J_{0}}^{-1}  \mathcal{O} \left( 2^{-j} T^{-2/3} \log^{2} (T) \right)^{1/2} \right)^{2}\\
&= \mathcal{O} \left( \left(T^{\alpha-2/3} \log^{2} (T) \right)^{1/2} \right)^{2} = \mathcal{O} \left(T^{\alpha-2/3} \log^{2} (T) \right),
\end{align*}
by Equation (6) from the main text, and using that $\Psi_{j} (\tau) \leq 1$ for all $j$ and $\tau$. Hence, provided that $T^{\alpha-2/3} \log^{2} (T) \rightarrow 0$ as $T \rightarrow \infty$, the estimator is mean square consistent.

\subsection{Proof of Proposition \protect\ref{nth-diff-exp}} 
In order to derive the formula for the expectation of the squared wavelet coefficients of a general $n$-th difference, we require the formula the $n$-th difference itself, which is a well-known result. Denote by $\nabla^{n} {X}_{t}$ the n-th difference of the time series $\{X_{t}\}$ at time $t$. Then,
\begin{equation}\label{difflemma}
\nabla^{n} X_{t} = \sum_{k=0}^{n} (-1)^{k} {n \choose k} X_{t-k} = \sum_{k=0}^{n}  (-1)^{k} {n \choose k} \epsilon_{t-k}   +\mathcal{O}(T^{-1}) ,
\end{equation}
which follows from the differentiability assumption of the trend $\mu$. Now, observe that, for an $n$-th difference, the expectation will involve the sum of the spectrum over all scales, multiplied by a linear combination of lagged inner product $A$-matrix entries, denoted $A^{\tau}$, from lag $0$ to lag $n$. To calculate the coefficient in front of each of the $A^{\tau}$, we simply calculate the sum of the coefficients of the $\{\epsilon_{t} \epsilon_{s} \}$ for each particular lag in the expansion obtained by squaring the $n$-th difference of the time series, for which we can use Equation \eqref{difflemma}.

For example, to calculate the coefficient of $A$, we add together the coefficients of the squared terms in the square of the differenced series, i.e. add the coefficients of $\epsilon_{t}^{2}, \epsilon_{t-1}^{2}, \ldots \epsilon_{t-n}^{2}$. For the coefficient of $A^{1}$, we add together the coefficients of the terms in the square of the difference that differ in index by 1, i.e. we add the coefficients of $\epsilon_{t}\epsilon_{t-1}, \epsilon_{t-1} \epsilon_{t-2}, \epsilon_{t-1} \epsilon_{t}, \ldots , \epsilon_{n-1} \epsilon_{n}$. In particular, if we were interested in the third difference, then the coefficient in front of $A$ would be $1+9+9+1$, and the coefficient in front of $A^{1}$ would be $-3-3-9-9-3-3 = -30$, and so on. Hence the formula for the coefficient in front of $A$ is given by 
\begin{equation*}
\sum_{r=0}^{n} {n \choose r}^{2} = {2n \choose n},
\end{equation*}
and similarly, the formula for the coefficient in front of $A^{\tau}$, for $\tau \geq 1$, is given by 
\begin{equation*}
2 (-1)^{\tau} \sum_{r=0}^{n-\tau} {n \choose r} {n \choose r + \tau} = 2 (-1)^{\tau} {2n \choose n + \tau },
\end{equation*}
where the multiplication by 2 arises due to symmetry (for example we must add both the coefficients of $\epsilon_{t-1} \epsilon_{t-2}$ and $\epsilon_{t-2} \epsilon_{t-1}$). The equality on the right follows from a counting argument. The number of ways to choose $n+\tau$ objects from $2n$ choices is the same as the number of ways of choosing $r$ objects from the first $n$ and choosing $r+\tau$ from the remaining $n$ objects, for $0 \leq r \leq n-\tau$. The expectation is accurate up to order $\mathcal{O} (T^{-1})$ by the same argument as in the proof of Proposition \ref{A1exp}. Hence, the squared expectation of the wavelet coefficients of the $n$-th differenced series are given by 
\begin{equation*}
\mathbb{E} (\tilde{I}_{k}^{j}) = \sum_{l} S_{l} \left( \frac{k}{T} \right) \left[ {2n \choose n} A_{jl} + 2 \sum_{\tau=1}^{n} (-1)^{\tau} {2n \choose n+\tau}A_{jl}^{\tau}  \right]  + \mathcal{O}(T^{-1}) .
\end{equation*}

\subsection{Proof of Theorem \protect\ref{2nd-diff-spec}}

By the Daubechies characterisation of H{\"o}lder spaces, rescaling of the LSW process, and Proposition \ref{A1exp}, the expectation of the rescaled raw wavelet periodogram of the differenced time series is given by
\begin{equation}\label{a7eq}
\mathbb{E} (\tilde{I}_{k}^{j}) = \sum_{l} P_{jl} \tilde{S}_{l} (k/T) + \mathcal{O}(2^{-7j/2} T^{-2}) +  \mathcal{O} (T^{-1}) .
\end{equation}
Hence, the mean squared error of the smoothed wavelet periodogram is given by
\begin{align*}
 & \mathbb{E} \left[  \int_{0}^{1}  \left(\widehat{S}_{j} (z) - S_{j} (z) \right)^{2} dz \right]  = 2^{-j}  \mathbb{E} \left[  \int_{0}^{1}  \left(\hat{\tilde{S}}_{j} (z) - \tilde{S}_{j} (z) \right)^{2} dz \right] \\[1ex]
 & \leq  2^{-j+1} \mathbb{E} \left[  \int_{0}^{1} \left( \sum_{l=-J_{1}}^{-1} \left(\hat{I}^{l}_{\lfloor zT\rfloor} - \Lambda_{l} (z) \right) P_{jl}^{-1} \right)^{2} dz \right]  + 2^{-j+1} \int^{1}_{0} \left( \sum_{l < -J_{1}} \Lambda_{l} (z) P_{jl}^{-1} \right)^{2} dz
 \end{align*}
 where $\hat{I}^{l}_{\lfloor zT\rfloor }$ is the smoothed estimate of the raw wavelet periodogram and $\Lambda_{l} (z) = \sum_{n} P_{nl} \tilde{S}_{n} (z)$. The first term can be bounded as
 \begin{align*}
 I & \leq 2^{-j+1}  \left( \sum_{l=-J_{1}}^{-1} P_{jl}^{-1} \left( \mathbb{E} \left[ \int_{0}^{1} \left(\hat{I}^{l}_{\lfloor zT \rfloor} - \Lambda_{l} (z) \right)^{2} dz \right] \right)^{1/2} \right)^{2} \\
& \leq 2^{-j+1}  \left( \sum_{l=-J_{1}}^{-1} P_{jl}^{-1} \left( \mathcal{O} \left( 2^{-7l} T^{-4}  \right) +  \mathcal{O} \left( T^{-2/3} \log^{2} (T)  \right)\right)^{1/2} \right)^{2} \\
& = \mathcal{O} \left(2^{-j} T^{7\beta-4} \right)  + \mathcal{O} \left( 2^{-j} T^{-2/3} \log^{2} (T)\right),
\end{align*}
which follows since $P$ possesses a bounded inverse with exponentially decaying entries, and using Equation \eqref{a7eq}. The second term can be bounded in the same fashion as in the proof of Theorem \ref{smooth_theorem4}, and is of order $\mathcal{O}(2^{-j} T^{-\beta})$, which gives the stated consistency result. Similarly, the mean squared error of the LACV estimator is calculated as
\begin{equation*}
\mathbb{E} \left[  \int_{0}^{1} \left( \hat{c} (z,\tau) - c(z,\tau) \right)^{2} dz \right] \leq 2 \mathbb{E} \left[  \int_{0}^{1} \left( \sum_{j=-J_{0}}^{-1} \left(\widehat{S}_{j} (z) - S_{j} (z) \right) \Psi_{j} (\tau) \right)^{2} dz \right]  + 2R_{J_{0}},
\end{equation*}
where $R_{J_{0}}$ is asymptotically negligible by the argument in Proposition \ref{A1-c-est}. For the first term, we obtain
\begin{align*}
& \mathbb{E} \left[  \int_{0}^{1} \left( \sum_{j=-J_{0}}^{-1} \left(\hat{S}_{j} (z) - S_{j} (z) \right) \Psi_{j} (\tau) \right)^{2} dz \right] \\
%& \leq \left( \sum_{j=-J_{1}}^{-1} \Psi_{j} ( \tau ) \left( \mathbb{E} \left[ \int_{0}^{1}  \left(\hat{S}_{j} (z) - S_{j} (z) \right)^{2} dz \right] \right)^{1/2} \right)^{2} \\[1ex]
&\leq \left( \sum_{j=-J_{0}}^{-1} 2^{-j/2} \left(  \mathcal{O} \left(T^{7\beta-4} \right) + \mathcal{O} \left( T^{-2/3} \log^{2} (T) \right) + \mathcal{O} (T^{-\beta}) \right) ^{1/2} \right)^{2}\\[1ex]
&= \left(  \mathcal{O} \left(T^{7\beta-4} \right) + \mathcal{O} \left( T^{-2/3} \log^{2} (T) \right) + \mathcal{O} (T^{- \beta}) \right)  \left( \sum_{j=-J_{0}}^{-1} 2^{-j/2} \right)^{2}\\[1ex]
&=  \mathcal{O} \left(T^{\alpha + 7\beta-4} \right) + \mathcal{O} \left( T^{\alpha-2/3} \log^{2} (T) \right) + \mathcal{O} (T^{\alpha - \beta}) .
\end{align*}

\subsection{Proof of Proposition \protect\ref{smooth_theorem5}} 

The appropriate threshold is derived by using an analogous result to \ref{smooth_theorem1}, from which we obtain
\begin{equation*}
\mathbb{E} ( \hat{v}_{rs} ) - \int_0^1 \mu(z) {\psi}'_{rs} (z) dz = \mathcal{O} \left( 2^{r/2} T^{-1} \right),
\end{equation*}
and
\begin{equation*}
\text{Var} ( \hat{v}_{rs} ) =  T^{-1} \int_0^1 \sum_n {S}_n (z) {\psi}'^2_{rs} (z) dz + \mathcal{O} \left( 2^{r} T^{-2} \right).
\end{equation*}
The mean squared error rate is obtained by using Theorem 1 of \cite{von2000non}, with the specific case of a Lipschitz continuous trend.

\end{appendices}

\end{document}